\documentclass[aps,prx,onecolumn,superscriptaddress,10pt]{revtex4-2}
\usepackage{silence}
\usepackage{amsthm}
\usepackage{amssymb}
\usepackage{amsmath}
\usepackage{amsfonts}
\usepackage{physics}
\usepackage{graphicx}
\usepackage{dcolumn}
\usepackage{bm}
\usepackage[utf8]{inputenc}
\usepackage[english]{babel}
\usepackage[table]{xcolor}
\usepackage[T1]{fontenc}
\usepackage[colorlinks=true,citecolor=teal,linkcolor=teal]{hyperref}
\hypersetup{allcolors={teal}}
\usepackage{tikz}
\usepackage{epsfig}
\usepackage{epstopdf}
\usepackage{mathtools}
\usepackage{biolinum}
\usepackage{braket}
\usepackage{hyperref}
\usepackage{amsthm}
\usepackage{enumitem}
\usepackage{dsfont}
\usepackage{bbold}
\usepackage{eurosym}
\usepackage{diagbox}
\usepackage{comment}

\let\coloneqq\relax

\usepackage{amsthm}
\usepackage{amssymb}
\usepackage{amsmath}
\usepackage{bbold}
\usepackage{bbm}
\usepackage{xr-hyper}
\usepackage[capitalise]{cleveref}
\usepackage{braket}
\usepackage{dsfont}
\usepackage{mathdots}
\usepackage{mathtools}
\usepackage{enumerate}
\usepackage{amsfonts}
\usepackage{mathrsfs}
\usepackage{csquotes}
\usepackage[cal=boondox]{mathalfa}
\usepackage{graphicx}
\usepackage{stackengine}
\usepackage{scalerel}
\usepackage{tensor}
\usepackage{array}
\usepackage{makecell}
\newcolumntype{x}[1]{>{\centering\arraybackslash}p{#1}}
\usepackage{tikz}
\usetikzlibrary{shapes.geometric, shapes.misc, positioning, arrows, arrows.meta, decorations.pathreplacing, decorations.pathmorphing, patterns, angles, quotes, calc}
\usepackage{booktabs}
\usepackage{xfrac}
\usepackage{centernot}
\usepackage{comment}
\usepackage{chngcntr}
\usepackage{titlesec}

\newtheorem*{thm*}{Theorem}

\newtheorem*{prop*}{Proposition}

\newtheorem*{lemma*}{Lemma}

\newtheorem*{cor*}{Corollary}

\newtheorem*{cj*}{Conjecture}

\newtheorem*{Def*}{Definition}

\makeatletter
\def\thmhead@plain#1#2#3{%
  \thmname{#1}\thmnumber{\@ifnotempty{#1}{ }\@upn{#2}}%
  \thmnote{ {\the\thm@notefont#3}}}
\let\thmhead\thmhead@plain
\makeatother

\theoremstyle{definition}

\newcommand{\be}{\begin{equation}\begin{aligned}\hspace{0pt}}
\newcommand{\bee}{\begin{equation*}\begin{aligned}}
\newcommand{\ee}{\end{aligned}\end{equation}}
\newcommand{\eee}{\end{aligned}\end{equation*}}
\newcommand*{\coloneqq}{\mathrel{\vcenter{\baselineskip0.5ex \lineskiplimit0pt \hbox{\scriptsize.}\hbox{\scriptsize.}}} =}

\newcommand\ceil[1]{\left\lceil#1\right\rceil}

\newcommand{\eqt}[1]{\stackrel{\mathclap{ \mbox{\scriptsize #1}}}{=}}
\newcommand{\leqt}[1]{\stackrel{\mathclap{ \mbox{\scriptsize #1}}}{\leq}}

\newcommand{\geqt}[1]{\stackrel{\mathclap{ \mbox{\scriptsize #1}}}{\geq}}

\newcommand{\tcb}[1]{{\color{blue} #1}}

\newcommand{\R}{\mathds{R}}
\newcommand{\N}{\mathds{N}}
\newcommand{\C}{\mathds{C}}
\newcommand{\E}{\mathds{E}}

\DeclareMathAlphabet{\pazocal}{OMS}{zplm}{m}{n}

\DeclareMathOperator{\diag}{diag}

\newcommand{\lsmatrix}{\left(\begin{smallmatrix}}
\newcommand{\rsmatrix}{\end{smallmatrix}\right)}

\stackMath

\stackMath

\makeatletter
\newcommand*\rel@kern[1]{\kern#1\dimexpr\macc@kerna}
\newcommand*\widebar[1]{%
  \begingroup
  \def\mathaccent##1##2{%
    \rel@kern{0.8}%
    \overline{\rel@kern{-0.8}\macc@nucleus\rel@kern{0.2}}%
    \rel@kern{-0.2}%
  }%
  \macc@depth\@ne
  \let\math@bgroup\@empty \let\math@egroup\macc@set@skewchar
  \mathsurround\z@ \frozen@everymath{\mathgroup\macc@group\relax}%
  \macc@set@skewchar\relax
  \let\mathaccentV\macc@nested@a
  \macc@nested@a\relax111{#1}%
  \endgroup
}

\counterwithin*{equation}{part}
\counterwithin*{thm}{part}
\counterwithin*{figure}{part}

\tikzset{meter/.append style={draw, inner sep=10, rectangle, font=\vphantom{A}, minimum width=30, line width=.8, path picture={\draw[black] ([shift={(.1,.3)}]path picture bounding box.south west) to[bend left=50] ([shift={(-.1,.3)}]path picture bounding box.south east);\draw[black,-latex] ([shift={(0,.1)}]path picture bounding box.south) -- ([shift={(.3,-.1)}]path picture bounding box.north);}}}
\tikzset{roundnode/.append style={circle, draw=black, fill=gray!20, thick, minimum size=10mm}}
\tikzset{squarenode/.style={rectangle, draw=black, fill=none, thick, minimum size=10mm}}

\definecolor{Blues5seq1}{RGB}{239,243,255}
\definecolor{Blues5seq2}{RGB}{189,215,231}
\definecolor{Blues5seq3}{RGB}{107,174,214}
\definecolor{Blues5seq4}{RGB}{49,130,189}
\definecolor{Blues5seq5}{RGB}{8,81,156}

\definecolor{Greens5seq1}{RGB}{237,248,233}
\definecolor{Greens5seq2}{RGB}{186,228,179}
\definecolor{Greens5seq3}{RGB}{116,196,118}
\definecolor{Greens5seq4}{RGB}{49,163,84}
\definecolor{Greens5seq5}{RGB}{0,109,44}

\definecolor{Reds5seq1}{RGB}{254,229,217}
\definecolor{Reds5seq2}{RGB}{252,174,145}
\definecolor{Reds5seq3}{RGB}{251,106,74}
\definecolor{Reds5seq4}{RGB}{222,45,38}
\definecolor{Reds5seq5}{RGB}{165,15,21}

\usepackage{pgfplots}
\pgfplotsset{compat=1.18} 

\allowdisplaybreaks



\newcommand{\ba}{\begin{aligned}}
\newcommand{\ea}{\end{aligned}}

\newcommand{\bc}{\begin{center}}
\newcommand{\ec}{\end{center}}
\newcommand{\beq}{\begin{equation}}
\newcommand{\eeq}{\end{equation}}
\newcommand{\beqq}{\begin{equation*}}
\newcommand{\eeqq}{\end{equation*}}
\newcommand{\beqa}{\begin{align}}
\newcommand{\eeqa}{\end{align}}
\newcommand{\barr}{\begin{array}}
\newcommand{\earr}{\end{array}}
\newcommand{\bi}{\begin{itemize}}
\newcommand{\ei}{\end{itemize}}

\theoremstyle{plain}
\newcounter{tho}
\newtheorem{theo}[tho]{Theorem}
\newtheorem{lem}[tho]{Lemma}

\newtheorem{coro}[tho]{Corollary}
\newtheorem{defi}[tho]{Definition}

\newtheorem{ex}[tho]{Example}
\newtheorem*{defi*}{Definition}
\newtheorem*{coro*}{Corollary}
\newtheorem*{conj*}{Conjecture}
\newtheorem*{theo*}{Theorem}
\newtheorem*{lem*}{Lemma}
\DeclareMathOperator{\poly}{poly\,}

\newcommand{\parhead}[1]{\noindent \textbf{\textsf{#1}}}

\renewcommand{\qed}{\hfill\ensuremath{\square}}
\allowdisplaybreaks

\begin{document}

\title{The symplectic rank of non-Gaussian quantum states} 

\author{Francesco A.\ Mele}\email{francesco.mele@sns.it}
\affiliation{NEST, Scuola Normale Superiore and Istituto Nanoscienze, Piazza dei Cavalieri 7, IT-56126 Pisa, Italy}

\author{Salvatore F.\ E.\ Oliviero}\email{salvatore.oliviero@sns.it}
\affiliation{NEST, Scuola Normale Superiore and Istituto Nanoscienze, Piazza dei Cavalieri 7, IT-56126 Pisa, Italy}

\author{Varun Upreti}
\email{varun.upreti@inria.fr}
\affiliation{DIENS, \'Ecole Normale Sup\'erieure, PSL University, CNRS, INRIA, 45 rue d’Ulm, Paris, 75005, France}

\author{Ulysse Chabaud}
\email{ulysse.chabaud@inria.fr}
\affiliation{DIENS, \'Ecole Normale Sup\'erieure, PSL University, CNRS, INRIA, 45 rue d’Ulm, Paris, 75005, France}


\begin{abstract}
Non-Gaussianity is a key resource for achieving quantum advantages in bosonic platforms. Here, we investigate the \textit{symplectic rank}: a novel non-Gaussianity monotone that satisfies remarkable operational and resource-theoretic properties. Mathematically, the symplectic rank of a pure state is the number of symplectic eigenvalues of the covariance matrix that are strictly larger than the ones of the vacuum. Operationally, it (i) is easy to compute, (ii) emerges as the smallest number of modes onto which all the non-Gaussianity can be compressed via Gaussian unitaries, (iii) lower bounds the non-Gaussian gate complexity of state preparation independently of the gate set, (iv) governs the sample complexity of quantum tomography, and (v) bounds the computational complexity of bosonic circuits. Crucially, the symplectic rank is non-increasing under post-selected Gaussian operations, leading to strictly stronger no-go theorems for Gaussian conversion than those previously known. Remarkably, this allows us to show that the resource theory of non-Gaussianity is irreversible under exact Gaussian operations. Finally, we show that the symplectic rank is a robust non-Gaussian measure, explaining how to witness it in experiments and how to exploit it to meaningfully benchmark different bosonic platforms. In doing so, we derive lower bounds on the trace distance (resp.~total variation distance) between arbitrary states (resp.~classical probability distributions) in terms of the norm distance between their covariance matrices, which may be of independent interest.
\end{abstract}

\maketitle


\tableofcontents


\section{Introduction} Quantum computers are expected to provide considerable advantages over their classical counterparts \cite{deutsch1992rapid,Simon1997,Shor}. This promise has sparked a technological race to build the first universal quantum computer, involving several candidate physical platforms such as neutral atoms \cite{LogicalLukin}, quantum dots \cite{xue2021cmos}, ion traps \cite{pogorelov2021compact}, nitrogen vacancy centers \cite{childress2013diamond}, photonics \cite{kok2007linear}, and superconducting circuits \cite{acharya2024quantum}, among others \cite{horowitz2019quantum}. To date, it remains to be seen which architecture will come out ahead, and it is thus crucial to develop tools to compare quantum computers which are based on radically different physical platforms \cite{eisert_quantum_2020,proctor2025benchmarking}.
Among these platforms, bosonic ones have emerged in particular as promising physical systems for realizing quantum computers, both from theoretical \cite{knill2001scheme,takeda2019toward,bourassa2021blueprint,bartolucci2023fusion} and experimental standpoints: for instance, photonic processors feature deterministic generation of large-scale entangled states \cite{yokoyama2013ultra} and long coherence times \cite{fabre2020modes}, while bosonic modes of superconducting cavities provide flexible quantum state engineering \cite{eriksson2024universal} as well as remarkable error-correction capabilities \cite{sivak2023real}.

Yet, not all bosonic quantum computations are equally powerful: Gaussian computations --- a well-studied subclass of bosonic computations --- can be simulated efficiently on a classical computer \cite{bartlett2002efficient}, thus identifying non-Gaussian quantum states and operations as necessary ingredients for universal quantum computations \cite{Lloyd1999}. As it turns out, non-Gaussian resources are also instrumental for a variety of other quantum information processing tasks \cite{wenger2003maximal,garcia2004proposal,adesso2009optimal}, including entanglement distillation \cite{nogo1,nogo2,nogo3}, quantum error-correction \cite{Niset2009}, quantum metrology \cite{frigerio2025}, quantum communication~\cite{Mele_2025,Die-Hard-2-PRL,Die-Hard-2-PRA}, and learning of quantum states~\cite{mele2024learning,fawzi2024optimalfidelityestimationbinary}. 
The importance of non-Gaussian resources for quantum information processing has prompted the introduction of several measures aiming to capture the non-classical capabilities of quantum devices based on their non-Gaussianity \cite{Albarelli2018,zhuang2018resource,Takagi2018}. In particular, in the context of quantum computing, both the negative volume of the Wigner function \cite{Kenfack2004} and the stellar rank \cite{chabaud2020stellar,chabaud2021holomorphic} have been shown to provide non-Gaussian measures of the cost of simulating quantum computations on a classical computer \cite{Mari2014,pashayan2015estimating,chabaud2023resources}.

In principle, such resource measures~\cite{Albarelli2018,takagi2018convex} could be used to compare the non-classical capabilities of different bosonic quantum computers. However, these non-Gaussianity measures are either: (i) sensitive to a specific type of non-Gaussianity, or (ii) not known to possess an operational interpretation. 
Indeed, the stellar rank measures the complexity of engineering a quantum state when using non-Gaussian photon additions or subtractions \cite{chabaud2020stellar,chabaud2021holomorphic}, which are natural non-Gaussian operations to consider for photonics platform \cite{marco2010manipulating}; on the other hand, the negative volume of the Wigner function can be used independently of the type of non-Gaussianity, but is not known to possess a similar operational meaning based on state complexity. 
As such, it is currently highly challenging to use these resource measures to meaningfully compare quantum devices based on different physical bosonic platforms. This is an instance of a formidable challenge posed by the development of quantum technologies: how to assess and compare the non-classical capabilities of bosonic quantum computers that are based on fundamentally different physical platforms and native operations?

We propose to address this challenge by using the \textit{symplectic rank}. The symplectic rank has originally been introduced for Gaussian quantum states in Ref.~\cite{adesso2006generic}, where it is defined as the number of symplectic eigenvalues of the covariance matrix that are strictly larger than those of the vacuum. We extend its definition to arbitrary non-Gaussian quantum states and show that it provides a powerful measure for characterizing bosonic devices.

We prove that the symplectic rank of a quantum state has an operational interpretation based on the concept of non-Gaussian compression \cite{mele2024learning}, and lower bounds the non-Gaussian gate complexity of a bosonic quantum state, independently of the type of non-Gaussianity required to engineer it. Moreover, we prove that bosonic quantum computations with low symplectic rank can be simulated efficiently by classical computers, making it a natural measure for benchmarking quantum computers that are based on different bosonic platforms. We also show that the symplectic rank governs the sample complexity of quantum tomography of bosonic states.

Beyond its operational interpretations for quantum information processing tasks, we show that the symplectic rank satisfies a number of desirable resource-theoretic properties, among which: it has a natural definition based on the symplectic eigenvalues of the covariance matrix, making it easy to compute and to witness in experiments; it is a non-Gaussianity monotone, and can thus be used to bound the efficiency of Gaussian protocols. For example, this property allows us to prove that no Gaussian operation can deterministically convert a single-mode non-Gaussian pure state into two non-Gaussian pure states, establishing that spreading non-Gaussianity via Gaussian operations is impossible. Furthermore, we use it to prove the irreversibility of the resource theory of non-Gaussianity for exact conversion. 

Moreover, we show that the symplectic rank is robust,i.e.~it cannot decrease under small perturbations. In doing so, we introduce a novel perturbation bound on the covariance matrix of general (possibly non-Gaussian) states, which may be of independent interest. Specifically, our bound addresses the seemingly simple yet fundamental question: if two states are \(\varepsilon\)-close (e.g.~in trace distance~\cite{NC}), how far apart are their covariance matrices? This advances the rapidly growing literature on perturbation bounds for continuous-variable quantum states~\cite{mele2024learning,holevo2024estimatestracenormdistancequantum,holevo2024estimatesburesdistancebosonic,bittel2024optimalestimatestracedistance,fanizza2024efficienthamiltonianstructuretrace}, where such bounds had previously been derived only for the specific case of Gaussian states.

Finally, we extend the symplectic rank to the approximate setting. We show that the resulting approximate symplectic rank is a monotone under Gaussian operations, that it can be computed numerically via optimizations over Gaussian unitaries, and that it provides a robust and operational measure on non-Gaussianity, agnostic to the type of non-Gaussianity.



\subsection{Overview}

This section provides a non-technical summary of our key findings regarding the symplectic rank, while formal statements, proofs, and additional results are deferred to Sections~\ref{app:symprank}-\ref{Sec_classical_bound}. We refer to \cref{sec_prel} for background material and choice of conventions.

\medskip
 
\parhead{The symplectic rank.}
As mentioned in the introduction, we extend the notion of \emph{symplectic rank}~\cite{adesso2006generic} to non-Gaussian states. Specifically, we demonstrate that the symplectic rank provides a measure of non-Gaussianity with remarkable properties. Let us begin by defining the symplectic rank~\cite{adesso2006generic}.

\begin{defi}[(Symplectic rank of pure states)]\label{def:symp-rank_pure}
Let \(\psi\) be a pure state with a well-defined covariance matrix $V(\psi)$. The symplectic rank of \(\psi\), denoted as \(\mathfrak{s}(\psi)\), is the number of symplectic eigenvalues of \(V\!(\psi)\) that are strictly greater than the ones of the vacuum.
\end{defi}

\noindent This definition is motivated by the fact that a pure state $\psi$ is Gaussian if and only if all its symplectic eigenvalues are exactly equal to the ones of the vacuum,  which is equivalent to having a vanishing symplectic rank. Throughout the rest of this paper, we use the convention $\hbar=1$ and define the covariance matrix such that the symplectic eigenvalues of the vacuum are all equal to one.

Notably, the symplectic rank has an operational interpretation in terms of its relation with the concept of \emph{$t$-compressibility}, originally introduced in~\cite{oliviero_unscrambling_2024,leone_learning_2024} and later extended in~\cite{leone2023learning, mele2024efficient, mele2024learning}.  By definition, a state is $t$-compressible if a Gaussian unitary operation can transform it into the tensor product of a $t$-mode state and the vacuum state on all other modes. We obtain the following result: 
\begin{theo}[(Structure theorem, \textit{informal version of}~\cref{charact_symp})]\label{thm_comp}
The symplectic rank $\mathfrak{s}(\psi)$ of a pure state $\psi$ is equal to the smallest $t$ for which $\psi$ is $t$-compressible. Specifically, there exists a Gaussian unitary operation ${G}$ and a $\mathfrak{s}(\psi)$-mode state ${\phi}$ such that
\be
    \ket{\psi}={G}\left( \ket{\phi}\otimes \ket{0}^{\otimes ( n-\mathfrak{s}(\psi))}\right)\,,
\ee
where $n$ denotes the total number of modes. Moreover, the compressed state $\phi$ has mean photon number upper bounded by the one of $\psi$.
\end{theo}
\begin{figure}
    \centering
    \includegraphics[width=0.3\linewidth]{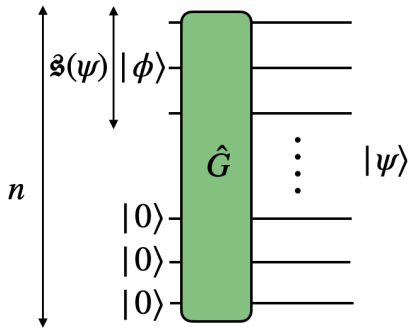}
    \caption{Operational interpretation of the symplectic rank (\cref{thm_comp}). Given a pure state $\psi$ over $n$ modes, its symplectic rank $\mathfrak s(\psi)$ is the minimum number of modes $t$ for which ${\psi}$ can be written as $\ket \psi =  G (\ket \phi \otimes \ket 0^{\otimes(n - t)})$, where ${\phi}$ is a non-Gaussian state over $t$ modes and $G$ is a Gaussian unitary.}
    \label{fig:structure_thm}
\end{figure}

\noindent\emph{Sketch of the proof}. The first part of the result follows from the techniques in~\cite{mele2024learning}, and we recall the general idea behind the proof hereafter. To begin, consider a pure state $\psi$ with a given symplectic rank $\mathfrak{s}(\psi)$. By the definition of symplectic eigenvalues, there exists a symplectic transformation, i.e.~a Gaussian unitary operation that maps the covariance matrix $V(\psi)$ into a diagonal matrix. The entries of this matrix are the symplectic eigenvalues of $V(\psi)$. Notably, $\psi$ has $n - \mathfrak{s}(\psi)$ symplectic eigenvalues equal to one. The only state with such eigenvalues is the vacuum state. From there, it is easy to understand the general idea of the first part of the proof. The part of the proof regarding the mean photon number is more technical and can be found in~\cref{charact_symp}.\qed
\smallskip

This structure theorem gives an operational meaning to the symplectic rank: it represents the smallest number of modes onto which all the non-Gaussianity of the state can be compressed by acting with a Gaussian unitary operation (see \cref{fig:structure_thm}). Moreover, \cref{thm_comp} allows us to extend the definition of symplectic rank to all pure states, even when their covariance matrix is not well-defined: the symplectic rank \(\mathfrak{s}(\psi)\) can be defined as the smallest \(t\) for which \(\psi\) is \(t\)-compressible. This definition can be further generalized to arbitrary mixed states through a convex roof construction, a widely used technique in the quantum resource theory literature~\cite{Chitambar_2019}. Note that the following generalization to mixed states is different from the one employed for Gaussian states in Ref.~\cite{adesso2006generic}.

\begin{defi}[(Symplectic rank of mixed states)]\label{def:symp-rank_mix}
The symplectic rank of a mixed state $\rho$, denoted as $\mathfrak{s}(\rho)$, is defined as
\be
    \mathfrak{s}(\rho)\coloneqq \min_{\substack{\rho=\sum_{i}p_i\psi_i\\p_i>0}}\max_{i}\mathfrak{s}(\psi_i)\,,
\ee
where the minimum is taken over all the possible decompositions of $\rho$ as a convex combination of pure states.
\end{defi}
The symplectic rank satisfies the following key properties:  \textit{(i)} it vanishes if and only if the state is a convex combination of Gaussian states (as it can be easily proven by exploiting that any Gaussian state is a convex combination of pure Gaussian states~\cite{BUCCO}); \textit{(ii)} it is additive under the tensor product of pure states and sub-additive under the tensor product of mixed states (as it easily follows from the definition of symplectic rank for pure states); and, most importantly, (iii) it is a \emph{quantum non-Gaussianity monotone} under post-selected Gaussian operations (i.e.~tensoring with a Gaussian state, applying a Gaussian unitary operation, performing classical mixing, performing a heterodyne measurement with the ability to post-select on the outcome, and taking a partial trace~\cite{nogo3}):

\begin{theo}[(Monotonicity of the symplectic rank, \textit{informal version of}~\cref{thm_main_sm})]\label{thm_monotonicity}
    The symplectic rank is a quantum non-Gaussianity monotone. That is, for all states $\rho$ and all post-selected Gaussian operations $\mathcal{G}$, the symplectic rank of $\mathcal{G}(\rho)$ is not larger than the one of $\rho$:
   \be
        \mathfrak{s}\!\left(\mathcal{G}(\rho)\right)\le \mathfrak{s}\!\left(\rho\right)\,.
    \ee
\end{theo}



\parhead{Probabilistic Gaussian conversion.}
The monotonicity of the symplectic rank establishes a powerful tool for proving no-go theorems in Gaussian protocols for non-Gaussian state conversion~\cite{Albarelli2018,takagi2018convex}. In particular, due to the additivity of the symplectic rank for pure states, it directly implies that no Gaussian protocol can exactly convert a single-mode non-Gaussian pure state to a tensor product of two non-Gaussian pure states, establishing that separately spreading non-Gaussianity with Gaussian operations is impossible, even allowing post-selection. More generally, it is impossible to exactly map \( n \) single-mode non-Gaussian pure states to \( m \) non-Gaussian pure states via Gaussian operations whenever \( m > n \). Remarkably, this impossibility persists even if the $n$ input states are very distant from the set of Gaussian states and the $m$ output states are extremely close to the set of Gaussian states. This is summarized by the following result:

\begin{coro}[(Impossibility of separately spreading non-Gaussianity via Gaussian operations)] Let $n<m$, and let $\rho$ be a state over $n$ modes and $\phi_1,\dots,\phi_m$ be pure single-mode non-Gaussian states. Then, no post-selected Gaussian protocol can map $\rho$ to $\otimes_{i=1}^m\phi_i$.
\end{coro}


Once the impossibility of converting certain states into others is understood, a natural question to explore is the irreversibility of resource theories. Irreversibility is closely related to the resources required to convert a state $\rho$ into a state $\sigma$, and then to revert $\sigma$ back to $\rho$. Informally, a resource theory is considered reversible when the resources employed in both processes are equal (we refer to Definition~\ref{def_rev_resource} for a formal definition).  For example, in entanglement theory, reversibility holds for pure states~\cite{Bennett-distillation} but fails for mixed states~\cite{Vidal-irreversibility}. 
In this paper, we investigate the irreversibility of non-Gaussianity under exact asymptotic conversions, specifically considering the asymptotic rate of exactly converting $n$ copies of $\rho$ into $m$ copies of a target state $\sigma$ using post-selected Gaussian operations. In this scenario, the following result holds:
\begin{theo}[(Irreversibility of the resource theory of non-Gaussianity, \textit{informal version of}~\cref{appthm:irreversibility})]\label{thm:irreversibility}
  The resource theory of quantum non-Gaussianity is irreversible under post-selected Gaussian protocols in the {exact} conversion setup.
\end{theo}
\noindent \emph{Sketch of the proof}. The proof follows directly by considering the transformation between two Fock states $\ket1^{\otimes2}$ and $\ket2$. It leverages the fact that the conversion rate can be bounded as a function of monotones of quantum non-Gaussianity, specifically the symplectic rank and stellar rank.

\qed





\medskip
\parhead{Complexity of tomography and simulation.}
We have seen that the symplectic rank is a measure of non-Gaussianity which has relevant applications in the resource theory of quantum non-Gaussianity. Here, we show that the symplectic rank has also applications in fundamental quantum information processing tasks, for which it can be interpreted as a measure of complexity. Specifically, we demonstrate how the symplectic rank quantifies the difficulty of performing \emph{quantum state tomography} and \emph{classical simulation}. 

\begin{theo}[(Tomography based on the symplectic rank, \textit{informal version of}~\cref{appthm:tomography})]
Let $\psi$ be an unknown $n$-mode pure state, which is promised to have a symplectic rank of at most $s$ and second moment of the photon number upper bounded by $E$. 
Learning a classical description of a state that is $\varepsilon$-close in trace distance to $\psi$ with high probability requires a number of state copies $N$ satisfying the following upper and lower bounds:
\be\label{bounds_main}
\tilde\Omega\!\left(\left(\frac{E}{12 s\varepsilon}\right)^{\!\! s}\right)\le N\le\tilde{\mathcal O}\!\left(\left(\frac{22E}{ s\varepsilon^2}\right)^{\!\! s}\right),
\ee
where the tilde notation abbreviates $\mathrm{poly}$ and $\mathrm{log}$ factors in all parameters. 
\end{theo}

\noindent The proof of this result is directly based on~\cite[Theorem S74 and S80]{mele2024learning}. We provide a method to improve the energy dependence of the upper bound, compared to the results in~\cite{mele2024learning}. 
Since the energy in an extensive quantity, which usually grows linearly with the number of modes, we can expect that $E > s$. Hence, \eqref{bounds_main} establishes that the number of state copies $N$ must scale exponentially with the symplectic rank, as it is both upper and lower bounded by quantities scaling exponentially in $s$. So, the key takeaway is that the sample complexity of quantum tomography for bosonic states grows exponentially with the symplectic rank.


Similarly, we show that the symplectic rank can also be used to quantitatively bound the computational complexity of simulating bosonic computations:

\begin{theo}[(Classical simulation based on the symplectic rank, \textit{informal version of} \cref{appthm_simul})]
Let $\psi$ be an $n$-mode pure state with symplectic rank $ s$ and mean photon number less than $E$. There is a classical randomized algorithm which takes as input parameters $\varepsilon >0$, the description of a Gaussian state $\ket{G}$, and the decomposition of $\psi$ from \cref{charact_symp}, and outputs an estimate of the overlap $|\braket{G|\psi}|^2$ up to additive precision $\varepsilon$ with high probability, in time $T$ satisfying
\be
T\le\tilde{\mathcal O}\!\left(\left(\frac{6E}{ s\varepsilon^2}\right)^{\!\!3 s}\right),
\ee
where the tilde notation abbreviates $\mathrm{poly}$ and $\mathrm{log}$ factors in all parameters.
\end{theo}
 
\noindent This result establishes high symplectic rank as a necessary --- although not sufficient --- ingredient for quantum computational advantage with bosonic systems: when ${\psi}$ is the output of a bosonic circuit, it implies that the statistics of any Gaussian measurement at the output of the circuit can be estimated efficiently by a classical computer, whenever the symplectic rank is low enough, assuming that the decomposition of the state from \cref{thm_comp} is available. As it turns out, this decomposition can often be computed efficiently given the description of a bosonic circuit that prepares $\psi$ (see \cref{theo:efficient_nG}), and we give two illustrating examples hereafter. 



Firstly, consider an $n$-mode state $\psi$ prepared from the vacuum by a circuit consisting of alternating Gaussian unitary gates and $s$ cubic phase gates. This class of circuits is natural when computing with superconducting modes, for instance for performing fault-tolerant computations in the Gottesman--Kitaev--Preskill encoding \cite{Gottesman2001,sivak2023real}. We have $\mathfrak s(\psi)\le s$, and applying \cref{appthm_simul} allows us to efficiently estimate the statistics of any Gaussian measurement on the output states of such circuits whenever the number of non-Gaussian layers satisfies $s=\mathcal O(1)$ (see \cref{thm:simucubic}). This complements recent algorithms for estimating photon-number statistics \cite{chabaud2024bosonic,upreti2025bounding} and computing quadrature expectation values \cite{upreti2025interplay} at the output of such circuits.

Secondly, consider an $n$-mode state $\psi$ prepared from the vacuum by a circuit consisting of alternating Gaussian unitary gates and $s$ applications of creation and/or annihilation operators (i.e.~operators of the form $ a^{\dag k} a^l$). This class of circuits is natural in quantum optics, where heralded photon additions and/or subtractions are commonly used to engineer non-Gaussianity \cite{lvovsky2020production,walschaers2021non}. As before, $\mathfrak s(\psi)\le s$, and we can efficiently estimate Gaussian statistics of such states whenever the number of non-Gaussian layers satisfies $s=\mathcal O(1)$. We further refine this result to obtain a classical algorithm for strongly simulating any Gaussian measurement at the output of such circuits (see \cref{th:simuphotaddsub}) that is efficient when $s=\mathcal O(1)$. This improves upon a previous classical simulation algorithm \cite{chabaud2020classical}, valid for single-photon additions and subtractions only.

\medskip
\parhead{Robustness to small perturbations.}
The pursuit of precise non-Gaussian measures must be balanced against the pragmatic reality that experimental realizations can only achieve finite accuracy. To that end, we show that the symplectic rank is robust, in the sense that it cannot decrease under small perturbations of the state. 
Specifically, we prove that any pure state admits a ball around it of states with non-smaller symplectic rank, establishing also a simple estimate for the minimal radius of such a ball w.r.t.~trace distance. This result makes the symplectic rank relevant as a practical measure of non-Gaussianity, as it remains unaffected by small perturbations induced by noise inherent in any quantum device. This is formalized by the following theorem:
\begin{theo}[(Slight perturbations cannot decrease the symplectic rank, \textit{informal version of \cref{slight_pert}})]\label{slight_perttt}
Given a pure state $\psi$ such that its second moment of the energy is bounded, then there exists an $\varepsilon$-ball (w.r.t.\ the trace distance) around it such that any energy-constrained state inside the ball has symplectic rank at least as large as that of $\psi$.
\end{theo}
\noindent \emph{Sketch of the proof}. 
The main idea behind the proof is that a variation in the symplectic rank of a pure state is intrinsically connected to a perturbation in the symplectic eigenvalues of its covariance matrix. To formalise such an idea, we introduce novel lower bounds on the trace distance between two states in terms of the norm distance between their covariance matrices, as reported in Theorem~\ref{thm_pert_V:tt} in Section~\ref{sec:end}.
\qed

\medskip
\parhead{The approximate symplectic rank.}
The robustness of the symplectic rank motivates the introduction of an approximate measure based on the symplectic rank, which quantifies how the symplectic rank decreases as one gets further away from the state.
The fidelity is often much easier to compute than the trace distance, which is why we use this measure to define the notion of approximate symplectic rank and symplectic fidelities, similarly to what is done in~\cite{hahn2024assessing}:

\begin{defi}[($\varepsilon$-approximate symplectic rank)]\label{def:asr}
Let $\rho$ be a quantum state over $n$ modes, and let $\varepsilon > 0$. The $\varepsilon$-approximate symplectic rank $\mathfrak s_{\varepsilon}(\rho)$ of $\rho$ is the smallest symplectic rank among all states $\tau$ that are $\varepsilon$-close to $\rho$ — where “$\varepsilon$-close” means their fidelity with $\rho$ is at least $1-\varepsilon$.
\end{defi}

\noindent Interestingly, we prove that the $\varepsilon$-approximate symplectic rank is a monotone under Gaussian operations, as shown in~\cref{tt:mono}. 

\begin{defi}[(\texorpdfstring{$k$}{k}-symplectic fidelity)]
Given a quantum state $\rho$ over $n$ modes and $0\le k\le n$, its $k$-symplectic fidelity $f_k^{\mathfrak s}(\rho)$ is defined as the highest fidelity between $\rho$ and all quantum states whose symplectic rank is less than or equal to $k$.
%
%

\end{defi}

\noindent In particular, $f^{\mathfrak s}_0$ is the fidelity with the closest convex mixture of Gaussian states. Note also that the symplectic fidelity $f_k^{\mathfrak s}(\rho)$ is equal to $1$ whenever $k\ge \mathfrak{s}(\rho)$. In general, symplectic fidelities and approximate symplectic rank are related as follows:

\begin{lem}[(Equivalence between approximate symplectic rank and symplectic fidelities, \textit{informal version of}~\cref{lem:asr-sF:supp})]\label{lem:asr-sF}
Let $\rho$ be a state over $n$ modes.
For all $k\le n$ and all $0\le\varepsilon\le1$,
\begin{equation}
    f^{\mathfrak s}_k(\rho)\ge1-\varepsilon\Leftrightarrow\forall\varepsilon'>\varepsilon,\,\mathfrak{s}_{\varepsilon'}(\rho)\le k.
\end{equation}
\end{lem}

\noindent We also show that the symplectic fidelities can be obtained via an optimization over Gaussian unitary operators for pure states:

\begin{lem}[(Computing symplectic fidelities, \textit{informal version of}~\cref{appth:compsympfid})]\label{th:compute_sf}
    Let $\psi$ be a pure state over $n$ modes. For all $k\le n$, its symplectic fidelities can be expressed as
    \be
    f^{\mathfrak s}_k(\psi)=\sup_{ G}\Tr[({\mathbb1}_k\otimes\ketbra{0}^{\otimes(n-k)}) G\ket\psi\!\bra\psi G^\dag],
    \ee
    where the supremum is over an arbitrary $n$-mode Gaussian unitary $G$.
\end{lem}
\noindent \emph{Sketch of the proof}. Firstly, since we are focusing on pure states, the expression for the fidelity simplifies significantly. Secondly, we may use the fact that pure states with a symplectic rank less or equal to $s$ are equivalent to $s$-compressible states to complete the proof. \qed

\medskip

\noindent Together with \cref{lem:asr-sF}, this result drastically reduces the complexity of computing of the approximate symplectic rank, as compared to \cref{def:asr}, since it can be deduced from the list of symplectic fidelities. While the optimization above is still non-trivial when the number of modes increases, \cref{th:compute_sf} implies that the approximate symplectic rank can be computed efficiently for continuous-variable quantum systems over a small numbers of modes.


\medskip
\parhead{Witnessing the symplectic rank.}
The robustness of the symplectic rank under small perturbations indicates that it may be possible to witness it in experimental scenarios, in which quantum states are inevitably noisy and information reconstruction suffers from statistical errors.

This intuition can be made formal using symplectic fidelities, leading to an explicit and practical protocol for witnessing the symplectic rank of some unknown experimental state $\rho$ as follows: (i) pick a pure target state $\psi$ and compute its symplectic fidelities $f_k^\mathfrak s(\psi)$; (ii) measure several copies of $\rho$ and witness its fidelity with $\psi$ using the samples; (iii) compare the estimated lower bound on the fidelity with the symplectic fidelities --- if it is above $f_k^\mathfrak s(\psi)$ for some $k$, this implies $\mathfrak s(\rho)>k$, thus witnessing the symplectic rank of the experimental state.

This protocol relies on two subroutines: computing the symplectic fidelities of a target pure state of high symplectic rank and estimating a lower bound on the fidelity of an experimental state with that target state based on measurement samples. The former can be achieved via \cref{th:compute_sf}, while the latter can be done efficiently for target pure states that are tensor products \cite[Lemma 2]{chabaud2020efficient} of arbitrary symplectic rank, as well as for a large class of $t$-doped states for $t=\mathcal O(\log n)$ \cite[Protocol 4]{upreti2024efficient}, of symplectic rank less or equal to $t$.



\medskip
\parhead{Approximate Gaussian conversion.}
Like the symplectic rank, the $\varepsilon$-approximate symplectic rank is also a monotone under Gaussian operations (without post-selection) for all $\varepsilon$, as we show in~\cref{lem:approx_monotone}. This implies that the approximate symplectic rank may be used to give bounds on deterministic Gaussian protocols for non-Gaussian state conversion \cite{Albarelli2018}: if $\mathcal G(\rho^{\otimes k})=\sigma^{\otimes m}$ for some Gaussian protocol $\mathcal G$, then $\mathfrak{s}_\varepsilon(\sigma^{\otimes m})\le \mathfrak{s}_\varepsilon(\rho^{\otimes k})$.

More importantly, since we are dealing with an approximate measure, we can extend these bounds to the experimentally relevant case of approximate Gaussian conversion, similar to the approximate stellar rank \cite{hahn2024assessing}: if $\mathcal G(\rho^{\otimes k})$ is $\delta$-close in trace distance to $\sigma^{\otimes m}$, then
%
$\mathfrak{s}_{\varepsilon+\sqrt{2\delta}}(\sigma^{\otimes m})\le \mathfrak{s}_\varepsilon(\rho^{\otimes k})$,
%
for all $\varepsilon\ge0$. 
%
Using \cref{lem:asr-sF}, these bounds can equivalently be phrased in terms of symplectic fidelities as
$f_n^{\mathfrak s}(\rho^{\otimes k})\le f_{n}^{\mathfrak s}(\sigma^{\otimes m})+\sqrt{2\delta}$.
Finally, we show that for $\varepsilon=0$, the previous bound $\mathfrak{s}_{\sqrt{2\delta}}(\sigma^{\otimes m})\le \mathfrak{s}(\rho^{\otimes k})$ also applies to the case of approximate and probabilistic Gaussian conversion, i.e.~including post-selection on Gaussian measurement outcomes (see \cref{th:conversion}).

Combined with \cref{th:compute_sf}, these bounds provide strong constraints on the possibilities of non-Gaussian state engineering using Gaussian protocols \cite{Ourjoumtsev2007generation,Etesse:14,Etesse:2014ty,Arzani2017polynomial,Weigand2018generating,takagi2018convex,Albarelli2018,zheng2020gaussian,hahn2022deterministic,Zheng2023gaussian,takase2024generation,pizzimenti2024optical}.

\medskip
\parhead{Discussion.}
Our work demonstrates that it is possible to define a meaningful measure of non-Gaussianity based on properties of the covariance matrix, often perceived as characterizing Gaussian features of a quantum state. Doing so, we have introduced a novel measure for non-Gaussian states, the symplectic rank, and have established it as a non-Gaussianity monotone, leading to strong constraints on the Gaussian conversion of non-Gaussian states. What is more, we have shown that the symplectic rank is a multi-faceted measure connected to the concepts of non-Gaussian compression and non-Gaussian gate complexity, as well as to the complexity of tomography of bosonic quantum states and of classical simulation of bosonic computations. Additionally, we have shown that the symplectic rank is robust to small perturbations, leading to the notions of approximate symplectic rank and symplectic fidelities, ultimately showing that the symplectic rank can be effectively witnessed in experiments.

Our findings open several research directions. 
From an experimental standpoint, the hardware-agnostic nature of the symplectic rank makes it a compelling candidate for benchmarking and comparing the non-Gaussianity of various bosonic platforms --- an area that has so far lacked such a tool. We expect that the methods introduced here will lead to multiple experiments witnessing the symplectic rank, based on very different hardware and native non-Gaussian interactions, such as photonics with photon addition or subtraction \cite{lvovsky2020production}, or superconducting devices with SNAP gates \cite{krastanov2015universal,heeres2015cavity}.

From a theoretical standpoint, the symplectic rank is set to play a central role in advancing our understanding of non-Gaussianity in the multi-mode setting, particularly through its operational interpretation and associated conversion bounds. 
In particular, future research could explore the similarities between the symplectic rank and the stabilizer nullity in discrete-variable systems~\cite{Beverland_2020} to gain new insights into the symplectic rank and non-Gaussianity, as well as their interplay with entanglement. Comparative studies of the conversion bounds derived here with those based on Wigner negativity or approximate stellar rank would also be especially valuable for assessing bosonic information processing protocols, while combining approximate symplectic rank and approximate stellar rank could settle the open question of the irreversibility of the resource theory of non-Gaussianity in the approximate setting.
We leave these questions for future work.


\subsection{Technical tools}
\label{sec:end}

In this section, we highlight technical tools of independent interest which arise from the proofs of our results.

\medskip
\parhead{Block decomposition of Gaussian unitary operations.} 
We present a novel decomposition for Gaussian unitaries. Specifically, in the forthcoming \cref{lem:decomp0}, we show that any Gaussian unitary can be written as a composition between a global passive Gaussian unitary and the product of two non-global Gaussian unitaries, each of which acts trivially on a suitable subset of modes. See \cref{fig:circuit} for a schematic representation of such a decomposition.

   \begin{figure}
       \centering
       \includegraphics[width=0.6\linewidth]{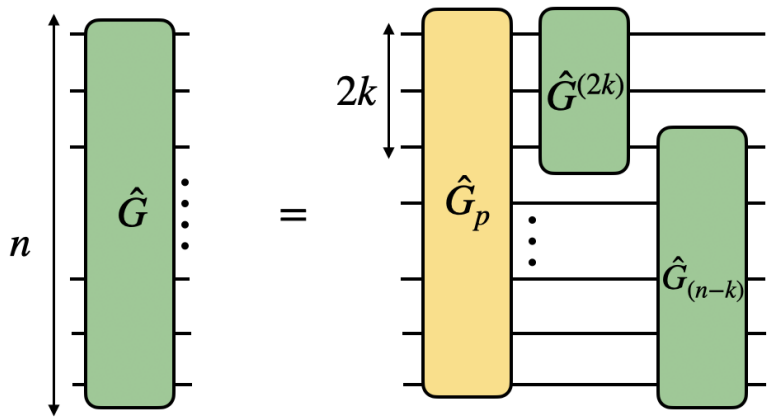}
       \caption{Schematic representation of the block decomposition of Gaussian unitaries in \cref{lem:decomp0}. Specifically, any Gaussian unitary ${G}$ acting on $n$ modes can be decomposed, for all $k \le n/2$, into the composition of: a global passive Gaussian unitary ${G}_p$, a Gaussian unitary ${G}^{(2k)}$ acting on the first $2k$ modes, and a Gaussian unitary ${G}_{(n-k)}$ acting on the last $n-k$ modes. }
       \label{fig:circuit}
   \end{figure}
\begin{lem}[(Block decomposition of Gaussian unitary operations, \textit{informal version of} \cref{lem:decomp})]\label{lem:decomp0}
    Let ${G}$ be an $n$-mode Gaussian unitary and let $k\le n/2$ be an integer. Then, there exists an $n$-mode passive Gaussian unitary ${G}_p$, a $2k$-mode Gaussian unitary ${{G}^{(2k)}}$ acting only on the first $2k$ modes, and a $(n-k)$-mode Gaussian unitary ${G}_{(n-k)}$ acting only on the last $n-k$ modes such that
    \be 
        {G}=\left( \mathbb{1}_{k}\otimes {G}_{(n-k)}\right) \left(  {{G}^{(2k)}}\otimes\mathbb{1}_{n-2k}\right) {G}_p\,.
    \ee
\end{lem}

\noindent This decomposition, proven in  \cref{lem:decomp}, leads to the following remarkable consequence:

\begin{coro}[(Efficient disentanglement of Gaussian states)]
Let $m$ be an even natural number. Any bipartite (possibly mixed) Gaussian state over $m+m$ modes can be disentangled by applying a suitable $(\frac{3}{2}m)$-mode Gaussian unitary. 
\end{coro}
\noindent Remarkably, this result allows one to \emph{gain} $25\%$ of modes: although a priori one would use a $(2m)$-mode unitary to disentangle a bipartite entangled Gaussian state over $m+ m$ modes, our result establishes that $\frac{3}{2}m$ modes are sufficient. 

\medskip
\parhead{Perturbation bounds on first moments, covariance matrices, and symplectic eigenvalues.}\label{sec:perturbation}
In the forthcoming Theorems~\ref{thm_pert_m:tt}-\ref{pert_bound_symp:tt}, we introduce perturbation bounds for the first moments, covariance matrices, and symplectic eigenvalues of arbitrary bosonic states. These results address the seemingly simple yet fundamental question: "If two quantum states are close in trace distance, how close are their corresponding first moments, covariance matrices, and symplectic eigenvalues?" We believe these bounds are of independent interest within the field of quantum information with continuous-variable systems. In particular, they may find applications for the analysis of quantum state tomography~\cite{mele2024learning}, characterization of bosonic observables~\cite{holevo2024estimatestracenormdistancequantum}, and testing of properties of bosonic systems~\cite{bittel2024optimaltracedistanceboundsfreefermionic}. More broadly, our results contribute to the rapidly growing literature regarding trace distance bounds for bosonic systems~\cite{mele2024learning,holevo2024estimatestracenormdistancequantum,holevo2024estimatesburesdistancebosonic,fanizza2024efficienthamiltonianstructuretrace,mele2025achievableratesnonasymptoticbosonic,bittel2024optimaltracedistanceboundsfreefermionic}, extending previous findings that were limited to the special case of Gaussian states~\cite{mele2024learning}.


We start with the perturbation bound for first moments:
\begin{theo}[(Trace distance perturbation bounds on first moments, \textit{informal version of} \cref{thm_pert_m})]\label{thm_pert_m:tt}
    Let $\rho$ and $\sigma$ be two quantum states with mean energy upper bounded by a finite constant $E$.  Then, the Euclidean distance between their first moments can be upper bounded in terms of the trace distance as
    \be
        \|\bm{m}(\rho)-\bm{m}(\sigma)\|\le 4\sqrt{E\|\rho-\sigma\|_1}\,,
    \ee
    where the Euclidean norm $\|\cdot\|$ is defined as $\|\bm{v}\|\coloneqq \sqrt{\bm{v}^\intercal\bm{v}}$.
\end{theo}
A similar bound can be derived in terms of covariance matrices:
\begin{theo}[(Trace distance perturbation bounds on covariance matrices, \textit{informal version of} \cref{thm_pert_V})]\label{thm_pert_V:tt}
    Let $\rho$ and $\sigma$ be two quantum states with second moment of the energy upper bounded by a finite constant $E$.  Then, the operator norm distance between their covariance matrices can be upper bounded in terms of the trace distance as
    \be\label{tight_bound:tt}
        \|V(\rho)-V(\sigma)\|_\infty\le  40\,E \sqrt{\|\rho-\sigma\|_1} \,,
    \ee
    where $\|\cdot\|_\infty$ denotes the operator norm.
\end{theo}
The above theorem shows that if two energy-constrained states are $\varepsilon$-close in trace distance, then their covariance matrices differ by at most $O(\sqrt{\varepsilon})$. Moreover, we prove that the bound in Eq.~\eqref{tight_bound:tt} is tight with respect to both the second moment of the energy $E$ and the trace norm $\|\rho - \sigma\|_1$ (see Example~\ref{thm_pert_V}), thus establishing that the mentioned scaling $O(\sqrt{\varepsilon})$ is tight. Notably, if we have the additional assumption that the states are Gaussian, this scaling can be improved to $O(\varepsilon)$~\cite{mele2024learning}. Finally, as an application of the above Theorem~\ref{thm_pert_V:tt}, we establish the following perturbation bound on symplectic eigenvalues, which turns out to be crucial in analysing the symplectic rank (e.g.~for the proof of the above Theorem~\ref{slight_perttt}).
\begin{theo}[(Trace distance perturbation bound on symplectic eigenvalues, \textit{informal version of} \cref{pert_bound_symp})]\label{pert_bound_symp:tt}
    Given a quantum state $\rho$, let us denote as $\nu_i^{\downarrow}(\rho)$ the $i$th largest symplectic eigenvalue of its covariance matrix $V(\rho)$. Let $\rho$ and $\sigma$ be two quantum states with second moment of the energy upper bounded by a finite constant $E$. Then, it holds that
    \begin{equation}
        \max_i\left|\nu_i^{\downarrow}(\rho)-\nu_i^{\downarrow}(\sigma)\right|\le 640 E^2\sqrt{\|\rho-\sigma\|_1}\,.
    \end{equation}
\end{theo}
\medskip
\parhead{Classical perturbation bounds.}\label{sec:classicalperturbation}
By adapting the proof technique used for the perturbation bounds above, we derive analogous bounds in the classical setting, considering probability distributions and total variation distance. These bounds may be of independent interest in the context of classical probability theory. In particular, we establish the following lower bound on the total variation distance in terms of differences in mean values:
\begin{lem}[(Lower bounds on the total variation distance in terms of the difference of the mean values, \textit{informal version of~\cref{classical_lower_mean}})]\label{classical_lower_mean_main}
    Let $p$ and $q$ be probability distributions over $\mathbb{R}^n$. Then, the total variation distance between $p$ and $q$ can be lower bounded as
    \begin{equation}
        \|p-q\|_1\ge \frac{\left\|\mathbb{E}_p\!\left[\bm{x}\right]-\mathbb{E}_q\!\left[\bm{x}\right]\right\|^2}{8\max\left(\mathbb{E}_p\!\left[\|\bm{x}\|^2\right],\mathbb{E}_q\!\left[\|\bm{x}\|^2\right]\right)}\,,
    \end{equation}
    where $\|\bm{x}\|\coloneqq \sqrt{\sum_{i=1}^nx_i^2}$ and $\mathbb{E}_p$ denotes the expectation value with respect to $p$.
\end{lem}
We also show a similar result in terms of covariance matrices:
\begin{lem}[(Lower bounds on the total variation distance in terms of the difference of the covariance matrices)]
    Let $p$ and $q$ be probability distributions over $\mathbb{R}^n$. Then, the total variation distance between $p$ and $q$ can be lower bounded as
    \begin{equation}
        \|p-q\|_1\ge \frac{\left\|V(p)-V(q)\right\|_\infty^2}{72\max\left(\mathbb{E}_p\!\left[\|\bm{x}\|^4\right],\mathbb{E}_q\!\left[\|\bm{x}\|^4\right]\right)}\,,
    \end{equation}
    where $V(p)$ denotes the covariance matrix of $p$ defined by $V(p)\coloneqq \mathbb{E}_p\!\left[(\bm{x}-\mathbb{E}_p[\bm{x}])(\bm{x}-\mathbb{E}_p[\bm{x}])^\intercal \right]$.
\end{lem}

\medskip

\parhead{Monotonicity of approximate resource measures.}\label{tt:mono}
We show that any resource monotone can be extended to a family of approximate resource measures that are also monotones:

\begin{lem}[(Approximate monotone)]\label{lem:approx_monotone}
    Let us consider a quantum resource theory~\cite{gour2024resourcesquantumworld} characterized by a set of free states $\pazocal F$ and by a set of free quantum channels $\pazocal G$. Let $\pazocal M$ be a monotone under free operations, i.e.~$\pazocal M(\Lambda(\rho))\le\pazocal M(\rho)$ for all states $\rho$ and all quantum channels $\Lambda\in\pazocal G$. For $\varepsilon\in(0,1)$, we define the $\varepsilon$-approximate monotone $\pazocal M_\varepsilon$ as
    \begin{equation}
        \pazocal M_\varepsilon(\rho)\coloneqq \inf_{\substack{\tau:\\\frac12\|\rho-\tau\|_1\le\varepsilon}}\pazocal M(\tau),
    \end{equation}
    for all $\rho$, where the infimum is taken over all the states $\tau$ that are $\varepsilon$-close to $\rho$ in trace distance. Then, $\pazocal M_\varepsilon$ is also a monotone under free operations.
\end{lem}

\begin{proof}
Fix $\varepsilon\in(0,1)$, a quantum state $\rho$, and a free quantum channel $\Lambda\in\pazocal G$. By definition of $\pazocal M_\varepsilon$, for all $\delta>0$ there exists a density operator $\tau_\delta$ such that $\frac12\|\rho-\tau_\delta\|_1\le\varepsilon$ and
\begin{equation}\label{eq:mapprox1}
    \pazocal M(\tau_\delta)\le\pazocal M_\varepsilon(\rho)+\delta.
\end{equation}
On the other hand, since the trace distance does not increase under quantum channels~\cite{NC}, it follows that
\begin{equation}
    \frac12\|\Lambda(\rho)-\Lambda(\tau_\delta)\|_1\le \frac12\|\rho-\tau_\delta\|_1\le\varepsilon\,.
\end{equation}
 Hence, 
\begin{equation}\label{eq:mapprox2}
    \pazocal M_\varepsilon(\Lambda(\rho))\le \pazocal M(\Lambda(\tau_\delta)),
\end{equation}
by definition of $\pazocal M_\varepsilon$. Moreover,
\begin{equation}\label{eq:mapprox3}
    \pazocal M(\Lambda(\tau_\delta))\le\pazocal M(\tau_\delta),
\end{equation}
since $\pazocal M$ is a monotone. Combining \cref{eq:mapprox1,eq:mapprox2,eq:mapprox3} shows that 
\begin{equation}
    \pazocal M_\varepsilon(\Lambda(\rho))\le\pazocal M_\varepsilon(\rho)+\delta
\end{equation}
for all $\delta>0$, and thus $\pazocal M_\varepsilon(\Lambda(\rho))\le\pazocal M_\varepsilon(\rho)$. Consequently, $\pazocal M_\varepsilon$ is a monotone under free operations.
\end{proof}

Note that this result also holds if the approximation is measured via the (in)fidelity rather than the trace distance. This implies, in combination with \cref{thm_monotonicity}, that the $\varepsilon$-approximate symplectic rank is a non-Gaussianity monotone.

\section{Preliminaries on continuous-variable quantum systems}
\label{sec_prel}

In this section we introduce the relevant preliminaries of continuous-variable quantum systems; for more details, we refer to~\cite{BUCCO}. First, we start by defining what a continuous-variable quantum system is, by introducing the Hilbert space of \( n \) modes. We then define key concepts such as first moments, covariance matrices, energy operator, and photon number operator. Next, we discuss Gaussian states and Gaussian unitary operations. Finally, we introduce coherent states and the stellar representation of continuous-variable quantum states. 

\subsection{The Hilbert space of \texorpdfstring{$n$}{n} modes}

A continuous-variable system is represented by the Hilbert space $L^2(\mathbb R^n)$ of all square-integrable complex-valued functions over $\R^n$, where $n$ represents the number of modes~\cite{BUCCO}. A state and unitary acting on such a Hilbert space are dubbed $n$-mode state and $n$-mode unitary, respectively. The \emph{quadrature operator vector}, denoted as $\hat{\bm{R}}$, is defined as 
 \be
 \hat{\bm{R}}\coloneqq (\hat{x}_1,\hat{p}_1,\ldots,\hat{x}_n,\hat{p}_n)\,,
 \ee
 where $\hat{x}_j$ and $\hat{p}_j$ are the well-known \emph{position} and \emph{momentum} operators of the $j$th mode~\cite{BUCCO}. The {position operator} $\hat{x}_j$ and the {momentum operator} $\hat{p}_j$ of the $j$th mode can be written in terms of the well-known \emph{annihilation operator} $a_j$ and \emph{creation operator} $a_j^\dagger$ as
 \be
    \hat{x}_j&\coloneqq\frac{a_j+a_j^\dagger}{\sqrt2}\,,\\
   \hat{p}_j&\coloneqq\frac{a_j-a_j^\dagger}{i\sqrt2}\,,
 \ee
with the convention $\hbar=1$.
The vector of quadrature operators satisfies the following relation, known as \emph{canonical commutation relation}:
\be
[\hat{\bm{R}},\hat{\bm{R}}^{\intercal}]=i\,\Omega\,\mathbb{{1}}\,,
\ee 
where the symbol $(\cdot)^\intercal$ denotes the transpose operation and where 
\be\label{symp_form}
        \Omega = \bigoplus_{i=1}^{n} 
\begin{pmatrix}
    0 & 1 \\ 
    -1 & 0
\end{pmatrix}
\ee
is the so-called symplectic form. The Hilbert space $L^2(\mathbb{R}^n)$ can be seen as the span of the so-called \emph{multi-mode Fock states} $\{\ket{\bm{k}}\}_{\bm{k}\in\mathbb{N}^n}$~\cite{BUCCO}:
\be\label{eq_l2}
    L^2(\mathbb{R}^n)\coloneqq \mathrm{span}\left\{\ket{\bm{k}}:\quad\bm{k}\in\mathbb{N}^n\right\}\,,
\ee
where $\ket{\bm{k}}\coloneqq \ket{k_1}\otimes\ldots\otimes\ket{k_n}$, with $\ket{k}$ denoting the single-mode $k$th \emph{Fock state}~\cite{BUCCO}. The latter represents the state with $k$ \emph{photons} (or more generally, bosons), and it can be obtained by applying the creation operator on the vacuum state as 
\be
    \ket{k}\coloneqq \frac{(a^\dagger)^k}{\sqrt{k!}}\ket{0}\,.
\ee
The Fock states $\ket{0}$ and $\ket{1}$ are called the \emph{vacuum state} and \emph{single-photon state}, respectively~\cite{BUCCO}. Via Eq.~\eqref{eq_l2}, an $n$-mode continuous-variable quantum system can be viewed as an $n$-qu$d$it system with $d=\infty$.

\subsection{First moments, covariance matrices, and symplectic matrices}

The \emph{first moment} $\bm{m}(\rho)$ and the \emph{covariance matrix} $V\!(\rho)$ of a quantum state $\rho$ are defined as~\cite{BUCCO}
\begin{align}
    \bm{m}(\rho)&\coloneqq \Tr[\bm{\hat{R}}\, \rho ]\,,\\
    V\!(\rho)  &\coloneqq\Tr\!\left[\left\{\bm{(\hat{R}-m(\rho)\,{\mathbb{1}}),(\hat{R}-m(\rho)\,{\mathbb{1}})}^{\intercal}\right\}\rho\right],
\end{align}
where $\{{A},{B}\}\coloneqq {A}{B}+{B}{A}$ is the anti-commutator. Note that the definition of the covariance matrix differs by a factor $\frac12$ compared to \cite{biblioparis}, such that the covariance matrix of the vacuum is the identity matrix, while $\hbar=1$. Any covariance matrix $V$ is strictly positive and satisfies the so-called uncertainty relation~$V+i\Omega\succeq0$. As such, any covariance matrix $V$ can be written in the so-called \emph{Williamson decomposition} as follows: there exists a \emph{symplectic} matrix $S$ and $n$ real numbers $d_1\ge d_2\ge \ldots\ge d_n$ --- called the \emph{symplectic eigenvalues} of $V$ --- such that
\be\label{will_eq_def}
V=SDS^{\intercal},
\ee 
where $D\coloneqq \operatorname{diag}(d_1, d_1, \ldots, d_n, d_n)$. Importantly, the symplectic eigenvalues of a covariance matrix are unique and always larger than or equal to one. By definition, a matrix $S$ is said to be symplectic if it satisfies $S \Omega S^\intercal =\Omega$, where $\Omega$ is the symplectic form defined in Eq.~\eqref{symp_form}. Symplectic matrices satisfy the following properties:
\begin{itemize}
    \item The product of symplectic matrices is symplectic;
    \item The transpose of a symplectic matrix is symplectic;
    \item The inverse of symplectic matrix is symplectic;
    \item Any symplectic matrix $S$ can be decomposed in the so-called \emph{Euler decomposition} as follows: there exist suitable symplectic orthogonal matrices $O_1,O_2$, and suitable real numbers $z_1,z_1,\ldots,z_n\ge 1$ such that
    \be\label{Euler_dec2}
        S=O_1ZO_2\,,
    \ee
    where $Z$ is a diagonal matrix --- known as \emph{squeezing matrix} --- defined by
    \be\label{sq_matttt}
    Z\coloneqq \bigoplus_{j=1}^n\left(\begin{matrix}z_j&0\\0&z_j^{-1}\end{matrix}\right)\,.
\ee
\end{itemize}


\subsection{Energy operator and photon number operator}

The \emph{energy operator}, denoted as $\hat{E}$, is defined as
\be
    \hat{E}\coloneqq \frac12\hat{\bm{R}}^\intercal \hat{\bm{R}}=\frac12\sum_{i=1}^n\left(\hat{x}_i^2+\hat{p}_i^2\right)\,.
\ee
Moreover, the \emph{total photon number operator}, denoted as $\hat{N}$, is defined as $\hat{N}\coloneqq \hat{E}-\frac{n}{2}\mathbb{1}$ and can be written in terms of the annihilation and creation operators as 
\be
    \hat{N}=\sum_{i=1}^na_i^\dagger a_i\,.
\ee
The energy operator and the total photon number operator can be diagonalized in the basis of multi-mode Fock states $\{\ket{\bm{k}}\}_{\bm{k}\in\mathbb{N}^n}$ as follows~\cite{BUCCO}:
\be
    \hat{N}&=\sum_{\bm{k}\in\mathbb{N}^n}\left(\sum_{i=1}^n k_i\right)\ketbra{\bm{k}}\,,\\
    \hat{E}&=\sum_{\bm{k}\in\mathbb{N}^n}\left(\sum_{i=1}^n k_i+\frac{n}{2}\right)\ketbra{\bm{k}}\,.
\ee
By definition, the \emph{mean energy} of a state $\rho$ is given by $\Tr[\hat{E}\rho]$, while its \emph{mean photon number} is given by $\Tr[\hat{N}\rho]$. From these definitions, it easily follows that the mean energy and the mean photon number can be expressed in terms of the first moment and the covariance matrix as
    \be\label{formula_mean_energy}
        \Tr[\hat{E}\rho]&=\frac{\Tr V(\rho)}{4}+\frac{\|\bm{m}(\rho)\|^2}{2}\,,\\
        \Tr[\hat{N}\rho]&=\frac{\Tr\!\left[ V(\rho)-\mathbb{1}\right]}{4}+\frac{\|\bm{m}(\rho)\|^2}{2}\,.  
    \ee
Note that the first moment and the covariance matrix of a state can admit infinite elements, i.e.~they can be ill-defined. The following lemma guarantees that the first moment and the covariance matrix are well-defined as long as the mean energy of the state is finite.
\begin{lem}[(Finite mean energy implies well-defined first moment and covariance matrix)]\label{well_defined_cov}
    Let $\rho$ be a quantum state with finite mean energy. Then, all the elements of the covariance matrix $V(\rho)$ and first moment $\bm{m}(\rho)$ are finite.
\end{lem}
\begin{proof}
    The general formula for the mean energy in Eq.~\eqref{formula_mean_energy} directly implies that all the elements of the first moment are finite if the state has finite mean energy. It suffices to prove that the operator norm of the covariance matrix can be upper bounded in terms of the mean energy. This follows by the following chain of inequalities:
\be\label{bounded_cov_matrix}
    \|V(\rho)\|_\infty\le \Tr [V(\rho)]= 4\Tr[\hat{E}\rho]-2\|\bm{m}(\rho)\|^2\le 4\Tr[\hat{E}\rho]\,,
\ee
where the first inequality follows from the fact that the covariance matrix is positive, and in the equality we exploited the general formula for the mean energy in Eq.~\eqref{formula_mean_energy}.
\end{proof}


\subsection{Gaussian states and Gaussian unitary operations}

Let us proceed with the definition of Gaussian unitary operations and states~\cite{BUCCO}. By definition, a \emph{Gaussian unitary operation} $G$ is a unitary operator that can be prepared by means of evolutions induced by quadratic Hamiltonians in $\bm{\hat{R}}$. The most general form of a Gaussian unitary operation is as follows~\cite{BUCCO}:  
\be\label{gauss_unitary}
    G =  U_S{D}_{\bm{r}}\,,  
\ee  
where 
\be\label{disp_real}
{D}_{\bm{r}}\coloneqq e^{-i\,\bm{r}^\intercal\Omega \hat{\bm{R}}}
\ee
represents the \emph{displacement operator}~\cite{BUCCO}, characterixed by the displacement vector \( \bm{r} \in \mathbb{R}^{2n} \), while $U_S$ corresponds to the \emph{symplectic Gaussian unitary}~\cite{BUCCO}, associated with a \emph{symplectic matrix} $S$. The displacement operator \( {D}_{\bm{r}} \), associated with a vector \( \bm{r} \in \mathbb{R}^n \), acts on the quadrature operator vector as
\be  
    {D}_{\bm{r}}^\dagger \hat{\bm{R}} {D}_{\bm{r}} = \hat{\bm{R}} + \bm{r} {\mathbb{1}} \,,  
\ee  
while it acts at the level of first moments and covariance matrices as
\be
    \bm{m}({D}_{\bm{r}}\rho {D}_{\bm{r}}^\dagger)&=\bm{m}(\rho)+\bm{r}\,,\\
    V({D}_{\bm{r}}\rho {D}_{\bm{r}}^\dagger)&=V(\rho)\,,
\ee
for any state $\rho$. The above definition in Eq.~\eqref{disp_real} of the displacement operator ${D}_{\bm{r}}$ is with respect to the so-called \emph{real notation}. It is sometimes useful to parametrise displacement operators with respect to the \emph{complex notation}. That is, one may define the displacement operator, associated with the complex vector $\bm{\beta}\coloneqq (\beta_1,\beta_2,\ldots,\beta_n)\in\mathbb{C}^n$, as
\be\label{disp_complex}
    {D}(\bm{\beta})\coloneqq e^{i\left(\bm{a}^\dagger \bm{\beta}-\bm{\beta}^\dagger\bm{a}\right)}=\bigotimes_{j=1}^n e^{i\left(\beta_ja_j^\dagger-\beta_j^\ast a_j\right)}\,.
\ee
To switch from one definition to the other, one may use the following change of variables: 
\be
    r_{2j-1}&\coloneqq \sqrt{2}\,\mathrm{Re}(\beta_j)\,,\\
    r_{2j}&\coloneqq \sqrt{2}\,\mathrm{Im}(\beta_j)\,,
\ee
for all $j=1,\ldots,n$, which guarantees that ${D}(\bm{\beta})={D}_{\bm{r}}$.

For a given symplectic matrix \( S \), the symplectic Gaussian unitary $U_S$ acts on the quadrature operator vector as 
\be
    U_S^\dagger \hat{\bm{R}} U_S = S \hat{\bm{R}} \,,
\ee  
while it acts at the level of first moments and covariance matrices as
\be\label{action_gauss_moment}
    \bm{m}(U_S\rho U_S^\dagger)&=S\bm{m}(\rho)\,,\\
    V(U_S\rho U_S^\dagger)&=SV(\rho)S^\intercal\,,
\ee
for any state $\rho$. Moreover, given two symplectic matrices $S_1$ and $S_2$, the associated symplectic Gaussian unitary satisfies $U_{S_1S_2}=U_{S_1}U_{S_2}$. Hence, the Euler decomposition reported in Eq.~\eqref{Euler_dec2} establishes that any symplectic Gaussian unitary $U_S$ can be decomposed as:
\be\label{Gauss_eul}
    U_S=U_{O_2}U_ZU_{O_1}\,,
\ee
where $O_1,O_2$ are orthogonal symplectic matrices and $Z$ is the squeezing matrix in Eq.~\eqref{sq_matttt}. By definition, a \emph{passive Gaussian unitary} is a Gaussian unitary operation associated with an orthogonal symplectic matrix $O$. Importantly, any passive Gaussian unitary $U_O$ preserves the energy (and in particular the photon number operator)~\cite{BUCCO}, i.e.~$U_O^\dagger \,\hat{E}\,U_O=\hat{E} $. Notably, the set of $2n\times 2n$ orthogonal symplectic matrices is in one-to-one correspondence with the set of $n\times n$ unitary matrices. In particular, any $n\times n$ unitary matrix $U$ defines a suitable $n$-mode passive Gaussian unitary, denoted as $G_U$, which acts on the vector of annihilation operators $\bm{a}\coloneqq (a_1,a_2,\ldots,a_n)$ as
\be\label{gauss_passs}
    G_U^\dagger \bm{a} G_U= U^\ast\bm{a} \,.
\ee
Given a squeezing matrix $Z=\diag(z_1,z_1^{-1},\ldots,z_n,z_n^{-1})$ as in Eq.~\eqref{sq_matttt}, the Gaussian unitary operation $G_Z$ is called a \emph{squeezing unitary}. Specifically, it can be shown that $G_Z$ can be expressed at the level of the Hilbert space as~\cite{BUCCO}
\be
    G_Z= \bigotimes_{i=1}^n e^{\frac12\ln(z_i)^2({a_i^\dagger}^2-a_i^2)}\,.
\ee
It is sometimes useful to denote the squeezing unitary as ${S}(\bm{\xi})$, where $\bm{\xi}\coloneqq(\xi_1,\ldots,\xi_n)$ and $\xi_i\coloneqq \ln z_i$, so that
\be\label{def_squeeeez}
    {S}(\bm{\xi})=\bigotimes_{i=1}^n e^{\frac12\xi_i({a_i^\dagger}^2-a_i^2)}\,.
\ee
By putting everything together, Eq.~\eqref{gauss_unitary} and the Euler decomposition in Eq.~\eqref{Gauss_eul} implies that any Gaussian unitary operation $G$ can always be expressed as the composition between a displacement operator ${D}(\bm{\beta})$, a passive Gaussian unitary operation $G_{U_1}$, a squeezing unitary operation ${S}(\bm{\xi})$, and another passive Gaussian unitary operation $G_{U_2}$:
\be\label{most_gen_gauss}
    G=G_{U_2}{S}(\bm{\xi})G_{U_1}{D}(\bm{\beta})\,.
\ee
Moreover, since Eq.~\eqref{gauss_passs} implies that
\be\label{transf_disp_st}
    G_{U} {D}(\bm{\beta})G_U^\dagger= {D}(U^\intercal\bm{\beta})\,,
\ee
it follows that \cref{most_gen_gauss} can also be expressed as
\be\label{eq_ulti}
    G=G_{U_2}{S}(\bm{\xi}){D}(U_1^\intercal\bm{\beta})G_{U_1}\,.
\ee
Consequently, we conclude that the following lemma holds.
\begin{lem}[(Parametrization of Gaussian unitary operations)]\label{lem_param_Gauss}
An arbitrary $n$-mode Gaussian unitary $G$ can be parametrized as follows:
\be\label{eq_euler}
    G= G_US(\bm\xi)D(\bm\beta)G_V,
\ee
where $G_U,G_V$ are passive Gaussian unitary operations associated with some unitary matrices $U,V\in\mathbb{C}^{n\times n}$, $ S(\bm{\xi})$ is the $n$-mode squeezing unitary operation of squeezing vector $\bm{\xi}\in\mathbb{R}^{n}$ defined in Eq.~\eqref{def_squeeeez},
and $ D(\bm\beta)$ is the displacement operator of displacement vector $\bm\beta\in\mathbb{C}^n$ defined in Eq.~\eqref{disp_complex}.
\end{lem}

Let us now proceed with the definition of Gaussian states.  By definition, a \emph{Gaussian state} is a tensor product of Gibbs state of Hamiltonians that are quadratic in the quadrature operator vector $\bm{\hat{R}}$. A fundamental example of Gaussian state is the single-mode vacuum state $\ket{0}$. Any $n$-mode pure Gaussian state ${\psi}$ can be prepared by applying a Gaussian unitary operation $ G$ to the vacuum, i.e.~$\ket{\psi}= G\ket{0}^{\otimes n}$. By exploiting the parametrization of Gaussian unitary operations presented in  \cref{lem_param_Gauss} and the fact that passive Gaussian unitary operations preserves the vacuum~\cite{BUCCO}, it follows that following lemma holds.

\begin{lem}[(Parametrization of pure Gaussian states)]\label{lem_param_pure}
An arbitrary $n$-mode pure Gaussian state ${\psi}$ can be parametrized as follows:
\be\label{eq_pure}
    \ket{\psi}=G_{U}{S}(\bm{\xi}){D}(\bm{\beta})\ket{0}^{\otimes n}\,,
\ee
where $G_U$ is the passive Gaussian unitary associated with some unitary matrix $U,V\in\mathbb{C}^{n\times n}$, $ S(\bm{\xi})$ is the $n$-mode squeezing unitary of squeezing vector $\bm{\xi}\in\mathbb{R}_+^{n}$ defined in Eq.~\eqref{def_squeeeez},
and $ D(\bm\beta)$ is the displacement operator of displacement vector $\bm\beta\in\mathbb{C}^n$ defined in Eq.~\eqref{disp_complex}.
\end{lem}

It is well known that a Gaussian state is uniquely identified by its fist moment $\bm{m}$ and its covariance matrix $V$. That is, for any $\bm{m}\in\mathbb{R}^{2n}$ and any $V\in\mathbb{R}^{2n\times 2n}$ such that $V+i\Omega\succeq0$ there exists a Gaussian state with first moment $\bm{m}$ and covariance matrix $V$~\cite{BUCCO}. Moreover, any (possibly mixed) Gaussian state satisfies the following decomposition --- known as \emph{normal mode decomposition}~\cite{BUCCO}.
\begin{lem}[(Normal mode decomposition of a Gaussian state~\cite{BUCCO})]\label{lem_normal_mode}
    Let $\rho$ be an $n$-mode Gaussian state with first moment $\bm{m}(\rho)$ and covariance matrix $V(\rho)$. Let 
    \be
        V(\rho)=SDS^{\intercal}
    \ee
    be the Williamson's decomposition of $V$ as in Eq.~\eqref{will_eq_def}, where $D\coloneqq(d_1,d_1,\ldots,d_n,d_n)$ is the diagonal matrix of symplectic eigenvalues. Then, the Gaussian state $\rho$ is unitarily equivalent --- via Gaussian unitary operations --- to a tensor product of thermal states, i.e.~
    \be
        \rho= {D}_{\bm{m}(\rho)}U_S\left(\tau_{\frac{d_1-1}{2}}\otimes\ldots \tau_{\frac{d_n-1}{2}}\right)U_S^{\dagger}{D}_{\bm{m}(\rho)}^\dagger\,,
    \ee
    where:
    \begin{itemize}
        \item ${D}_{\bm{m}(\rho)}$ is the displacement operator with amplitude equal to the first moment $\bm{m}(\rho)$;
        \item $U_S$ is the symplectic Gaussian unitary associated with the symplectic matrix $S$ that puts the covariance matrix $V(\rho)$ in Williamson's decomposition;
        \item For any $\nu\ge0$, the state $\tau_\nu$ is the single-mode thermal state of mean photon number $\nu$, defined as
        \be
            \tau_\nu\coloneqq\frac{1}{\nu+1}\sum_{n=0}^\infty\left(\frac{\nu}{\nu+1}\right)^n\ketbra{n}\,,
        \ee
        with $\ket{n}$ being the $n$th Fock state.
    \end{itemize}
\end{lem}
As a consequence of the above decomposition, a Gaussian state is pure if and only if its symplectic eigenvalues are all equal to one. Additionally, it is also worth to mention the fact that a pure state is Gaussian if and only if all its symplectic eigenvalues are equal to one.

 
\subsection{Coherent states and stellar functions}

Given a complex vector $\bm z\in\mathbb{C}^n$, the $n$-mode \emph{coherent state} of amplitude $\bm{z}$ is defined as 
\be
\ket{\bm{z}}\coloneqq {D}(\bm{z})\ket{0}^{\otimes n}\,,
\ee
where ${D}(\bm{z})$ is the displacement operator defined in Eq.~\eqref{disp_complex}. Notably, coherent states form an overcomplete basis of the $n$-mode Hilbert space, i.e.~\cite{BUCCO}
\be\label{rel_compl_coh} \int_{\mathbb{C}^n}\frac{\mathrm{d}^{2n}\bm{z}}{\pi^n}\ketbra{\bm{z}}=\mathbb{1}\,.
\ee
Let us proceed with the definition of the \emph{stellar function}~\cite{chabaud2020stellar,chabaud2021holomorphic}.

\begin{defi}[(Stellar function)]
The {stellar function} $F^\star_{\ket{\psi}}:\mathbb{C}^n\longmapsto\mathbb{C}$ of a (possibly unnormalized) $n$-mode state ${\psi}$ is defined as
\be
    F^\star_{\ket \psi}(\bm z)\coloneqq e^{\frac12\|\bm z\|^2}\langle\bm z^*|\psi\rangle\qquad\forall\bm{z}\in\mathbb{C}^n\,,
\ee
where $\ket{\bm{z}}$ denotes a coherent state of amplitude $\bm{z}$, and $\|\bm z\|^2=|z_1|^2+\dots+|z_n|^2$.
\end{defi}
As a consequence of Eq.~\eqref{rel_compl_coh}, any vector ${\psi}$ can be retrieved by its stellar function via the following formula:
\be
    \ket{\psi}=\int_{\mathbb{C}^n}\frac{\mathrm{d}^{2n}\bm{z}}{\pi^n} \,e^{-\frac12\|\bm z\|^2}F^\star_{\ket \psi}(\bm z)\ket{\bm{z}^\ast}.
\ee
In particular, stellar functions and pure states are in one-to-one correspondence. In addition, Eq.~\eqref{rel_compl_coh} also implies that any stellar function satisfies the following normalization condition:
\be\label{eq_normalizing}
    \int_{\mathbb{C}^n}\frac{\mathrm{d}^n \bm{z}}{\pi^n}  e^{-\|\bm z\|^2} |F^\star_{\ket \psi}(\bm z)|^2=\braket{\psi|\psi}\,.
\ee
Let us now understand how stellar functions transform under passive Gaussian unitary operations. By exploiting Eq.~\eqref{transf_disp_st}, given a unitary matrix $U\in\mathbb{C}^{n\times n}$, the $n$-mode passive Gaussian unitary operation $G_U$ acts at the level of stellar functions as follows:
\be\label{transf_st}
    F^\star_{G_U\ket \psi}(\bm z)=F^\star_{\ket \psi}(U\bm z)\qquad\forall\bm{z}\in\mathbb{C}^n\,.
\ee
Recall that  \cref{lem_param_pure} establishes that an arbitrary pure state can be parametrized in terms of a unitary matrix $U\in\mathbb{C}^{n\times n}$, a real vector $\bm{\xi}\in\mathbb{R}_+^{n}$, and a complex vector $\bm{\beta}\in\mathbb{C}^n$. The following lemma establishes a closed formula for the stellar function of a Gaussian state in terms of its parameters $U$, $\bm{\xi}$, and $\bm{\beta}$.
\begin{lem}[(Stellar function of a Gaussian state~\cite{chabaud2020stellar,chabaud2021holomorphic})]\label{lemma_stellar}
Let ${\psi}$ be an $n$-mode pure Gaussian state, as parametrized as in  \cref{lem_param_pure} in terms of a unitary matrix $U\in\mathbb{C}^{n\times n}$, a real vector $\bm{\xi}\in\mathbb{R}_+^{n}$, and a complex vector $\bm{\beta}\in\mathbb{C}^n$. Then, the stellar function of ${\psi}$ is given by
\be
    F^\star_{\ket \psi}(\bm z)=\frac1{\mathcal N}e^{\frac12\bm z^\intercal A\bm z+\bm{w}^\intercal \bm z+C}\qquad\forall \bm{z}\in\mathbb{C}^n,
\ee
where we defined
\begin{equation}\label{eq:stellarGmultinotations}
    \begin{aligned}
        A&\coloneqq U^\intercal \mathrm{Diag}\!\left(\tanh \xi_1,\ldots,\tanh \xi_n\right) U\,,\\
        \bm{w}&\coloneqq U^\intercal \bm b\,,\\
        C&\coloneqq \sum_{j=1}^nc_j\,,\\
        \mathcal N&\coloneqq\sqrt{\cosh \xi_1\cdots\cosh \xi_n}\,,       
    \end{aligned}
\end{equation}
and 
\be
    b_j&\coloneqq \frac{\beta_j}{\cosh\xi_j}\,,\\ 
    c_j&\coloneqq -\frac12(\tanh\xi_j)\beta_j^2-\frac12|\beta_j|^2 \,,
\ee
for all $j\in\{1,\dots,n\}$.
\end{lem}


\section{The symplectic rank}\label{app:symprank}
A key property of pure Gaussian states is that all the symplectic eigenvalues of their covariance matrices are equal to one~\cite{BUCCO}. Conversely, if a pure state has a symplectic eigenvalue strictly greater than one, it cannot be Gaussian~\cite{BUCCO}. This suggests that, for pure states, the number of symplectic eigenvalues strictly larger than one may play a role in quantifying the non-Gaussianity of the state. This intuition motivates the definition of \emph{symplectic rank}, which may be defined as the number of symplectic eigenvalues of the covariance matrix that are strictly greater than one. 

However, this definition would be ill-defined for states whose covariance matrices are ill-defined, such as those with infinite energy. In the forthcoming~\cref{charact_symp}, we establish a deeper connection between this definition --- valid only for states with a well-defined covariance matrix --- and the concept of \emph{$t$-compressibility} introduced in~\cite{mele2024learning}. The latter concept applies to all states, even those with an ill-defined covariance matrix. Consequently, in the forthcoming definition, we provide a general definition of symplectic rank based on the concept of $t$-compressibility. 

\begin{defi}[(Symplectic rank)]\label{defsmsymp}
The symplectic rank of an $n$-mode pure state ${\psi}$, denoted as $\mathfrak{s}(\psi)$, is defined as the minimum $t$ for which there exists a Gaussian unitary operation $G$ and a $t$-mode pure state ${\phi}$ such that
\be
    \ket{\psi}=G\left( \ket{\phi}\otimes \ket{0}^{\otimes ( n-t)}\right)\,.
\ee
Moreover, the symplectic rank of a mixed state $\rho$, denoted as $\mathfrak{s}(\rho)$, is defined as
\be\label{def_symprank_mixed_sm}
    \mathfrak{s}(\rho)\coloneqq \min_{\substack{\rho=\sum_{i}p_i\psi_i\\p_i>0}}\max_{i}\mathfrak{s}(\psi_i)\,,
\ee
where the minimum is taken over all the possible decompositions of $\rho$ as a convex combination of pure states.
\end{defi}

Note that we can define the symplectic rank of mixed states as a $\min\,\max$ rather than an $\inf\,\sup$, because the symplectic rank of an $n$-mode state is integer-valued and bounded between $0$ and $n$.

In other words, the symplectic rank of a pure state ${\psi}$ is the minimum $t$ for which the state is $t$-compressible~\cite{mele2024learning}. The forthcoming~\cref{charact_symp} gives a characterization of the symplectic rank of a pure state in terms of the symplectic eigenvalues of its covariance matrix~\cite{BUCCO}. In short, the symplectic rank of a pure state is the number of symplectic eigenvalues of its covariance matrix that are strictly larger than one. Moreover, this lemma establishes also that the compressed state $\phi$ in \cref{def_symprank_mixed_sm} can be chosen so that its mean photon number is bounded by the one of the original state $\psi$. Additionally, the Gaussian unitary operation $G$ in \cref{def_symprank_mixed_sm} can be chosen so that it maps the Gaussification $\rho_\psi$ of $\psi$ (i.e.~$\rho_\psi$ is the Gaussian state with same first moment and same covariance matrix as $\psi$) to its normal mode decomposition (see  \cref{lem_normal_mode}) as \be
    \rho_\psi=G\left(\tau\otimes \ketbra{0}^{\otimes (n-t)} \right)G^\dagger\,, 
\ee
where $\tau$ is a tensor product of suitable $t$ thermal states. We are now ready to present~\cref{charact_symp}.

\begin{theo}[(Characterization of the symplectic rank)]\label{charact_symp}
    Let $\psi$ be a pure state with a well-defined covariance matrix $V(\psi)$, i.e.~such that all the elements of $V(\psi)$ are finite. The following facts hold:
    \begin{itemize}
        \item The symplectic rank $\mathfrak{s}(\psi)$ is equal to the number of symplectic eigenvalues of $V(\psi)$ which are strictly larger than one.
        \item The state $\psi$ can be written as 
        \be\label{eq_compr_modified}
            \ket{\psi}=G\left( \ket{\phi}\otimes \ket{0}^{\otimes ( n-\mathfrak{s}(\psi))}\right)\,,
        \ee 
        where, by denoting as $V(\psi)=SDS^\intercal$ the Williamson decomposition of $V(\psi)$ as in Eq.~\eqref{will_eq_def}, we have that: (i) $\phi$ is a state over $\mathfrak{s}(\psi)$ modes with zero first moment and covariance matrix given by $V(\phi)=\diag\!\left(d_1,d_1,\ldots,d_{\mathfrak{s}(\psi)},d_{\mathfrak{s}(\psi)}\right)$, with $d_1,\ldots, d_{\mathfrak{s}(\psi)}$ being the symplectic eigenvalues of $V(\psi)$ that are strictly larger than one; and (ii) $G$ is a Gaussian unitary operation of the form $G={D}_{\bm{m}(\psi)}U_S$, with $U_S$ being the symplectic Gaussian unitary associated with $S$, and ${D}_{\bm{m}(\psi)}$ denoting the displacement operator with amplitude given by the first moment of $\psi$.
        \item The compressed state $\phi$ in Eq.~\eqref{eq_compr_modified} has mean total photon number upper bounded by the one of $\psi$:
        \be
            N(\phi)\le N(\psi)\,,
        \ee
        where $N(\ast)$ denotes the mean total photon number of the state $\ast$.
    \end{itemize}
\end{theo}

\begin{proof}
    The first and second claims are direct consequences of \cite[Lemma~S70]{mele2024learning}. Hence, it suffices to prove the third claim. To do this, let us first state the following result, which follows by taking $k=n$ and $M=\mathbb{1}$ in the statement of~\cite[Lemma~5.1]{Bhatia_2015} (see also \cite[Lemma~1]{Hiroshima_2006}): For any covariance matrix $V$ it holds that
    \be\label{bathiaaa}
        \Tr[D]\le \Tr[V]\,,
    \ee
    where $D$ denotes the matrix of symplectic eigenvalues of $V$ as in Eq.~\eqref{will_eq_def}. Hence, note that
    \be\label{eq:bound_energy_t_comp} N(\phi)&\eqt{(i)}\frac{\Tr[V(\phi)]-2\mathfrak{s}(\psi) }{4}+\frac{\|\bm{m}(\phi)\|^2}{2}\\
        &\eqt{(ii)}\frac{\Tr[D]-2n }{4}+\frac{\|\bm{m}(\phi)\|^2}{2}\\
        &= \frac{\Tr[D] -2n}{4}\\
        &\leqt{(iii)} \frac{\Tr[SDS^\intercal] -2n}{4}\\
        &=\frac{\Tr[V(\psi)]-2n }{4}\\
        &\le \frac{\Tr[V(\psi)] -2n}{4}+\frac{\|\bm{m}(\psi)\|^2}{2}\\
        &\eqt{(iv)} N(\psi)\,,
    \ee
    where in (i) we exploited the general formula for the mean photon number of a state in Eq.~\eqref{formula_mean_energy}; in (ii) we exploited that $\Tr[D]=2\sum_{i=1}^{\mathfrak{s}(\psi)}d_i+2(n-\mathfrak{s}(\psi))=\Tr[V(\phi)]+2(n-\mathfrak{s}(\psi))$; in (iii) we used Eq.~\eqref{bathiaaa}; and in (iv) we used again the general formula for the mean photon number of a state.
\end{proof}

Note that the only assumption of the above~\cref{charact_symp} is that the covariance matrix of the state is well defined. This assumption is always satisfied by states with finite mean photon number, as established in~\cref{well_defined_cov}. Consequently, for states with finite energy, the symplectic rank is always equal to the number of symplectic eigenvalues which are strictly larger than one.

Moreover, the symplectic rank of a state vanishes if and only if the state is a convex combination of Gaussian states, as proven in the following lemma.
\begin{lem}[(Faithfulness of the symplectic rank)]
    Let $\rho$ be a quantum state. It holds that \(\mathfrak{s}(\rho) = 0\) if and only if \(\rho\) is a convex combination of Gaussian states.
\end{lem}
\begin{proof}
       The non-trivial direction of the proof is to show that if \(\rho\) is a convex combination of Gaussian states, then \(\mathfrak{s}(\rho) = 0\). To establish this, it suffices to note that any convex combination of (possibly mixed) Gaussian states can always be expressed as a convex combination of pure Gaussian states. This follows from the fact that any mixed Gaussian state can be decomposed as a convex combination of pure Gaussian states~\cite[Lemma S4]{bittel2024optimalestimatestracedistance}. 
\end{proof}
Additionally, the forthcoming result establishes that the symplectic rank is additive under the tensor product of pure states and sub-additive under the tensor product of mixed states.
 \begin{lem}[(Sub-additivity under tensor product)]
     Let $\rho$ be a quantum state. Then, it holds that
     \be
        \mathfrak{s}(\rho\otimes\psi)&=\mathfrak{s}(\rho)+\mathfrak{s}(\psi)\,,\\
        \mathfrak{s}(\rho\otimes\sigma)&\le\mathfrak{s}(\rho)+\mathfrak{s}(\sigma)\,,
     \ee
     for all pure states $\psi$ and mixed states $\sigma$.
 \end{lem}
 \begin{proof}
    By exploiting the definition of symplectic rank for pure states, one can easily observe that $\mathfrak{s}(\phi\otimes\psi)=\mathfrak{s}(\phi)+\mathfrak{s}(\psi)$ for all pure states $\phi,\psi$. Consequently, the proof of $\mathfrak{s}(\rho\otimes\psi)=\mathfrak{s}(\rho)+\mathfrak{s}(\psi)$ follows from the fact that any decomposition of $\rho\otimes\psi$ as a convex combination of pure states is of the form $\rho\otimes\psi=\sum_ip_i\phi_i\otimes \psi$, with $\rho=\sum_ip_i\phi_i$.

    Let us now show that $\mathfrak{s}(\rho\otimes\sigma)\le\mathfrak{s}(\rho)+\mathfrak{s}(\sigma)$. Let us consider optimal decompositions of $\rho$ and $\sigma$ as convex combinations of pure states, i.e.~$\rho=\sum_ip_i\phi_i$ and $\sigma=\sum_iq_i\psi_i$ such that $\mathfrak{s}(\rho)=\max_i\mathfrak{s}(\phi_i)$ and $\mathfrak{s}(\sigma)=\max_i\mathfrak{s}(\psi_i)$. Then, $\sum_{i,j}p_iq_j\phi_i\otimes \psi_i$ constitutes a decomposition of $\rho\otimes \sigma$ as a convex combination of pure states, and hence
    \be
        \mathfrak{s}(\rho\otimes \sigma)\le \max_{i,j}\mathfrak{s}(\phi_i\otimes\psi_j)=\max_{i,j}[\mathfrak{s}(\phi_i)+\mathfrak{s}(\psi_j)]=\mathfrak{s}(\rho)+\mathfrak{s}(\sigma)\,,
    \ee
    which concludes the proof.
 \end{proof}
Most importantly, the symplectic rank is also a \emph{non-Gaussianity monotone}, as proven in the following section.


\subsection{The symplectic rank is non-increasing under post-selected Gaussian operations}
\label{app:monotonicity}

In this section, we first prove our main result: the symplectic rank is non-increasing under \emph{post-selected Gaussian operations}~\cite{giedke2002characterization}. Then, we show its fundamental consequences: (i) powerful no-go theorems on the task of state conversion under Gaussian protocols, and (ii) irreversibility of the resource theory of non-Gaussianity. 
\begin{defi}[(Post-selected Gaussian operation)]\label{def_post_sel}
    A post-selected Gaussian operation is a composition of the following five building blocks:
\begin{itemize}
    \item Tensoring with a Gaussian state;
    \item Applying a Gaussian unitary operation;
    \item Performing (non-Gaussian) classical mixing;
    \item Performing an heterodyne measurement, and post-selecting on the outcome;
    \item Taking a partial trace.
\end{itemize}
\end{defi}

Let us remark that an arbitrary composition of the above five operations include any protocol one can make with linear optics. For example, any Gaussian measurement can be implemented by performing a Gaussian unitary operation followed by a heterodyne detection \cite{giedke2002characterization,chabaud2020classical}. The ability to post-select on measurement outcomes also include (and subsumes) the ability to condition quantum operations on the outcomes of measurement (i.e.~feed-forward, or adaptive protocols). Moreover, note that any Gaussian channel can be implemented by partial traces, Gaussian unitary operations, and tensoring with the vacuum state via the Gaussian Stinespring dilation~\cite{BUCCO}.

Let us state our main result.

\begin{theo}[(Monotonicity of the symplectic rank)]\label{thm_main_sm}
    The symplectic rank is non-increasing under post-selected Gaussian operations: 
    \begin{itemize}
    \item \textbf{Tensoring with Gaussian states}: For any state $\rho$ and any Gaussian state $\sigma_G$, it holds that $\mathfrak{s}\!\left(\rho\otimes\sigma_G\right)\le \mathfrak{s}\!\left(\rho\right)$.
    \item \textbf{Gaussian unitary operations}: For any state $\rho$ and any Gaussian unitary operation $G$, it holds that $\mathfrak{s}\!\left(G\rho G^\dagger\right)\le \mathfrak{s}\!\left(\rho\right)$.
        \item \textbf{Classical mixing}: For any ensemble of states $\{p_i,\rho_i\}_i$ it holds that 
    \be
        \mathfrak{s}\left(\sum_ip_i\rho_i\right)\le\max_i\mathfrak{s}\left(\rho_i\right)\,.
    \ee
    \item \textbf{Heterodyne measurement with post-selection}: The symplectic rank of any post-outcome state, after an heterodyne measurement on a subset of modes, is not larger than the symplectic rank of the original state. That is, for any bipartite state $\rho_{AB}$ and any coherent state $\ket{\bm{\alpha}}_A$ on $A$ (where $\bm{\alpha}$ represents the outcome of the measurement performed on the system $A$), it holds that
    \be
        \mathfrak{s}\left( \frac{\bra{\bm{\alpha}}_A \rho_{AB} \ket{\bm{\alpha}}_A}{\Tr\bra{\bm{\alpha}}_A \rho_{AB}   \ket{\bm{\alpha}}_A} \right)\le \mathfrak{s}\left( \rho\right)\,.
    \ee
    \item \textbf{Partial traces}: For any bipartite state $\rho_{AB}$, it holds that $ \mathfrak{s}\!\left(\Tr_A\rho_{AB}\right)\le \mathfrak{s}\!\left(\rho_{AB}\right)$.
\end{itemize}
\end{theo}

Before proving \cref{thm_main_sm}, let us show some preliminary lemmas.

\begin{lem}[(Projecting a two-mode pure Gaussian state on a single-mode vacuum state)]\label{lem_LL}
    Let $G_{12}$ be a two-mode Gaussian unitary such that 
    \be
     \frac{\bra{0}_1G_{12}\ket{0}^{\otimes 2}}{\left\| \bra{0}_1G_{12}\ket{0}^{\otimes 2}  \right\|}=\ket{0}_2\,.
     \ee
     Then, there exists $\bm{y}\in\mathbb{C}^{2}$ such that
    \be
     \frac{\bra{0}_1G_{12}\ket{\bm{\alpha}}}{\left\|\bra{0}_1G_{12}\ket{\bm{\alpha}}\right\|}=e^{i\theta(\bm{\alpha})}\ket{\bm{y}^\intercal \bm{\alpha}}\qquad\forall\,\bm{\alpha}\in\mathbb{C}^2\,,
    \ee
where $e^{i\theta(\bm{\alpha})}$ is a suitable phase.
\end{lem}

\begin{proof}
    By exploiting Eq.~\eqref{eq_euler}, we can write
    \be\label{eq_g12prof}
        G_{12}=G_US(\bm\xi)D(\bm\beta)G_V\,,
    \ee
    where $G_U,G_V$ are passive Gaussian unitary operations associated with suitable unitary matrices $U,V\in\mathbb{C}^{2\times 2}$, $ S(\bm\xi)$ is the $2$-mode squeezing unitary of squeezing vector $\bm\xi\in\mathbb{R}^{2}$, and $ D(\bm\beta)$ is the displacement operator of displacement vector $\bm\beta\in\mathbb{C}^2$~\cite{BUCCO,chabaud2021holomorphic}. 
    Then, it holds that
    \be\label{eq_g12def}
    G_{12}\ket{\bm{\alpha}}&=G_US(\bm\xi)D(\bm\beta)G_VD(\bm\alpha)\ket{0}^{\otimes 2}\\
    &=G_US(\bm\xi)D(\bm\beta)D(V^\intercal\bm\alpha)\ket{0}^{\otimes 2}\\
    &=e^{i\phi(\bm{\alpha})} G_US(\bm\xi)D(\bm\beta+V^\intercal\bm\alpha)\ket{0}^{\otimes 2}\,,
    \ee
    where in the second line we used Eq.~\eqref{transf_disp_st}, and in the last line we exploited the composition rule between displacement operators~\cite{BUCCO}, with $e^{i\phi(\bm{\alpha})} $ being a suitable phase. Thanks to  \cref{lemma_stellar}, the stellar function of $G_US(\bm\xi)D(\bm\beta+V^\intercal\bm\alpha)\ket{0}^{\otimes 2}$ reads:
\be
    F^\star_{G_US(\bm\xi)D(\bm\beta+V^\intercal\bm\alpha)\ket{0}^{\otimes 2}}(\bm z)=\frac1{\mathcal N}e^{\frac12\bm z^\intercal A\bm z+\bm{w}^\intercal \bm z+\bm{v}^\intercal \bm z+\tilde{C}}\qquad\forall \bm{z}\in\mathbb{C}^2,
\ee
where we defined
\begin{equation}\label{eq:stellarGmultinotations2}
    \begin{aligned}
        A&\coloneqq U^\intercal \mathrm{Diag}\!\left(\tanh \xi_1,\tanh \xi_2\right) U\,,\\
        \bm{w}&\coloneqq U^\intercal \bm b\,,\\
        \bm{v}&\coloneqq U^\intercal \bm a\,,\\
        \tilde{C}&\coloneqq \tilde{c}_1+\tilde{c}_2\,\\
        \mathcal N&\coloneqq\sqrt{\cosh \xi_1\cosh\xi_2}\,,       
    \end{aligned}
\end{equation}
and 
\be
    b_j&\coloneqq \frac{\beta_j}{\cosh\xi_j}\,,\\ 
    a_j&\coloneqq \frac{\left(V^\intercal \bm{\alpha}\right)_j}{\cosh\xi_j}\,,\\
    \tilde{c}_j&\coloneqq -\frac12(\tanh\xi_j)\left(\beta_j+\left(V^\intercal \bm{\alpha}\right)_j\right)^2-\frac12|\beta_j+\left(V^\intercal \bm{\alpha}\right)_j|^2 \,,
\ee
for all $j\in\{1,2\}$. Consequently, by defining    
\be
    C&\coloneqq c_1+c_2\,,\\
    c_j&\coloneqq -\frac12(\tanh\xi_j)\beta_j^2-\frac12|\beta_j|^2 \qquad j\in\{1,2\}\,,
\ee
we obtain that
\be
    F^\star_{G_US(\bm\xi)D(\bm\beta+V^\intercal\bm\alpha)\ket{0}^{\otimes 2}}(\bm z)&=e^{\tilde{C}-C+\bm{v}^\intercal \bm z}\frac1{\mathcal N}e^{\frac12\bm z^\intercal A\bm z+\bm{w}^\intercal \bm z+C}\\
    &=e^{\tilde{C}-C+\bm{v}^\intercal \bm z}F^\star_{G_US(\bm\xi)D(\bm\beta)\ket{0}^{\otimes 2}}(\bm z)\,,
\ee
Hence, it holds that
\be\label{eq_crucial}
    F^\star_{\bra{0}_1G_US(\bm\xi)D(\bm\beta+V^\intercal\bm\alpha)\ket{0}^{\otimes 2}}(z)&=F^\star_{G_US(\bm\xi)D(\bm\beta+V^\intercal\bm\alpha)\ket{0}^{\otimes 2}}(0,z)\\
    &=e^{\tilde{C}-C+v_2 z }F^\star_{G_US(\bm\xi)D(\bm\beta)\ket{0}^{\otimes 2}}(0,z)\\
    &=e^{\tilde{C}-C+v_2 z }F^\star_{\bra{0}_1G_US(\bm\xi)D(\bm\beta)\ket{0}^{\otimes 2}}(z)\\
    &\eqt{(i)}\left\|\bra{0}_1 G_{12}\ket{0}^{\otimes 2}\right\|e^{\tilde{C}-C+v_2 z }F^\star_{\ket{0}}(z)\\
    &\eqt{(ii)}\left\|\bra{0}_1 G_{12}\ket{0}^{\otimes 2}\right\|e^{\tilde{C}-C+v_2 z }\\
    &\eqt{(iii)}\left\|\bra{0}_1 G_{12}\ket{0}^{\otimes 2}\right\|e^{\tilde{C}-C-C'} F^\star_{\ket{v_2}}(z) \,.
\ee
Here, in (i), we used Eq.~\eqref{eq_g12prof} and the fact that passive Gaussian unitary operations preserve the vacuum to conclude that 
\be
\bra{0}_1G_US(\bm\xi)D(\bm\beta)\ket{0}^{\otimes 2}=\bra{0}_1G_US(\bm\xi)D(\bm\beta)G_V\ket{0}^{\otimes 2}=\bra{0}_1 G_{12}\ket{0}^{\otimes 2}=\left\|\bra{0}_1 G_{12}\ket{0}^{\otimes 2}\right\|\ket{0}\,;
\ee 
in (ii), we exploited that the stellar function of the vacuum is one; in (iii), we identified the stellar function of the coherent state $\ket{v_2}$ by defining $C'\coloneqq -\frac12|v_2|^2$. Taking advantage of the normalization condition of the stellar function in Eq.~\eqref{eq_normalizing}, Eq.~\eqref{eq_crucial} implies that
\be
    \frac{\bra{0}_1G_US(\bm\xi)D(\bm\beta+V^\intercal\bm\alpha)\ket{0}^{\otimes 2}}{\left\|\bra{0}_1G_US(\bm\xi)D(\bm\beta+V^\intercal\bm\alpha)\ket{0}^{\otimes 2}\right\|}=e^{i\varphi(\bm{\alpha})}\ket{v_2}\,,
\ee
where $e^{i\varphi(\bm{\alpha})}$ is a suitable phase. By exploiting the definition of $v_2$ in Eq.~\eqref{eq:stellarGmultinotations2}, we have that
\be
    v_2= U_{12}\frac{\left(V^\intercal \bm{\alpha}\right)_1}{\cosh\xi_1}+U_{22}\frac{\left(V^\intercal \bm{\alpha}\right)_2}{\cosh\xi_2}\,.
\ee
Hence, we can identify a complex vector $\bm{y}\in\mathbb{C}^2$ such that $v_2=\bm{y}^\intercal \bm{\alpha}$, so that
\be
    \frac{\bra{0}_1G_US(\bm\xi)D(\bm\beta+V^\intercal\bm\alpha)\ket{0}^{\otimes 2}}{\left\|\bra{0}_1G_US(\bm\xi)D(\bm\beta+V^\intercal\bm\alpha)\ket{0}^{\otimes 2}\right\|}=e^{i\varphi(\bm{\alpha})}\ket{\bm{y}^\intercal \bm{\alpha}}\,.
\ee
By substituting in Eq.~\eqref{eq_g12def}, we obtain that
\be
    \frac{\bra{0}_1G_{12}\ket{\bm{\alpha}}}{\left\|\bra{0}_1G_{12}\ket{\bm{\alpha}}\right\|}=e^{i\theta(\bm{\alpha})}\ket{\bm{y}^\intercal \bm{\alpha}}\,,
\ee
where $e^{i\theta(\bm{\alpha})}$ is a suitable phase.
\end{proof}

Now, we introduce a general decomposition for Gaussian unitary operations, which is crucial for the proof of the monotonicity of the symplectic rank. See \cref{fig:circuit} for a schematic representation of such a decomposition.

\begin{lem}[(Block decomposition of Gaussian unitary operations)]\label{lem:decomp}
    Let $G$ be an $n$-mode Gaussian unitary. Moreover, let $k\le n/2$ be an integer. Then, there exists an $n$-mode passive Gaussian unitary $G_p$, a $(2k)$-mode Gaussian unitary ${G^{(2k)}}$ acting only on the first $2k$ modes, and an $(n-k)$-mode Gaussian unitary $G_{(n-k)}$ acting only on the last $n-k$ modes such that
    \be\label{dec_gauss_g}
        G=\left( \mathbb{1}_{k}\otimes G_{(n-k)}\right) \left(  {G^{(2k)}}\otimes\mathbb{1}_{n-2k}\right) G_p\,.
    \ee
    Moreover, ${G^{(2k)}}$ satisfies 
    \be\label{eq_add_req}
        \frac{\bra{0}^{\otimes k}{G^{(2k)}}\ket{0}^{\otimes 2k}}{ \left\|\bra{0}^{\otimes k}{G^{(2k)}}\ket{0}^{\otimes 2k}\right\|}=\ket{0}^{\otimes k}\,.
    \ee
\end{lem}

Here, we provide two different proofs of the decomposition in Eq.~\eqref{dec_gauss_g}. The first proof is conceptually simpler, but does not allow us to prove the additional requirement in Eq.~\eqref{eq_add_req}. In contrast, the second proof is more technical, but it allows one to completely prove the above  \cref{lem:decomp} by establishing both the decomposition in Eq.~\eqref{dec_gauss_g} and the additional requirement in Eq.~\eqref{eq_add_req}.
\begin{proof}[Simple proof of Eq.~\eqref{dec_gauss_g}]
    We follow a similar reasoning as in~\cite[Lemma~10]{LL-log-det}. By Schmidt decomposition, two complementary reduced states of a bipartite pure state have the same non-zero spectrum. Moreover, the spectrum of a $k$-mode Gaussian state is given by the spectrum of the tensor product of $k$ suitable thermal states, thanks to the normal mode decomposition in  \cref{lem_normal_mode}. Hence, since $n-k\ge k$, the reduced state on the last $n-k$ modes of the pure Gaussian state $G\ket{0}^{\otimes n}$ is unitarily equivalent --- through a suitable Gaussian unitary $(n-k)$ mode Gaussian unitary $G_{(n-k)}$ --- to the tensor product of $k$ suitable thermal states and the $(n-2k)$-mode vacuum state. Consequently, the state $G_{(n-k)}^\dagger G\ket{0}^{\otimes{n}}$ is of the form
    \be
        G_{(n-k)}^\dagger G\ket{0}^{\otimes{n}}=\ket{\phi_G}\otimes\ket{0}^{\otimes (n-2k)}\,,
    \ee
    where ${\phi_G}$ is a suitable $2k$-mode Gaussian pure state. In particular, there exists a suitable $2k$-mode Gaussian unitary $G^{(2k)}$ such that $\ket{\phi_G}=G^{(2k)}\ket{0}^{\otimes 2k}$. Hence, it follows
    \be
        (G^{(2k)})^\dagger G_{(n-k)}^\dagger G\ket{0}^{\otimes{n}}=\ket{0}^{\otimes n}\,,
    \ee
    meaning that the Gaussian unitary operation $(G^{(2k)})^\dagger G_{(n-k)}^\dagger G$ preserves the vacuum state. Moreover, any Gaussian unitary operation that preserves the vacuum must be passive, as it easily follows from Eq.~\eqref{action_gauss_moment}. Consequently, $(G^{(2k)})^\dagger G_{(n-k)}^\dagger G$ is equal to a suitable passive Gaussian unitary $G_p$, which implies 
    \be
        G=\left(\mathbb{1}_{k}\otimes{G}_{(n-k)}\right)\left({G^{(2k)}}\otimes \mathbb{1}_{n-2k} \right)G_p\,.
    \ee
\end{proof}

\begin{proof}[Proof of  \cref{lem:decomp}]
Since the state proportional to $\bra{0}^{\otimes k} G\ket{0}^{\otimes n}$ is a $(n-k)$-mode Gaussian state, it follows that there exists an $(n-k)$-mode Gaussian unitary $\tilde{G}_{(n-k)}$ such that
    \be\label{eqeqeq}
        \frac{\bra{0}^{\otimes k} G\ket{0}^{\otimes n}}{\mathcal N}= \tilde{G}_{(n-k)}\ket{0}^{\otimes (n-k)}\,,
    \ee
where we defined 
\be
    \mathcal N\coloneqq \left\|\bra{0}^{\otimes k} G\ket{0}^{\otimes n}\right\|\,.
\ee
    In particular, we have that
    \be\label{eq_3_prova}
        \frac{\bra{0}^{\otimes k} \left(\mathbb{1}_k\otimes\tilde{G}_{(n-k)}^\dagger\right) G\ket{0}^{\otimes n}}{\mathcal N}= \ket{0}^{\otimes (n-k)}\,.
    \ee    
Hence, at the level of stellar functions, we have that
\be
    F^\star_{\left(\mathbb{1}_k\otimes\tilde{G}_{(n-k)}^\dagger\right) G\ket{0}^{\otimes n}}(0,\ldots,0,z_{k+1},\dots,z_n)= \mathcal N F^\star_{\ket{0}^{\otimes (n-k)}}(z_{k+1},\dots,z_n)=\mathcal N
\ee
for any $z_{k+1},\dots,z_n\in\mathbb{C}$. Consequently,  \cref{lemma_stellar} implies that there exist two matrices $A_1\in\mathbb{C}^{k\times k}$, $A_2\in\mathbb{C}^{k\times(n-k)}$ and a vector $\bm{w}\in\mathbb{C}^{k}$ such that
\be
    F^\star_{\left(\mathbb{1}_k\otimes\tilde{G}_{(n-k)}^\dagger\right) G\ket{0}^{\otimes n}}(\bm{z})=\mathcal N\exp\left( \frac12(\bm{z}_A)^\intercal A_1\bm{z}_A+ (\bm{z}_A)^\intercal A_2\bm{z}_B+\bm{w}^\intercal \bm{z}_A    \right)\qquad\forall\,\bm{z}\in\mathbb{C}^n\,,
\ee
where $\bm{z}_A\coloneqq (z_1,\ldots , z_k)$ and $\bm{z}_B\coloneqq (z_{k+1},\ldots , z_n)$, so that $\bm{z}=(\bm{z}_A,\bm{z}_B)$. By exploiting the singular value decomposition of $A_2$, there exist unitary matrices $U_1\in\mathbb{C}^{k\times k}$, $U_2\in\mathbb{C}^{(n-k)\times (n-k)}$, and a rectangular matrix $\Sigma\in\mathbb{R}^{k\times (n-k)}$ having non-zero elements only on the main diagonal, so that $A_2=U_1\Sigma U_2$. Moreover, let $\bar{\Sigma}$ the $k\times k$ diagonal matrix such that $\Sigma=(\bar{\Sigma} \quad 0_{k\times(n-2k)})$, where $0_{k\times(n-2k)}$ denotes the $k\times(n-2k)$ zero matrix. Hence, by exploiting Eq.~\eqref{transf_st}, it follows that for any $\bm{z}\in\mathbb{C}^n$ it holds that
\be
    F^\star_{(\mathbb{1}_k\otimes G_{U_2}^\dagger) \left(\mathbb{1}_k\otimes\tilde{G}_{(n-k)}^\dagger\right) G\ket{0}^{\otimes n}}(\bm{z})&=\mathcal N\exp\left( \frac12(\bm{z}_A)^\intercal A_1\bm{z}_A+ (\bm{z}_A)^\intercal U_1\Sigma\bm{z}_B+\bm{w}^\intercal \bm{z}_A    \right)\,,\\
    &=\mathcal N\exp\left( \frac12(\bm{z}_A)^\intercal A_1\bm{z}_A+ (\bm{z}_A)^\intercal U_1\bar{\Sigma}\bm{z}_C+\bm{w}^\intercal \bm{z}_A    \right)\,,\\
\ee
where we defined $\bm{z}_C=(z_1,\ldots, z_{2k})$. This implies that there exists a $2k$-mode Gaussian unitary $G^{(2k)}$ such that
\be\label{eq_def_G2K}
    (\mathbb{1}_k\otimes G_{U_2}^\dagger) \left(\mathbb{1}_k\otimes\tilde{G}_{(n-k)}^\dagger\right) G\ket{0}^{\otimes n}= (G^{(2k)}\otimes\mathbb{1}_{n-2k})\ket{0}^{\otimes n}\,.
\ee
Consequently, since any Gaussian unitary operation that preserves the vacuum state must be passive (this can be proven
by considering the action of such a transformation at the level of covariance matrices), it follows that there exists a passive Gaussian
unitary $G_p$ such that
\be
        \left({G^{(2k)}}^\dagger\otimes\mathbb{1}_{n-2k}\right)  \left(\mathbb{1}_k\otimes G_{U_2}^\dagger \tilde{G}_{(n-k)}^\dagger\right) G=G_p\,,
\ee
 or, equivalently,
\be\label{eq_g_the}
    G=\left(\mathbb{1}_{k}\otimes{G}_{(n-k)}\right)\left({G^{(2k)}}\otimes \mathbb{1}_{n-2k} \right)G_p\,,
\ee
where we defined ${G}_{(n-k)}\coloneqq \tilde{G}_{(n-k)}G_{U_2}$. This proves the first part of the thesis. 

Moreover, by denoting as $\bra{0}_{[2k+1,n]}$ the vacuum state in the modes $2k+1,2k+2,\ldots, n$, we have that
\be
        \frac{\bra{0}^{\otimes k}{G^{(2k)}}\ket{0}^{\otimes 2k}}{ \left\|\bra{0}^{\otimes k}{G^{(2k)}}\ket{0}^{\otimes 2k}\right\|}&\eqt{(i)}\bra{0}_{[2k+1,n]}\frac{\bra{0}^{\otimes k}\left({G^{(2k)}}\otimes\mathbb{1}_{n-2k}\right)\ket{0}^{\otimes n}}{\left\|\bra{0}^{\otimes k} G\ket{0}^{\otimes n}\right\|}\\
        &\eqt{(ii)} \bra{0}_{[2k+1,n]}\frac{\bra{0}^{\otimes k}  \left(\mathbb{1}_k\otimes G_{U_2}^\dagger \tilde{G}_{(n-k)}^\dagger\right) G\ket{0}^{\otimes n}}{\left\|\bra{0}^{\otimes k} G\ket{0}^{\otimes n}\right\|}\\
        &=\bra{0}_{[2k+1,n]}G_{U_2}^\dagger \tilde{G}_{(n-k)}^\dagger \frac{\bra{0}^{\otimes k}G\ket{0}^{\otimes n}}{\left\|\bra{0}^{\otimes k} G\ket{0}^{\otimes n}\right\|}\\
        &\eqt{(iii)} \bra{0}_{[2k+1,n]}G_{U_2}^\dagger \tilde{G}_{(n-k)}^\dagger   \tilde{G}_{(n-k)}\ket{0}^{\otimes (n-k)}\\
        &= \bra{0}_{[2k+1,n]}G_{U_2}^\dagger\ket{0}^{\otimes (n-k)}\\
        &\eqt{(iv)} \bra{0}_{[2k+1,n]}\ket{0}^{\otimes (n-k)}\\
        &= \ket{0}^{\otimes k}\,,
    \ee
    where in (i) we used Eq.~\eqref{eq_g_the} and the fact that $G_{p}$ is passive; in (ii) we exploited Eq.~\eqref{eq_def_G2K}, (iii) follows from Eq.~\eqref{eqeqeq}, and in (iv) we exploited the fact that $G_{U_2}^\dagger$ is passive. 
\end{proof}

The decomposition proven in  \cref{lem:decomp} leads to the following remarkable consequence.

\begin{coro}
Let $m$ an even natural number. Any bipartite (possibly mixed) Gaussian state over $m+m$ modes can be disentangled by applying a suitable $(\frac{3}{2}m)$-mode Gaussian unitary. 
\end{coro}
\begin{proof}
Thanks to the normal mode decomposition in  \cref{lem_normal_mode}, any $(2m)$-mode (possibly mixed) Gaussian state $\rho$ is unitarily equivalent --- via a Gaussian unitary operation $G$ --- to a state $\tau$ which is a tensor product of $2m$ thermal states:
\be
    \rho=G \tau G^\dagger\,.
\ee
Moreover, it is well known that a multi-mode thermal state can be written as a convex combination of coherent states~\cite{BUCCO}. That is, there exists a (Gaussian) probability distribution $P(\cdot)$ over $\mathbb{C}^{2m}$ such that
\be
    \tau=\int_{\C^{2m}}\mathrm{d}^{2m}\bm{z}\, P(\bm{z})\ketbra{\bm{z}}\,.
\ee
Consequently, we have that
\be
    \rho=\int_{\C^{2m}}\mathrm{d}^{2m}\bm{z}\, P(\bm{z}) G\ketbra{\bm{z}} G^\dagger\,.
\ee
By using  \cref{lem:decomp} with $n=2m$ and $k=\frac{m}{2}$ , we can decompose $G$ as in Eq.~\eqref{dec_gauss_g} to deduce that
\be
 \rho=\int_{\C^{2m}}\mathrm{d}^{2m}\bm{z}\, P(\bm{z}) \left(\mathbb{1}_{\frac{m}{2}}\otimes{G}_{(\frac32m)}\right)\left({G^{(m)}}\otimes \mathbb{1}_{m} \right)G_p\ketbra{\bm{z}} G_p^\dagger\left({G^{(m)}}\otimes \mathbb{1}_{m} \right)^\dagger \left(\mathbb{1}_{\frac{m}{2}}\otimes{G}_{(\frac32m)}\right)^\dagger\,,
\ee
where ${G}_{(\frac32m)}$ and ${G^{(m)}}$ are suitable $(\frac32m)$-mode and $m$-mode Gaussian unitary operations, respectively, and $G_p$ is a suitable passive Gaussian unitary operation. By using Eq.~\eqref{transf_disp_st}, it easily follows that coherent states are mapped into coherent states via passive Gaussian unitary operations. Specifically, there exists a unitary matrix $U$ such that $G_p\ket{\bm{z}}=\ket{U\bm{z}}$ for all $\bm z\in\mathbb{C}^{2m}$ and, in particular, it holds that
\be
 \left(\mathbb{1}_{\frac{m}{2}}\otimes{G}_{(\frac32m)}\right)^\dagger \rho \left(\mathbb{1}_{\frac{m}{2}}\otimes{G}_{(\frac32m)}\right)&=\int_{\C^{2m}}\mathrm{d}^{2m}\bm{z}\, P(\bm{z}) \left({G^{(m)}}\otimes \mathbb{1}_{m} \right)\ketbra{U\bm{z}} \left({G^{(m)}}\otimes \mathbb{1}_{m} \right)^\dagger \,.
\ee
Since coherent states are product states, it follows that for all $\bm{z}\in\mathbb{C}^{2m}$ the state $\left({G^{(m)}}\otimes \mathbb{1}_{m} \right)\ket{U\bm{z}}$ is product among the bipartition of the first $m$ modes and the last $m$ modes. Consequently, we have proven that the state
    \be
         \left(\mathbb{1}_{\frac{m}{2}}\otimes{G}_{(\frac32m)}\right)^\dagger \rho \left(\mathbb{1}_{\frac{m}{2}}\otimes{G}_{(\frac32m)}\right)
    \ee
can be written as a convex combination of product states over $m+m$ modes, i.e.~the state is separable, which concludes the proof.
\end{proof}

Remarkably, this corollary allows one to \emph{gain $25\%$ of modes}: although a priori one would use a $(2m)$-mode unitary operation to disentangle a bipartite entangled Gaussian state over $m+m$ modes, our result establishes that just $\frac{3}{2}m$ modes are sufficient. 

We are now ready to prove \cref{thm_main_sm}.

\begin{proof}[Proof of \cref{thm_main_sm}]
In order to prove that the symplectic rank is non-increasing under Gaussian operations, it suffices to prove that is non-increasing under under each of the five operations listed in the statement of the theorem.

\medskip

\parhead{Tensoring with a Gaussian state.} 
Recall that the symplectic rank is sub-additive with respect to the tensor product, as it easily follows from the definition of symplectic rank.
Hence, since the symplectic rank of Gaussian states is zero, it follows that the symplectic rank is non-increasing when tensoring with a Gaussian state.   

\medskip

\parhead{Applying a Gaussian unitary operation.} 
Observe that any $t$-compressible state remains $t$-compressible under Gaussian unitary operations. Since by \cref{charact_symp} the symplectic rank of a pure state is the smallest $t$ for which the state is $t$-compressible, it follows that it is non-increasing under Gaussian unitary operations (and in fact, invariant by considering the inverse Gaussian unitary operation). As a consequence, the same statement is immediate for mixed states.

\medskip

\parhead{Classical mixing.} 
Recall that the symplectic rank of a mixed state $\rho$ is defined as in Eq.~\eqref{def_symprank_mixed_sm}. Given two states with optimal decompositions $\rho=\sum_ip_i\ket{\psi_i}\!\bra{\psi_i}$ and $\sigma=\sum_jq_j\ket{\phi_j}\!\bra{\phi_j}$, the symplectic rank of their mixture thus is directly less than or equal to the maximal symplectic rank of each state. Hence, the symplectic rank is non-increasing under classical mixing.

\medskip

\parhead{Heterodyne measurement with post-selection.} Let $\rho_{AB}$ be a bipartite state and let $\ket{\bm{\alpha}}_A$ be a coherent state on $A$ of amplitude $\bm{\alpha}$. We need to show that
    \be\label{eq_to_prove_monot}
        \mathfrak{s}\!\left( \frac{\bra{\bm{\alpha}}_A \rho_{AB} \ket{\bm{\alpha}}_A}{\Tr\bra{\bm{\alpha}}_A \rho_{AB}   \ket{\bm{\alpha}}_A} \right)\le \mathfrak{s}( \rho_{AB})\,.
    \ee
Note that it suffices to prove the statement for the particular case in which $\rho_{AB}$ is pure. Indeed, assume that the statement holds for pure states, and let $\rho_{AB}$ be a mixed state. Consider the optimal decomposition $\rho_{AB}=\sum_ip_i\psi_i$ saturating the definition of symplectic rank, i.e.~such that
\be
    \mathfrak{s}(\rho_{AB})=\max_i\,\mathfrak{s}(\psi_i)\,.
\ee
Then, we have that
    \be
        \mathfrak{s}\!\left( \frac{\bra{\bm{\alpha}}_A \rho_{AB} \ket{\bm{\alpha}}_A}{\Tr\bra{\bm{\alpha}}_A \rho_{AB}   \ket{\bm{\alpha}}_A} \right)&=\mathfrak{s}\!\left(\sum_ip_i\frac{\Tr\bra{\bm{\alpha}}_A \psi_i \ket{\bm{\alpha}}_A }{\Tr\bra{\bm{\alpha}}_A \rho_{AB}   \ket{\bm{\alpha}}_A} \frac{\bra{\bm{\alpha}}_A \psi_i \ket{\bm{\alpha}}_A}{ \Tr\bra{\bm{\alpha}}_A \psi_i \ket{\bm{\alpha}}_A  } \right)\\
        &\le\max_i\,\mathfrak{s}\!\left(\frac{\bra{\bm{\alpha}}_A \psi_i \ket{\bm{\alpha}}_A}{ \Tr\bra{\bm{\alpha}}_A \psi_i \ket{\bm{\alpha}}_A  }\right)\\
        &\le \max_i\mathfrak{s}(\psi_i)\\
        &=\mathfrak{s}(\rho_{AB})\,.
    \ee
Here, in the first inequality we exploited the definition of symplectic rank for mixed states, together with the facts that $\left\{ p_i\frac{\bra{\bm{\alpha}}_A \psi_i \ket{\bm{\alpha}}_A }{\Tr\bra{\bm{\alpha}}_A \rho_{AB}   \ket{\bm{\alpha}}_A}   \right\}_i$ is a probability distribution and $\left\{ \frac{\bra{\bm{\alpha}}_A \psi_i \ket{\bm{\alpha}}_A}{ \Tr\bra{\bm{\alpha}}_A \psi_i \ket{\bm{\alpha}}_A  } \right\}_i$ are pure states; and, in the second inequality, we exploited that Eq.~\eqref{eq_to_prove_monot} holds for pure states.

Hence, it suffices to show that
    \be 
        \mathfrak{s}\!\left( \frac{\bra{\bm{\alpha}}_A \psi_{AB} \ket{\bm{\alpha}}_A}{\Tr\bra{\bm{\alpha}}_A \psi_{AB}   \ket{\bm{\alpha}}_A} \right)\le \mathfrak{s}( \psi_{AB})
    \ee
for any pure state $\psi_{AB}$ and any coherent state $\ket{\bm{\alpha}}_A$. Moreover, note that it suffices to prove the statement in the particular case in which $A$ is a single-mode subsystem. Furthermore, since $\ket{\bm{\alpha}}_A=D(\bm{\alpha})\ket{0}_A$ and since the symplectic rank is invariant under Gaussian unitary operations, it suffices to prove the statement for $\bm{\alpha}=0$. We give a proof hereafter.

Let ${\psi}$ an $n$-mode state with symplectic rank $s$. Then, by definition, there exists a Gaussian unitary operation $G$ and an $s$-mode state ${\phi}$ such that
    \be
        \ket{\psi}=G\left( \ket{\phi}\otimes\ket{0}^{\otimes (n-s)}\right)\,.
    \ee
In order to prove that the symplectic rank is non-increasing under heterodyne measurement, it suffices to show that there exists a Gaussian unitary operation $G'$ and an $s$-mode state ${\phi'}$ such that
    \be
        \frac{\bra{0}_1\ket{\psi}}{\left\|\bra{0}_1\ket{\psi}\right\|}=G'\left( \ket{\phi'}\otimes \ket{0}^{\otimes (n-s-1)}  \right)\,.
    \ee
Note that the statement is trivial when $s\in\{0,n-1,n\}$ (since the symplectic rank is always less or equal to the number of modes). Hence, let us assume that $s\notin\{0,n-1,n\}$.

Using  \cref{lem:decomp} with $k=1$, it follows that there exists a passive Gaussian unitary operation $G_p$, a two-mode Gaussian unitary operation $G^{(2)}$ acting on the first two modes, and an $(n-1)$-mode Gaussian unitary operation acting on the last $(n-1)$ modes such that
    \be
        G=\left( \mathbb{1}_{1}\otimes G_{(n-1)}\right) \left(  G^{(2)}\otimes\mathbb{1}_{n-2}\right) G_p\,,
    \ee
where $G^{(2)}$ satisfies $\bra{0}_1G^{(2)}\ket{0}^{\otimes 2}\propto\ket{0}_2$. Let $U$ be the $2\times 2$ unitary matrix associated with the passive Gaussian unitary $G_p$ such that the action of $G_p$ on a coherent state $\ket{\bm{k}}$ satisfies 
    \be\label{eq_action_gp}
    G_p\ket{\bm{k}}=\ket{U\bm{k}}\qquad \forall\bm{k}\in\mathbb{C}^{n}\,.
    \ee
Then, we have that
    \be
        \frac{\bra{0}_1\ket{\psi}}{\|\bra{0}_1\ket{\psi}\|}&= \frac{1}{\|\bra{0}_1\ket{\psi}\|}\bra{0}_1 G\left( \ket{\phi}\otimes\ket{0}^{\otimes (n-s)}\right)\\
        &= \frac{1}{\|\bra{0}_1\ket{\psi}\|}G_{(n-1)} \bra{0}_1 \left(  G^{(2)}\otimes\mathbb{1}_{n-2}\right) G_p\left( \ket{\phi}\otimes\ket{0}^{\otimes (n-s)}\right)\\
        &\eqt{(i)}G_{(n-1)}\int_{\mathbb{C}^{s}}\frac{\mathrm{d}^{s}\bm{\alpha}}{{\pi}^s} \frac{\braket{\bm{\alpha}|\phi}}{\|\bra{0}_1\ket{\psi}\|}  \bra{0}_1 \left(  G^{(2)}\otimes\mathbb{1}_{n-2}\right) G_p\left( \ket{\bm{\alpha}}\otimes\ket{0}^{\otimes (n-s)}\right)\\
        &\eqt{(ii)}G_{(n-1)}\int_{\mathbb{C}^{s}}\frac{\mathrm{d}^{s}\bm{\alpha}}{{\pi}^s} \frac{\braket{\bm{\alpha}|\phi}}{\|\bra{0}_1\ket{\psi}\|}  \bra{0}_1 \left(  G^{(2)}\otimes\mathbb{1}_{n-2}\right) \bigotimes_{i=1}^n\left\lvert \sum_{j=1}^{s} U_{ij}\alpha_j\right\rangle_{\!\!i}\\
        &=G_{(n-1)}\int_{\mathbb{C}^{s}}\frac{\mathrm{d}^{s}\bm{\alpha}}{{\pi}^s} \frac{\braket{\bm{\alpha}|\phi}}{\|\bra{0}_1\ket{\psi}\|}  \left( \bra{0}_1 G^{(2)} \left| \sum_{j=1}^{s} U_{1j}\alpha_j\right\rangle_{\!\!1}\otimes  \left| \sum_{j=1}^{s} U_{2j}\alpha_j\right\rangle_{\!\!2}\right)\bigotimes_{i=3}^n\left|{ \sum_{j=1}^{s} U_{ij}\alpha_j}\right\rangle_{\!\!i}\\
        &\eqt{(iii)}G_{(n-1)}\int_{\mathbb{C}^{s}}\frac{\mathrm{d}^{s}\bm{\alpha}}{{\pi}^s} \frac{\braket{\bm{\alpha}|\phi}}{\|\bra{0}_1\ket{\psi}\|}  e^{i\phi(\bm{\alpha})}\mathcal N(\bm{\alpha})\left|w_1\sum_{j=1}^{s} U_{1j}\alpha_j +w_2\sum_{j=1}^{s} U_{2j}\alpha_j  \right\rangle_{\!\!2}\bigotimes_{i=3}^n\left| \sum_{j=1}^{s} U_{ij}\alpha_j\right\rangle_{\!\!i}\\
        &\eqt{(iv)}G_{(n-1)}\int_{\mathbb{C}^{s}}\frac{\mathrm{d}^{s}\bm{\alpha}}{{\pi}^s} \frac{\braket{\bm{\alpha}|\phi}}{\|\bra{0}_1\ket{\psi}\|}  e^{i\phi(\bm{\alpha})}\mathcal N(\bm{\alpha})\ket{ A\bm{\alpha}}\\
        &\eqt{(v)}G_{(n-1)}\int_{\mathbb{C}^{s}}\frac{\mathrm{d}^{s}\bm{\alpha}}{{\pi}^s} \frac{\braket{\bm{\alpha}|\phi}}{\|\bra{0}_1\ket{\psi}\|}  e^{i\phi(\bm{\alpha})}\mathcal N(\bm{\alpha})\ket{ U_1\Sigma U_2\bm{\alpha}}\\
        &\eqt{(vi)}G_{(n-1)}G_{U_1}\int_{\mathbb{C}^{s}}\frac{\mathrm{d}^{s}\bm{\alpha}}{{\pi}^s} \frac{\braket{\bm{\alpha}|\phi}}{\|\bra{0}_1\ket{\psi}\|}  e^{i{\phi}(\bm{\alpha})}\mathcal N(\bm{\alpha})\ket{\Sigma U_2\bm{\alpha}}\\
        &\eqt{(vii)}G_{(n-1)}G_{U_1}\left(\int_{\mathbb{C}^{s}}\frac{\mathrm{d}^{s}\bm{\alpha}}{{\pi}^s} \frac{\braket{\bm{\alpha}|\phi}}{\|\bra{0}_1\ket{\psi}\|}  e^{i{\phi}(\bm{\alpha})}\mathcal N(\bm{\alpha})\ket{\bar{\Sigma} U_2\bm{\alpha}}\otimes \ket{0}^{\otimes(n-s-1)}\right)\\
        &\eqt{(viii)}G'\left(\ket{\phi'}\otimes \ket{0}^{\otimes(n-s-1)}\right)\\
    \ee
Here, in (i), we exploited Eq.~\eqref{rel_compl_coh}. In (ii), we exploited Eq.~\eqref{eq_action_gp}. In (iii), we defined
\be
    \mathcal N(\bm{\alpha})\coloneqq \left\| \bra{0}_1 G^{(2)} \left| \sum_{j=1}^{s} U_{1j}\alpha_j\right\rangle_{\!1}\otimes  \left| \sum_{j=1}^{s} U_{2j}\alpha_j\right\rangle_{\!2} \right\|\,,
\ee
and we exploited  \cref{lem_LL}, which guarantees the existence of the complex numbers $w_1,w_1\in\mathbb{C}$ presented in (iii). In (iv), we identified a matrix $A\in\mathbb{C}^{(n-1)\times s}$, which are defined in terms of $w_1$, $w_2$, and $U$. In (v), we employed the singular value decomposition of $A$ as $A=U_1\Sigma U_2$, where $U_1\in\mathbb{C}^{(n-1)\times(n-1)}$, $U_2\in\mathbb{C}^{s\times s}$ are unitary operations, and $\Sigma\in\mathbb{R}^{(n-1)\times s}$ is a rectangular diagonal matrix with its $s\times s$ upper block being the only non-zero block. In (vi), we identified a passive Gaussian unitary operation $G_{U_1}$ that acts on coherent states as follows:
\be
    G_{U_1}\ket{\bm{k}}=\ket{U_1\bm{k}}\qquad k\in\mathbb{C}^n\,.
\ee
In (vii), we denoted as $\bar{\Sigma}$ the $s\times s$ upper-block of $\Sigma$. Finally, in (viii), we defined $G'\coloneqq G_{(n-1)}G_{U_1}$ and we introduced the following $s$-mode pure state:
\be
    \ket{\phi'}\coloneqq \int_{\mathbb{C}^{s}}\frac{\mathrm{d}^{s}\bm{\alpha}}{{\pi}^s} \frac{\braket{\bm{\alpha}|\phi}}{\|\bra{0}_1\ket{\psi}\|}  e^{i{\phi}(\bm{\alpha})}\mathcal N(\bm{\alpha})\ket{\bar{\Sigma} U_2\bm{\alpha}}\,.
\ee
Hence, we conclude that the state $ \frac{\bra{0}_1\ket{\psi}}{\|\bra{0}_1\ket{\psi}\|}$ is $s$-compressible, and thus its symplectic rank is less or equal to $s$.

\medskip

\parhead{Taking a partial trace.} Given a bipartite state $\rho_{AB}$, let us show that
\be
    \mathfrak{s}(\Tr_A\rho_{AB})\le \mathfrak{s}(\rho_{AB})\,.
\ee
To this end, by denoting as $|A|$ the number of modes of $A$, we observe that
\be
\Tr_A\rho_{AB}&\eqt{(ix)}\int_{\mathbb{C}^{|A|}}\frac{\mathrm{d}^{|A|}\bm{\alpha}}{\pi^{|A|}}\bra{\bm{\alpha}}_A\rho_{AB}\ket{\bm{\alpha}}_A\\
&\eqt{(x)}\int_{\mathcal{S}_\rho}\frac{\mathrm{d}^{|A|}\bm{\alpha}}{\pi^{|A|}}\Tr\!\left[\bra{\bm{\alpha}}_A\rho_{AB}\ket{\bm{\alpha}}_A\right]\frac{\bra{\bm{\alpha}}_A\rho_{AB}\ket{\bm{\alpha}}_A}{\Tr\!\left[\bra{\bm{\alpha}}_A\rho_{AB}\ket{\bm{\alpha}}_A\right]}\,.
\ee
Here, in (ix), we exploited Eq.~\eqref{rel_compl_coh} and, in (x), we defined \be
    \mathcal{S}_\rho\coloneqq\left\{\alpha\in\mathbb{C}^{|A|}:\quad \Tr\!\left[\bra{\bm{\alpha}}_A\rho_{AB}\ket{\bm{\alpha}}_A\right]\ne 0\right\}
\ee
and we exploited that if $\Tr\!\left[\bra{\bm{\alpha}}_A\rho_{AB}\ket{\bm{\alpha}}_A\right]=0$ then $\bra{\bm{\alpha}}_A\rho_{AB}\ket{\bm{\alpha}}_A=0$. Hence, we obtain that
\be
    \mathfrak{s}( \Tr_A\rho_{AB})\leqt{(xi)} \sup_{\alpha\in\mathcal{S}_\rho }\mathfrak{s}\!\left( \frac{\bra{\bm{\alpha}}_A\rho_{AB}\ket{\bm{\alpha}}_A}{\Tr\!\left[\bra{\bm{\alpha}}_A\rho_{AB}\ket{\bm{\alpha}}_A\right]} \right)\leqt{(xii)} \sup_{\alpha\in\mathcal{S}_\rho }\mathfrak{s}(\rho_{AB})=\mathfrak{s}(\rho_{AB}) \,,
\ee
where in (xi) we exploited that the symplectic rank is non-increasing under classical mixing, and in (xii)
we used that the symplectic rank is non-increasing under heterodyne measurement with post-selection. This concludes the proof.
\end{proof}


\subsection{Consequences of the monotonicity of the symplectic rank}

In this section, we explore some consequences of the monotonicity of the symplectic rank under post-selected Gaussian operations.

\subsubsection*{No-go theorems on state conversion under post-selected Gaussian operations}

Since they can only decrease under free operations, resource monotones typically imply no-go theorems for the task of state conversion under free operations. In this section, we consider the specific case of symplectic rank as a resource monotone for non-Gaussianity. 

Consider the task of converting $n$ copies of an initial state $\rho$ into $m$ copies of a target state $\sigma$ via post-selected free operations. In general, if the maximal probability of converting $\rho^{\otimes n}$ into $\sigma^{\otimes m}$ via post-selected free operations is strictly larger than zero, then the conversion is \emph{possible} and we call it \emph{probabilistic}. Additionally, if such a maximal probability is exactly one, we say that the conversion is \emph{deterministic}. We stress that deterministic transformations are a subset of probabilistic ones.

In the context of non-Gaussianity, the exact conversion $\rho^{\otimes n}\rightarrow\sigma^{\otimes m}$ is said to be \emph{possible probabilistically} if there exists a post-selected Gaussian operation $\Phi_G$ such that 
\be
\Phi_G(\rho^{\otimes n})=\sigma^{\otimes m}\,.
\ee

As a consequence of the monotonicity of the symplectic rank, if the exact conversion $\rho^{\otimes n}\rightarrow\sigma^{\otimes m}$ is possible probabilistically, then the symplectic rank of $\rho^{\otimes n}$ is larger or equal to the symplectic rank of $\sigma^{\otimes m}$:
\be
    \mathfrak{s}(\rho^{\otimes n})\ge \mathfrak{s}(\sigma^{\otimes m})\,.
\ee
In other words, if $\mathfrak{s}(\rho^{\otimes n})<\mathfrak{s}(\sigma^{\otimes m})$, then it is \emph{impossible} to convert --- even probabilistically --- the initial state $\rho^{\otimes n}$ into the target state $\sigma^{\otimes m}$.

Let us give some examples demonstrating how the monotonicity of the symplectic rank can be leveraged to establish strong no-go results for state conversion under Gaussian protocols. First, no probabilistic Gaussian protocol can exactly convert a single-mode non-Gaussian pure state to a tensor product of two non-Gaussian pure states, even when allowing post-selection. More generally, the following holds:

\begin{coro}[(Impossibility of separately spreading non-Gaussianity via post-selected Gaussian operations)]\label{coro_spr}
Let $n<m$, and let $\rho$ be a state over $n$ modes and $\phi_1,\dots,\phi_m$ be pure single-mode non-Gaussian states. Then, no post-selected Gaussian protocol can exactly map $\rho$ to $\otimes_{i=1}^m\phi_i$.
\end{coro}
\begin{proof}
    This follows directly by applying: (i) the monotonicity of the symplectic rank from \cref{thm_main_sm}; (ii) the fact that the symplectic rank of a tensor product of $m$ single-mode non-Gaussian pure states is exactly $m$; and (iii) the fact that the symplectic rank of a state over $n$ modes is less or equal to $n$.
\end{proof}

Remarkably, this impossibility persists even if the $n$ initial states are very distant from the set of Gaussian states and the $m$ target states are extremely close to the set of Gaussian states. These kinds of no-go results cannot be detected by other commonly used quantum non-Gaussianity monotones, such as Wigner logarithmic negativity (WLN) \cite{Kenfack2004,Albarelli2018} or stellar rank~\cite{chabaud2020stellar,chabaud2021holomorphic,hahn2024assessing}. For instance, given a single-mode state $\psi_{\mathrm{in}}$ with an extremely large WLN (or stellar rank) and a non-Gaussian state $\phi_{\mathrm{out}}$ with very small WLN (or stellar rank), \cref{coro_spr} implies that the probabilistic conversion $\psi_{\mathrm{in}}\rightarrow \psi_{\mathrm{out}}^{\otimes 2}$ is impossible. This is remarkable because the non-Gaussianity of the target state $\psi_{\mathrm{out}}^{\otimes 2}$ --- as measured by the WLN (or stellar rank) --- can be arbitrarily greater than the one of the initial state $\psi_{\mathrm{in}}$.

As a concrete example, both WLN and stellar rank do not preclude the conversion of the Fock state $\ket{4}$ (WLN $\approx 0.78$, stellar rank $=4$) into the tensor product of two single-photon states $\ket{1} \otimes \ket{1}$ (WLN $\approx 0.71$, stellar rank $=2$), since we are moving from a state with higher WLN and stellar rank to one with lower WLN and stellar rank. However, as stated before, this is not true when considering the symplectic rank: $\ket{4}$ has a symplectic rank $\mathfrak{s}(\ket{4}) = 1$ and $\ket{1} \otimes \ket{1}$ has a symplectic rank $\mathfrak{s}(\ket{1} \otimes \ket{1}) = 2$. This presents an opposite picture to that given by the WLN and stellar rank, where now the state with a higher non-Gaussianity is $\ket{1} \otimes \ket{1}$, while the state with a lower non-Gaussianity is $\ket{4}$, establishing that the exact conversion $\ket{4}\rightarrow\ket{1}\otimes\ket{1}$ is impossible (not only deterministically but also probabilistically). 


\subsubsection*{The resource theory of non-Gaussianity is irreversible under exact post-selected Gaussian operations}\label{app:irreversibility}
In this subsection, we show that the resource theory of non-Gaussianity is \emph{irreversible} under exact post-selected Gaussian operations. Investigating the reversibility of a resource theory is a fundamental question since the early days of quantum information science~\cite{Vidal-irreversibility,gour2024resourcesquantumworld,Regula_2024}. For example, the pioneering paper~\cite{Bennett-distillation} showed that the resource theory of entanglement is reversible for pure states, whereas for mixed states, this is no longer the case~\cite{Vidal-irreversibility}. 


Let us start by giving some preliminaries following the presentation of Ref.~\cite{gour2024resourcesquantumworld}. We consider the conversion of $n$ copies of $\rho$ into $m$ copies of $\sigma$. To analyse this problem, it is useful to define the \emph{exact asymptotic distillable rate} and the \emph{exact cost rate}.
\begin{defi}[(Exact distillable rate and Exact cost rate)]
Let us consider a resource theory and let $\mathcal{F}$ be the set of free operations~\cite{gour2024resourcesquantumworld}. Let $\rho$ and $\sigma$ be two quantum states. The exact distillable rate of $\rho$ into $\sigma$, denoted as $\mathrm{Distill}_{\mathcal F}(\rho\to\sigma)$, is defined as
\be
\mathrm{Distill}_{\mathcal{F}}(\rho\to\sigma)\coloneqq\sup_{m,n\in\mathbb{N} }\left\{  \frac{m}{n} \colon\, \exists\, \Lambda \in \mathcal{F} \colon \Lambda(\rho^{\otimes n} ) = \sigma^{\otimes m}  \right\}\,.
\ee 
Analogously, the exact cost rate of $\rho$ into $\sigma$, denoted as $\mathrm{Cost}_{\mathcal{F}}(\rho\to\sigma)$, is defined as
\be
\mathrm{Cost}_{\mathcal{F}}(\rho\to\sigma)=\inf_{m,n\in\mathbb{N} }\left\{  \frac{n}{m} \colon \exists\, \Lambda \in \mathcal{F} \colon \Lambda(\rho^{\otimes n} ) = \sigma^{\otimes m}  \right\}\,.
\ee 
\end{defi}
Note that the two definitions are not independent of each other, and it is possible to show that~\cite{gour2024resourcesquantumworld}
\be \label{eq:conv_inv}
\mathrm{Distill}_{\mathcal{F}}(\rho\to\sigma) = \frac{1}{\mathrm{Cost}_{\mathcal{F}} (\rho\to\sigma)},
\ee 
for all states $\rho$ and $\sigma$. We can now define the notion of \emph{reversibility} and \emph{irreversibility} of a resource theory under \emph{exact} free operations.
\begin{defi}[(Reversibility and irreversibility of a resource theory under exact free operations)]\label{def_rev_resource}
    A resource theory is said to be reversible if for all states $\rho$ and $\sigma$ it holds that
    \be
        \mathrm{Distill}_{\mathcal{F}}(\rho\to\sigma) = \mathrm{Cost}_{\mathcal{F}}(\sigma\to\rho)\,,
    \ee
    where $\mathcal{F}$ denotes the set of free operations.
    A resource theory is said to be irreversible if it is not reversible.
\end{defi}
We stress that here we are considering the \emph{exact setting}. where free operations map $\rho^{\otimes n}$ into \emph{exactly} $\sigma^{\otimes m}$, which differs from the \emph{approximate setting}, where free operations map $\rho^{\otimes n}$ into \emph{approximately} $\sigma^{\otimes m}$ with an error that vanishes in the asymptotic limit of many copies~\cite{gour2024resourcesquantumworld}. 

Let us give some examples of reversibility and irreversibility regarding the celebrated entanglement theory~\cite{gour2024resourcesquantumworld}: (i) the resource theory of entanglement is irreversible under exact local operations and classical communication (LOCCs)~\cite{gour2024resourcesquantumworld}; (ii) the resource theory of entanglement is reversible under exact post-selected LOCCs for pure states, where both the exact distillable entanglement and the exact entanglement cost are given by the logarithm of the Schmidt rank~\cite{gour2024resourcesquantumworld}; (iii) the resource theory of entanglement is reversible under (approximate) post-selected LOCCs for pure states~\cite{Bennett-distillation}, where both the distillable entanglement and exact entanglement cost are given by the celebrated entanglement entropy~\cite{Bennett-distillation}; (iv) the resource theory of entanglement is irreversible under (approximate) post-selected LOCCs for mixed states~\cite{Vidal-irreversibility};

Hereafter, we consider the set of free operations $\mathcal{F}$ to be the set $\mathcal{G}$ of all post-selected Gaussian operations, as introduced in \cref{def_post_sel}. In the following theorem, we prove that the resource theory of non-Gaussianity is irreversible under such a set of operations $\mathcal{G}$ in the exact setting scenario.
\begin{theo}[(Irreversibility of the resource theory of non-Gaussianity)]\label{appthm:irreversibility}
  The resource theory of non-Gaussianity is irreversible under post-selected Gaussian protocols in the exact setting. That is, there exist two states $\rho$ and $\sigma$ such that the exact distillable rate and exact cost rate satisfy
  \be
    \mathrm{Distill}_\mathcal{G}(\rho \rightarrow \sigma) \neq \mathrm{Cost}_\mathcal{G}(\sigma \rightarrow \rho)\,.
\ee
\end{theo}

\begin{proof}
Let us fix two states $\rho$ and $\sigma$, and two natural numbers $n,m\in\mathbb{N}$. Let us assume that there exists a post-selected Gaussian operation $\Lambda \in \mathcal{G}$ such that $\Lambda(\rho^{\otimes n}) = \sigma^{\otimes m}$. Then, for any post-selected non-Gaussianity monotone $M(\cdot)$ (i.e.~a real functional over the set of states which is non-increasing under post-selected Gaussian operations), it must hold that
\be
    M(\rho^{\otimes n})\ge M(\sigma^{\otimes m})\,.
\ee
Additionally, if the non-Gaussianity monotone $M(\cdot)$ also satisfies that 
\be\label{add_prop}
    M(\rho^{\otimes n})&=nM(\rho)\,,\\
    M(\sigma^{\otimes n})&=nM(\sigma)\,,
\ee
then it follows that
\be
        \frac{m}{n}  &\le  \frac{M(\rho)}{M(\sigma)}\,\\
        \frac{n}{m}  &\ge  \frac{M(\sigma)}{M(\rho)}\,.
\ee
In particular, these two relations imply that
\be
     \mathrm{Distill}_{\mathcal{G}}(\rho \to \sigma)&\le \frac{M(\rho)}{M(\sigma)}\,,\\
    \mathrm{Cost}_{\mathcal{G}}(\sigma \to \rho)&\ge \frac{M(\rho)}{M(\sigma)}\,.
\ee
respectively. In particular, by optimising over the set of non-Gaussianity monotones $M(\cdot)$ which satisfy the additivity property in Eq.~\eqref{add_prop}, we obtain that
\be\label{rel_monotone}
     \mathrm{Distill}_{\mathcal{G}}(\rho \to \sigma)&\le \inf_{M}\frac{M(\rho)}{M(\sigma)}\,,\\
    \mathrm{Cost}_{\mathcal{G}}(\sigma \to \rho)&\ge \sup_{M}\frac{M(\rho)}{M(\sigma)}\,.
\ee
Now, let us consider the states $\rho=\ket{2}$ and $\sigma=\ket{1}^{\otimes 2}$, and let us consider the symplectic rank $\mathfrak{s}(\cdot)$ and the \emph{stellar rank}~\cite{hahn2024assessing}, denoted as $r^{\star}(\cdot)$. The latter two quantities are non-increasing under post-selected Gaussian operations, and they are also additive under tensor product for pure states \cite{chabaud2021holomorphic}. Note that the monotonicity of the stellar rank under post-selected Gaussian measurements was not explicitly noted previously, but it is a direct consequence of the proof of \cite[Corollary 4]{chabaud2021holomorphic}.
Consequently, Eq.~\eqref{rel_monotone} imply that
\be 
     \mathrm{Distill}_{\mathcal{G}}\!\left(\ket{2} \to \ket{1}^{\otimes 2}\right)& \le \frac{ \mathfrak{s}(\ket{2}) }{  \mathfrak{s}(\ket{1}^{\otimes 2})  }\,,\\
    \mathrm{Cost}_{\mathcal{G}}\!\left(\ket{1}^{\otimes 2}\to \ket{2}\right)&\ge \frac{r^{\star}(\ket{2})}{r^{\star}(\ket{1}^{\otimes 2})}\,.
\ee
The choice of such states has been made because they have simple values for the symplectic rank and the stellar rank, specifically it holds that 
\be
    \mathfrak{s}(\ket{2}) &= 1\,,\\
    r^{\star}(\ket{2}) &= 2\,,\\
    \mathfrak{s}(\ket{1}^{\otimes 2}) &= 2\,,\\
    r^{\star}(\ket{1}^{\otimes 2}) &= 2\,.
\ee
As a consequence, we obtain that
\be 
     \mathrm{Distill}_{\mathcal{G}}\!\left(\ket{2} \to \ket{1}^{\otimes 2}\right)& \le \frac{ \mathfrak{s}(\ket{2}) }{  \mathfrak{s}(\ket{1}^{\otimes 2})  }=\frac12\,,\\
    \mathrm{Cost}_{\mathcal{G}}\!\left(\ket{1}^{\otimes 2}\to \ket{2}\right)&\ge \frac{r^{\star}(\ket{2})}{r^{\star}(\ket{1}^{\otimes 2})}=1\,,
\ee
which proves that
\be
    \mathrm{Distill}_{\mathcal{G}}\!\left(\ket{2} \to \ket{1}^{\otimes 2}\right) < \mathrm{Cost}_{\mathcal{G}}\!\left(\ket{1}^{\otimes 2}\to \ket{2}\right)\,.
\ee
This concludes the proof.
\end{proof}

In the above theorem we have shown the irreversibility of the resource theory of non-Gaussianity by proving that there are two states $\rho$ and $\sigma $ such that the ratio between the exact distillable rate and the exact cost rate can be smaller than $1/2$, i.e.~
\be
    \frac{\mathrm{Distill}_\mathcal{G}(\rho \rightarrow \sigma)}{  \mathrm{Cost}_\mathcal{G}(\sigma \rightarrow \rho) } \le\frac12\,.
\ee
A natural question arises: can this ratio be arbitrarily small? I.e., is the irreversibility of the resource theory of non-Gaussianity arbitrarily strong? The following proposition shows that this is indeed the case.
\begin{theo}[(The irreversibility of the resource theory of non-Gaussianity is arbitrarily strong)]
    For all $n\in\mathbb{N}$, it holds that
    \be
    \frac{\mathrm{Distill}_\mathcal{G}\!\left(\ket{n}\rightarrow \ket{1}^{\otimes n}\right) }{  \mathrm{Cost}_\mathcal{G}\!\left(\ket{1}^{\otimes n}\rightarrow\ket{n}\right) } \le\frac1{n}\,.
\ee
In particular, it holds that
    \be
    \lim\limits_{n\rightarrow\infty}\frac{\mathrm{Distill}_\mathcal{G}\!\left(\ket{n}\rightarrow \ket{1}^{\otimes n}\right) }{  \mathrm{Cost}_\mathcal{G}\!\left(\ket{1}^{\otimes n}\rightarrow\ket{n}\right) }=0\,.
    \ee
\end{theo}
\begin{proof}
    Analogously to the above proof, we have
    \be
        \mathrm{Distill}_\mathcal{G}\!\left(\ket{n}\rightarrow \ket{1}^{\otimes n}\right) \le \frac{ \mathfrak{s}( \ket{n}  ) }{ \mathfrak{s}( \ket{1}^{\otimes n} ) }=\frac1n\,.
    \ee
    Moreover, it also holds that
    \be
        \mathrm{Cost}_\mathcal{G}\!\left(\ket{1}^{\otimes n}\rightarrow\ket{n}\right) \ge \frac{r^{\star}(\ket{n})}{r^{\star}(\ket{1}^{\otimes n})}=1\,,
    \ee
    where we exploited that that $r^{\star}(\ket{1}^{\otimes n})=n$ and $r^{\star}(\ket{n})=n$~\cite{hahn2024assessing}. This concludes the proof.
\end{proof}


\section{The symplectic rank cannot decrease under little perturbation of the state}\label{sec_non_dec}

In this section, we demonstrate that the symplectic rank of a pure state cannot decrease if one applies a slight perturbation to the state. In other words, the degree of non-Gaussianity, as quantified by the symplectic rank, remains robust under little perturbations: if a state exhibits a certain level of non-Gaussianity, then all states sufficiently close to it (while satisfying reasonable energy constraints) must retain at least the same level of non-Gaussianity. This result indicates that the symplectic rank is a meaningful measure of non-Gaussianity also from the practical point of view, as it is robust under little perturbations induced by the noise, which is intrinsic to any quantum device. This result is formalized in the forthcoming \cref{slight_pert}. Before presenting the theorem, let us introduce some useful notation. 

First, in \cref{slight_pert} we assume that the \emph{second moment of the energy} of the state is finite --- an assumption that is satisfied in any practical scenario. Let us recall that the \emph{energy observable} is defined as $\hat{E}\coloneqq \frac12 \hat{\bm{R}}^\intercal\hat{\bm{R}}$, where $\hat{\bm{R}}$ denotes the quadrature operator vector~\cite{BUCCO}. The \emph{mean energy} of a state $\rho$ is given by $\Tr[\hat{E}\rho]$, while the \emph{second moment of the energy} is given by  $\Tr[\hat{E}^2\rho]$. In the following, a state $\rho$ is said to be \emph{energy-constrained} if it satisfies  $\Tr[\hat{E}^2\rho]\le E^2$, where $E$ is a finite constant. 

Secondly, let us recall the notions of trace distance and fidelity. The \emph{trace distance}~\cite{NC} is the most operationally meaningful notion of distance between quantum states, given its operational interpretation due to the Holevo--Helstrom theorem~\cite{HELSTROM,Holevo1976}. Mathematically, the trace distance between two states $\rho_1$ and $\rho_2$ is defined as $\frac{1}{2} \|\rho_1 - \rho_2\|_1$, where $\|\Theta\|_1 \coloneqq \Tr\sqrt{\Theta^\dagger\Theta}$ denotes the \emph{trace norm} of an operator $\Theta$. Another important metric for comparing quantum states is the \emph{fidelity}, defined as $F(\rho_1, \rho_2) \coloneqq \|\sqrt{\rho_1} \sqrt{\rho_2} \|_1^2$. 
The trace distance and the fidelity are related via the so-called Fuchs--van de Graaf inequalities~\cite{Fuchs1999} as follows:
\begin{equation}\label{fidel_trace_ineq}
    1 - \sqrt{F(\rho_1, \rho_2)} \leq \frac12\|\rho_1-\rho_2\|_1 \leq \sqrt{1 - F(\rho_1, \rho_2)}\,.
\end{equation}  
The \emph{infidelity} between two quantum states $\rho_1$ and $\rho_2$ is defined as $1-F(\rho_1,\rho_2)$. The infidelity and the trace distance are always in $[0,1]$, being equal to zero if and only if $\rho_1=\rho_2$, and being equal to one if and only if $\rho_1$ and $\rho_2$ are orthogonal~\cite{NC,wilde_quantum_2013}. The Fuchs--van de Graaf inequalities in Eq.~\eqref{fidel_trace_ineq} imply that two states are very close w.r.t.~the trace distance if and only if they are very close w.r.t.~the infidelity. Moreover, the inequality on the left hand side can be refined as $1 - F(\rho_1, \rho_2) \leq \frac12\|\rho_1-\rho_2\|_1$ when one of the two states is pure.

We are now ready to present \cref{slight_pert}, which roughly ensures that every energy-constrained pure state of symplectic rank $s$ is surrounded by a neighbourhood of energy-constrained states with a symplectic rank that is at least as large as $s$. See \cref{fig:eball} for a pictorial representation of this theorem.
\begin{theo}[(Slight perturbations cannot decrease the symplectic rank)]\label{slight_pert}
    Let $\psi$ be a pure state with second moment of the energy bounded by a finite constant $E$, that is $\Tr[\hat{E}^2\psi]\le E^2$. Then, there exists $\varepsilon>0$ for which any state $\sigma$ such that
    \begin{itemize}    
        \item (i) $\sigma$ is $\varepsilon$-close to $\psi$, w.r.t.~the trace distance (and thus also w.r.t.~the infidelity),
        \item (ii) $\sigma$ satisfies the energy-constraint $\Tr[\hat{E}^2\sigma]\le E^2$,
    \end{itemize}
    must have symplectic rank larger or equal to the one of $\psi$, i.e.~
    \be 
        \mathfrak{s}(\sigma)\ge \mathfrak{s}(\psi)\,.
    \ee
    In other words, any pure state of symplectic rank $s$ admits a ball around it such that all energy-constrained states inside the ball have symplectic rank $\ge s$.
\end{theo}

\begin{proof}
First, let us observe that the concavity of the square root function implies that, for any state, the mean energy is upper bounded by the square root of the second moment of the energy. Consequently, since we are considering states with finite second-moment of the energy, we deduce that also the mean energy is finite. Hence,  \cref{well_defined_cov} implies that we are considering states with well-defined covariance matrices. 

    Since the statement is trivial if $\mathfrak{s}(\psi)=0$, let us assume that $\mathfrak{s}(\psi)\ge 1$. Let $\mathcal{S}$ be the set of energy-constrained states with symplectic rank smaller than the one of $\psi$. That is, $\mathcal{S}$ comprises all the states $\sigma$ such that $\mathfrak{s}(\sigma)<\mathfrak{s}(\psi)$ and $\Tr[\hat{E}^2\sigma]\le E^2$. Thanks to the Fuchs--van de Graaf inequalities in Eq.~\eqref{fidel_trace_ineq}, the trace distance is always lower bounded by the infidelity when one of the two states is pure. Consequently, in order to prove the theorem it suffices to show that the minimum infidelity between $\psi$ and any state $\sigma\in\mathcal{S}$ is strictly positive, i.e.~
    \be\label{infff}
        \inf_{\sigma\in\mathcal{S}}\left[1-\bra{\psi}\sigma\ket{\psi}\right]>0\,.
    \ee
    To this end, let us fix $\sigma\in\mathcal{S}$. By definition of symplectic rank for mixed states in \cref{defsmsymp}, there exists an ensemble of pure states $\{p_i,\psi_i\}$ such that: (i) all the $p_i$'s are strictly positive; (ii) $\sigma=\sum_i p_i\psi_i$; and (iii) $\max_i\mathfrak{s}(\psi_i)=\mathfrak{s}(\sigma)<\mathfrak{s}(\psi)$. Moreover, the infidelity between $\psi$ and $\sigma$ can be lower bounded as:
    \be\label{boundsINf}
        1-F(\psi,\sigma)&=1-\bra{\psi}\sigma\ket{\psi}\\
        &=\sum_ip_i\left(1-|\braket{\psi|\psi_i}|^2\right)\\
        &\eqt{(i)}\frac12\sum_ip_i\|\psi-\psi_i\|^2_1\\
        &\geqt{(ii)} \frac{(\nu-1)^4}{(640)^4}\sum_i\frac{p_i}{\max(E,E_i)^8}\\
        &\geqt{(iii)}\frac{(\nu-1)^4}{(640)^4E^8}\left(\frac45\right)^4\sum_{\substack{i:\\E_i\le \sqrt{\frac{5}{4}}E}}p_i\\
        &\geqt{(iv)} \frac{(\nu-1)^4}{(640)^4E^8}\left(\frac45\right)^4\frac14\\
        &=\frac14\left(\frac{\nu-1}{800\,E^2}\right)^{\!4}\,.
    \ee
Here, in (i), we exploited the well-known formula for the trace distance between two pure states~\cite{NC}. In (ii), we used the lower bound on the trace distance between two states in terms of the distance between their symplectic eigenvalues proven in \cref{pert_bound_symp} in \cref{sec_pert_bounds} below. Specifically, the latter theorem implies that
\be
    \|\psi-\psi_i\|_1\ge \frac{(\nu-1)^2}{(640)^2\max(E,E_i)^4}\,,
\ee
where, first, we denoted as $\nu$ the $(\mathfrak{s}(\psi)-1)$th largest symplectic eigenvalue of $V(\psi)$, which satisfies $\nu>1$ by definition of symplectic rank for pure states; and, second, we used that the $(\mathfrak{s}(\psi)-1)$th largest symplectic eigenvalue of $V(\psi_i)$ is one, as a consequence of the fact that the symplectic rank of $\psi_i$ is strictly smaller than $\mathfrak{s}(\psi)$. In (iii), we denoted as $E_i\coloneqq \sqrt{\Tr[\psi_i \hat{E}^2]}$ the second moment of the energy of state $\psi_i$, which implies that $\sum_ip_iE_i^2\le E^2$. In (iv), we used that
\be\label{p34}
     \sum_{\substack{i:\\ E_i\le \sqrt{\frac{5}{4}}E}}p_i\ge \frac14\,,
\ee
which can be proven as follows. By exploiting the fact that $E_i\ge0$, it holds that
    \be 
        E^2\ge \sum_i p_i E_i^2\ge\sum_{\substack{i:\\ E_i> \sqrt{\frac{5}{4}}E}}p_iE_i^2\ge \frac54 E^2\sum_{\substack{i:\\ E_i> \sqrt{\frac{5}{4}}E}}p_i
        =\frac54E^2\left(1-\sum_{\substack{i:\\ E_i\le \sqrt{\frac{5}{4}}E}}p_i\right)\,,
    \ee
which implies the validity of Eq.~\eqref{p34}. In conclusion, Eq.~\eqref{boundsINf} implies that
\be\label{eq_lower_bound_nuu}
    \inf_{\sigma\in\mathcal{S}}\left[1-F(\psi,\sigma)\right]\ge\frac14\left(\frac{\nu-1}{800\,E^2}\right)^{\!4}\,,
\ee
which is strictly positive. This concludes the proof.
\end{proof}

The above proof tells us more than what reported in the statement of \cref{slight_pert}. Indeed, the above Eq.~\eqref{eq_lower_bound_nuu} implies that the trace distance between $\psi$ and any energy-constrained state $\sigma$ with symplectic rank smaller than the one of $\psi$ must satisfy
\be
    \frac12\|\psi-\sigma\|_1\ge \frac14\left(\frac{\nu-1}{800\,E^2}\right)^{\!4}\,,
\ee
where $\nu$ is the smallest symplectic eigenvalue of $V(\psi)$ that is strictly larger than one, and $E$ is an upper bound on the second moment of the energy of $\psi$ and $\sigma$. This result establishes a simple expression for the radius $\varepsilon$ of the ball mentioned in the statement of \cref{slight_pert}. This is reported in the forthcoming \cref{thm_refined}, which constitutes a refinement of \cref{slight_pert}.

\begin{theo}[(Expression for the radius of the ball of states with high symplectic rank)]\label{thm_refined}
    Let $\psi$ be a non-Gaussian pure state such that its second moment of the energy is bounded by a finite constant $E$, that is, $\Tr[\hat{E}^2\psi]\le E^2$. Let $\nu$ be the smallest symplectic eigenvalue of $V(\psi)$ that is strictly larger than one. Then, any state $\sigma$ such that
    \begin{itemize}
        \item (i) $\sigma$ is $\frac14\left(\frac{\nu-1}{800\,E^2}\right)^{\!4}$-close to $\psi$ with respect to the trace distance (and, in particular, with respect to the infidelity):
        \be
            0\le \frac12\|\sigma-\psi\|_1\le \frac14\left(\frac{\nu-1}{800\,E^2}\right)^{\!4}  \,;
        \ee
        \item (ii) $\sigma$ satisfies the energy-constraint $\Tr[\hat{E}^2\sigma]\le E^2$;
    \end{itemize}
    must have symplectic rank larger or equal to the one of $\psi$, i.e.~
    \be 
        \mathfrak{s}(\sigma)\ge \mathfrak{s}(\psi)\,.
    \ee
    In other words, any pure state of symplectic rank $s$ is such that all the energy-constrained states inside the ball of radius $\frac14\left(\frac{\nu-1}{800\,E^2}\right)^{\!4}$ (with respect the trace distance) around it have symplectic rank $\ge s$.
\end{theo}

\begin{figure}
    \centering
    \includegraphics[width=0.35\linewidth]{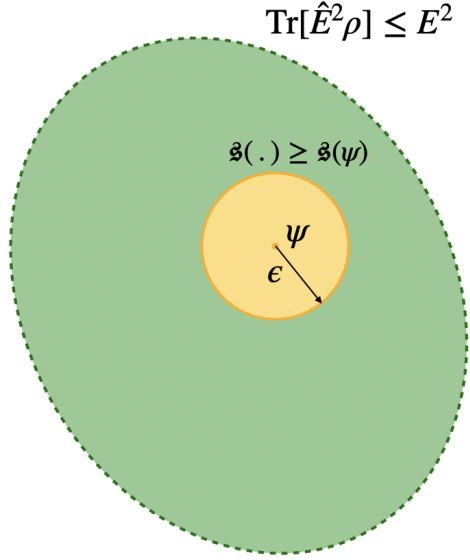}
    \caption{This figure provides a pictorial representation of \cref{slight_pert}: any state $\psi$, which satisfies the energy-constraint $\mathrm{Tr}\left[\hat{E}^2 \psi \right] \leq E^2$, admits an $\varepsilon$-ball (w.r.t.~the trace distance) around it such that any energy-constrained state inside the ball has symplectic rank at least as large as the one of $\psi$. Moreover, \cref{thm_refined} shows that one can take $\varepsilon\coloneqq \frac14\left(\frac{\nu-1}{800\,E^2}\right)^{\!4}$, where $\nu$ is the smallest symplectic eigenvalue of $V(\psi)$ which is strictly larger than one.}
    \label{fig:eball}
\end{figure}




\section{Symplectic rank as a measure of complexity}\label{app:compl}

In this section, we prove that the symplectic rank governs the \emph{complexity} of two fundamental tasks in bosonic quantum information processing: \emph{quantum state tomography} and \emph{classical simulation of quantum computations}.


\subsection{Quantum state tomography}

Quantum state tomography is a fundamental task in quantum physics, aiming to obtain a classical description of an unknown quantum state from experimental data~\cite{anshu2023survey}. Let us briefly give some preliminaries about quantum state tomography; for more details, we refer to Ref.~\cite{mele2024learning}.

We start by introducing the problem of quantum state tomography. Let us consider an unknown quantum state $\rho$ promized to be in a certain known class of states denoted as $\mathcal{S}$. Given $\varepsilon, \delta \in (0, 1)$, a tomography algorithm takes as input $N$ copies of the unknown quantum state $\rho\in\mathcal{S}$ and outputs a classical description of a quantum state $\overline{\rho}$ which is $\varepsilon$-close in trace distance to $\rho$, with probability $\ge 1-\delta$:
\be
\mathrm{Pr}\left[ \frac{1}{2} \|\rho - \overline{\rho}\|_1 \leq \varepsilon \right] \geq 1 - \delta\,,
\ee
where $\frac12 \|\rho - \overline{\rho}\|_1$ denotes the trace distance~\cite{NC} between the true unknown state $\rho$ and the estimate $\overline{\rho}$.  In quantum learning theory~\cite{anshu2023survey}, the trace distance arises as the most common notion of distance to measure errors in estimating states, as it has a strong operational meaning thanks to the Holevo--Helstrom theorem~\cite{HELSTROM,Holevo1976}.

In order to measure the performances of quantum state tomography, one usually considers the \emph{sample complexity}, defined as the minimum number of copies $N$ that allows one to solve the problem of quantum state tomography. Specifically, the sample complexity is a function of the set of states $\mathcal{S}$, of the trace distance error $\varepsilon$, and of the failure probability $\delta$. Moreover, another meaningful notion of complexity is the \emph{time complexity}, defined as the minimum time required to execute the tomography algorithm, accounting for both classical and quantum steps. It is worth mentioning that the time complexity always upper bounds the sample complexity~\cite{anshu2023survey}. Moreover, tomography is said to be \emph{efficient} if its sample and time complexities scale polynomially with the number of modes $n$. Otherwise, tomography is said to be \emph{inefficient}.

In Ref.~\cite{mele2024learning}, the problem of quantum state tomography of continuous-variable states has been extensively studied. In particular, Ref.~\cite{mele2024learning} proves upper~\cite[Theorem S74]{mele2024learning} and lower~\cite[Theorems S80]{mele2024learning} bounds on the sample and time complexities for quantum state tomography of $t$-compressible states, showing together with \cref{charact_symp} that both complexities grow exponentially with the symplectic rank. This result establishes that the symplectic rank is a fundamental measure that quantifies the hardness in the task of tomography: the higher the symplectic rank, the more challenging it becomes to perform tomography. In the following theorem, we slightly improve on the upper bound on the sample and time complexities proven in~\cite[Theorem S74]{mele2024learning}. 

\begin{theo}[(Tomography of states with a given symplectic rank)]\label{appthm:tomography}
Let $\mathcal{S}$ be the set of $n$-mode pure states $\psi$ which have symplectic rank at most $s$ and satisfy the energy constraint $\Tr\left[\hat{N}^2 \psi\right]\le E^2$. Then, the sample complexity $N$ of tomography with trace distance error $\varepsilon$ and failure probability $\delta$ from the set $\mathcal{S}$ satisfies the following bounds:
\be 
\tilde\Omega\!\left(\left(\frac{E}{12 s\varepsilon}\right)^{\!\!s}\right)\le N\le\tilde{\mathcal O}\!\left(\left(\frac{22E}{s\varepsilon^2}\right)^{\!\! s}\right),
\ee
where tilde abbreviates $\mathrm{poly}$ and $\mathrm{log}$ factors in all parameters. Additionally, the time complexity exhibits the same behaviour, thus establishing that  both complexities scale exponentially with the symplectic rank $s$. 
\end{theo} 

\begin{proof}
By \cref{charact_symp}, a pure state with a given symplectic rank $s$ corresponds to a $s$-compressible state, so that one can leverage the results of~\cite{mele2024learning} to obtain such bounds. The lower bound is fully proven in \cite[Theorems S80]{mele2024learning} and reads:
\be 
N&\ge\frac{1}{s\,g\!\left(  \frac{1}{ s} E \right)}\left[2(1-\delta)\left(\frac{ E  }{12 s\varepsilon  }-\frac{1}{s}\right)^s-(1-\delta)\log_2(32\pi)-H_2(\delta)\right] \\
\ee 
where $\varepsilon$ denotes the trace distance error, $\delta$ the failure probability, $g(x)\coloneqq (x+1)\log_2(x+1) - x\log_2 x\,$ is the bosonic entropy, and $H_2(x)\coloneqq -x\log_2x-(1-x)\log_2(1-x)$ is the binary entropy. 
Consequently, in asymptotic notation, neglecting all the polynomial and logarithmic terms, we can write
\be 
N \ge \tilde\Omega\!\left(\left(\frac{E}{12 s\varepsilon}\right)^{\!\! s}\right)\,,
\ee 
which concludes the proof of the lower bound on the sample complexity.

We now outline the general idea to obtain the improved upper bound, specifically by slightly improving on the scaling with the energy parameter $E$ compared to the upper bound of~\cite[Theorem S74]{mele2024learning}. For all technical details, as well as the full proof, we refer to~\cite[Theorem S74]{mele2024learning}. The proof relies on the fact that, thanks to~\cref{charact_symp}, the unknown state ${\psi}$ can be written as
\be 
    \ket{\psi} = G\left(\ket{\phi} \otimes \ket{0}^{\otimes(n - s)}\right)\,,
\ee 
where $G$ is a Gaussian unitary and $\phi$ is an $s$-mode state such that $N(\phi)\le N(\psi)$, where $N(\cdot)$ denotes the mean total photon number. By exploiting that $\Tr[\hat{N}^2\psi]\le E^2$, it thus follows that the mean total photon number of $\phi$ can be upper bounded as
\be\label{eq_meannnnn}
N(\phi)\le E\,.
\ee
Hence, in order to learn $\ket{\psi}$, it suffices to suitably learn the Gaussian unitary $G$ and the $s$-mode state $\ket{\phi}$. The rough idea of such a learning algorithm, as well as of its sample complexity analysis, is as follows (see~\cite[Theorem S74]{mele2024learning} for details):
\begin{itemize}
    \item First, one can estimate $G$ by estimating the first moment and covariance matrix of $\psi$, which requires $\mathrm{poly}(n,E)$ state copies;
    \item Secondly, one can approximately undo $G$ on the subsequent state copies, thus obtaining copies of a state $\ket{\tilde\psi}$ which is $o(\varepsilon)$-close in trace distance to $\ket{\phi}\otimes\ket{0}^{\otimes (n-s)}$ with high probability.
    \item Subsequently, one can perform a projective measurement to project $\ket{\tilde\psi}$ on the vacuum state of the last $n-s$ modes with high probability. The post-outcome state on the first $s$ modes, denoted as $\ket{\tilde\phi}$, turns out to be $o(\varepsilon)$-close in trace distance to $\ket{\phi}$ with high probability. Hence, in order to learn $\ket{\phi}$, one can perform full state tomography of the first $s$ modes.
    \item Finally, one can perform full state tomography of the $s$-mode state $\ket{\tilde\phi}$ using the tomography algorithm for energy-constrained states of~\cite[Theorem~S36]{mele2024learning}. Specifically,~\cite[Theorem~S36]{mele2024learning} guarantees the sample complexity of tomography of $s$-mode states with mean total photon number of at most $N_{\mathrm{phot}}$ is upper bounded by
    \be\label{sam_upp}
    \tilde{\mathcal{O}}\!\left(\left(\frac{eN_{\mathrm{phot}}}{s\varepsilon'^2}\right)^{\!s}\right)\,,
    \ee
    where $\varepsilon'$ is the trace distance error. Consequently, in order to exploit~\cite[Theorem~S36]{mele2024learning}, we need to establish an upper bound on the mean total photon number of $\ket{\tilde\phi}$. In~\cite[Theorem S74]{mele2024learning} the energy of $\ket{\tilde\phi}$ was overestimated and upper bounded by $80E^2$ (see~\cite[Eq.~(397)]{mele2024learning}). Here, by exploiting a completely analogous reasoning as~\cite[Theorem S74]{mele2024learning} together with \eqref{eq_meannnnn}, we can easily deduce that the mean total photon number of $\tilde{\phi}$ can be upper bounded as
    \be\label{eqineqeq}
        N\!\left(\tilde{\phi}\right)\le 2E\,.
    \ee
    The underlying idea of \eqref{eqineqeq} is that, since (i) \(\tilde{\phi}\) is an approximation of \(\phi\), and (ii) \(\phi\) has a mean total photon number bounded by \(E\), the mean total photon number of \(\tilde{\phi}\) can also be upper bounded by a slightly larger value than $E$ --- for instance, \(2E\). Hence, by using \eqref{sam_upp} with $N_{\mathrm{phot}}\coloneqq 2E$ and $\varepsilon'\coloneqq\varepsilon/2$ (this choice of $\varepsilon'$ is due to the details of the sample complexity analysis employed in~\cite[Theorem S74]{mele2024learning}), we conclude that this final step of the tomography algorithm requires a number of state copies given by
    \be\label{sam_upp2}
    \tilde{\mathcal{O}}\!\left(\left(\frac{8eE}{s\varepsilon^2}\right)^{\!s}\right)\,.
    \ee
    By using that $8e\le 22$, this concludes the proof.
\end{itemize}
\end{proof}

The main consequence of this result is that the sample complexity scales exponentially with the symplectic rank. This highlights the significance of the symplectic rank. Specifically, as long as the symplectic rank satisfies $s=\mathcal{O}(1)$, tomography requires only a polynomial number of copies in the number of modes. Otherwise, the number of required copies becomes super-polynomial, rendering tomography inefficient. In summary, the symplectic rank serves not only as a measure of non-Gaussianity but also as a measure of the hardness of tomography.


\subsection{Classical simulation}

Classical simulation is a fundamental tool for physicists, helping to understand and benchmark quantum systems. However, quantum mechanics is intrinsically challenging. A key difficulty arises from the fact that, for finite-dimensional systems of $n$ qudits, one must store and manipulate $d^n$ complex numbers. This exponential scaling imposes fundamental limitations on classical simulation. The challenge becomes even more severe when dealing with continuous-variable systems, where each mode has an infinite-dimensional Hilbert space. As a result, simulations require introducing a cut-off to make the problem tractable. In this setup, it is then fundamental to understand what can be simulated, in which way, and how many resources are required to do so. At this end, let us first briefly introduce for the sake of completeness the notion of strong and weak simulation \cite{terhal2002classical,pashayan2020estimation}. 

Let us first note that a quantum system and a measurement device are associated with a probability distribution, from which classical outcomes are sampled according to the Born rule. We refer to \emph{Strong simulation} when there exists a classical algorithm capable of evaluating the output probability distribution, or any of its marginals, in polynomial time relative to the size of the quantum system. Similarly, \emph{Weak simulation} refers to the existence of a classical algorithm that can generate samples from the output probability distribution in polynomial time relative to the size of the quantum system. These are two fundamental notions of classical simulability. It is important to note that the ability to perform strong simulation implies the ability to perform weak simulation. Additionally, strong simulation is generally a stricter requirement than what is needed for quantum computation, which is essentially equivalent to our ability to perform weak simulation. These two methods are the most relevant ways of performing classical simulation. However, for certain systems, even these approaches can be challenging. In such cases, one can look for simpler objects to estimate classically, such as performing a weaker form of simulation by estimating the overlap between two states. 

With these notions in mind, we prove that the overlaps between general states with low symplectic rank and Gaussian states can be estimated efficiently on a classical computer. Then, we refine this result to obtain efficient classical simulation algorithms for specific classes of non-Gaussian circuits.


\subsubsection*{General-purpose algorithm}

Let $\psi$ be an $n$-mode pure state with symplectic rank $s$ and mean photon number less than $E$. Recall that by \cref{charact_symp}, there exists a pure state $\phi$ over $s$ modes with mean photon number less than $E$ such that
\be\label{eq:decomppsi}
    \ket{\psi}={G}\left(\ket{\phi}\otimes \ket{0}^{\otimes(n-s)}\right).
\ee

We now state our general classical simulation result:

\begin{theo}[(Classical simulation based on the symplectic rank)]\label{appthm_simul}
Let $\psi$ be an $n$-mode pure state with symplectic rank $s$ and mean photon number less than $E$. There is a classical randomized algorithm which takes as input parameters $\varepsilon,\delta>0$, the description of a Gaussian state $\ket{G}$, and the decomposition $( G,\phi)$ of ${\psi}$ from \cref{eq:decomppsi} where $\phi$ is given in the Fock basis, and outputs an estimate of the overlap $|\braket{G|\psi}|^2$ up to additive precision $\varepsilon$ with probability $1-\delta$, in time 
\be
    \mathcal O\left(n^3+\frac{n^2s}{\varepsilon^2}\left(\frac{6E}{s\varepsilon^2}\right)^{3s}\log\left(\frac E{s\varepsilon^2\delta}\right)\right).
\ee
\end{theo}

\begin{proof}
    When $s=0$, the overlap is Gaussian and the result is trivial, so we assume $s\ge1$ in what follows. Let $\lambda,\mu>0$ be free parameters, to be fixed later on. By denoting as $\Pi_N$ the projector on the span of the Fock states with total photon number smaller or equal to $N\coloneqq\lceil\frac E{\lambda^2}\rceil$ (we omit the ceiling function in what follows for brevity) and by denoting as $\ket{\phi_N}\propto  \Pi_N\ket{\phi}$ the state obtained by truncating $\phi$ in Fock basis at photon number $N$, it holds that
    \be
        \frac12\left\|\phi-\phi_N\right\|_1\leqt{(i)} \sqrt{\Tr[(\mathbb{1}-\Pi_N)\phi]}\leqt{(ii)} \sqrt{\frac{\Tr[\phi\hat{N}]}{N}}\le \lambda\,,
    \ee 
    where (i) follows from the gentle measurement lemma~\cite{wilde_quantum_2013}, in (ii) we simply used the operator inequality $(1-\Pi_N)\le \frac{\hat{N}}{N}$, and in (iii) we used that the mean photon number of the compressed state $\phi$ is bounded by $E$, as per \eqref{eq:decomppsi}. Consequently, we deduce that the compressed state $\phi$ is $\lambda$-close in trace distance to the truncated state ${\phi_N}$. Moreover, by defining 
    \be
        \ket{\psi_N}\coloneqq G\left(\ket{\phi_N}\otimes\ket0^{\otimes(n-s)}\right)\,,
    \ee    
    and by exploiting that the trace distance is invariant when tensoring with identical states and under unitary operations, we obtain that
    \be
        \frac12\|\ket\psi\!\bra\psi-\ket{\psi_N}\!\bra{\psi_N}\|_1\le\lambda\,.
    \ee
    By the variational property of the trace distance~\cite{NC}, considering the binary measurement $\{\ket G\!\bra G,I-\ket G\!\bra G\}$, we obtain that
    \be\label{boundlambda}
        \left||\braket{G|\psi}|^2-|\braket{G|\psi_N}|^2\right|\le\lambda.
    \ee
    Now we have
    \be
    \begin{aligned}
        \braket{G|\psi_N}&=\bra G G\left(\ket{\phi_N}\otimes\ket0^{\otimes(n-s)}\right)\\
        &=\sum_{n_1+\dots+n_s\le N}\phi_{n_1\dots n_s}\braket{G'|n_1\dots n_s0\dots0},
    \end{aligned}
    \ee
    where we have expanded $\ket{\phi_N}$ in Fock basis and where we have defined $\ket{G'}\coloneqq G^\dag\ket G$, which is a Gaussian state whose Williamson decomposition (see \cref{lem_param_pure})
    \be\label{eq_willlll}
        \ket{G'}= U\bigotimes_{k=1}^n\left( D(\alpha_i) S(\xi_i)\right)\ket0^{\otimes n}
    \ee
    can be computed in time $\mathcal O(n^3)$ from the descriptions of $\ket G$ and $ G$. From  \cref{lem_param_pure}, we recall that ${U}$ represents a passive Gaussian unitary, ${D}(\alpha_i)$ is the displacement operator acting on the $i$th mode, and ${S}(\xi_i)$ denotes the single-mode squeezing operator on the same mode.    In particular, for all $(n_1,\dots,n_s)\in\mathbb N^s$,
    \be
        \braket{G'|n_1\dots n_s0\dots0}=\left(\bigotimes_{k=1}^n\bra{G_i}\right) U\ket{n_1\dots n_s0\dots0},
    \ee
    where $\ket{G_i}\coloneqq  D(\alpha_i) S(\xi_i)\ket{0}$ are single-mode Gaussian states and $ U$ is a passive Gaussian unitary. By \cite[Theorem 2]{lim2025efficient}, there is a classical randomized algorithm which estimates this quantity to additive precision $\varepsilon'$ with probability $1-\delta'$ in time $\mathcal O(\frac{n^2}{{\varepsilon'}^2}\log(\frac1{\delta'}))$. We denote by $A_{n_1\dots n_s}$ the responding classical estimate produced by this algorithm, and we define
    \be
        \tilde G\coloneqq\sum_{n_1+\dots+n_s\le N}\phi_{n_1\dots n_s}A_{n_1\dots n_s}.
    \ee
    Let also $M\coloneqq\binom{N+s}s$ be the number of terms in the sum. Setting $\varepsilon'=\frac\mu {2M}$ and $\delta'=\frac\delta M$ we obtain
    \be\label{estimatebounds}
    \begin{aligned}
        |\braket{G|\psi_N}-\tilde G|&=\left|\sum_{n_1+\dots+n_s\le N}\phi_{n_1\dots n_s}\left(\braket{G'|n_1\dots n_s0\dots0}-A_{n_1\dots n_s}\right)\right|\\
        &\le\sum_{n_1+\dots+n_s\le N}|\phi_{n_1\dots n_s}|\left|\braket{G'|n_1\dots n_s0\dots0}-A_{n_1\dots n_s}\right|\\
        &\le M\varepsilon'\\
        &=\frac12\mu,
    \end{aligned}
    \ee
    with probability $1-\delta$, using the union bound of failure probabilities. The classical randomized algorithm outputting $\tilde G$ computes $\ket{G'}$ in time $\mathcal O(n^3)$ and then executes $M$ times the algorithm from \cite[Theorem 2]{lim2025efficient}, for $\varepsilon'=\frac\mu{2M}$ and $\delta'=\frac\delta M$, so its running time is given by 
    \be
        \mathcal O\left(n^3+\frac{n^2M^3}{\mu^2}\log\left(\frac M\delta\right)\right),
    \ee
    where $M=\binom{N+s}s$ and $N=\frac E{\lambda^2}$. Moreover,
    \be
    \begin{aligned}
        \left||\braket{G|\psi_N}|^2-|\tilde G|^2\right|&\le|\braket{G|\psi_N}|\left|\braket{G|\psi_N}^*-\tilde G^*\right|+|\tilde G|\left|\braket{G|\psi_N}-\tilde G\right|\\
        &\le\mu,
    \end{aligned}
    \ee
    with probability $1-\delta$, where we used the triangle inequality in the first line, and Eq.~\eqref{estimatebounds} together with the facts that $|\braket{G|\psi_N}|\le1$ and $|\tilde G|\le1$ in the second line (if the latter does not hold, we can always rescale $\tilde G$ so that $|\tilde G|=1$ without loss of generality since it can only improve the quality of the estimate). With Eq.~\eqref{boundlambda} and the triangle inequality, this implies
    \be\label{boundlambdamu}
        \left||\braket{G|\psi}|^2-|\tilde G|^2\right|\le\lambda+\mu,
    \ee
    with probability $1-\delta$, and the estimate is obtained in time
    \be
        \mathcal O\left(n^3+\frac{n^2\binom{\frac E{\lambda^2}+s}s^3}{\mu^2}\log\left(\frac {\binom{\frac E{\lambda^2}+s}s}\delta\right)\right).
    \ee
    Using $\binom mk\le(\frac{em}k)^k$ (valid for all $m\ge k>0$), as well as $1+\frac E{s\lambda^2}\le\frac{2E}{s\lambda^2}$ and $\delta\ge\delta^s$, and setting $\lambda=\frac{\nu-1}{\nu}\varepsilon$ and $\mu=\frac1{\nu}\varepsilon$ for some constant parameter $\nu>1$, this gives a running time bounded as
    \be
        \mathcal O\left(n^3+\frac{n^2}{\varepsilon^2}\left(\frac E{s\varepsilon^2}\cdot\frac{2e\nu^2}{(\nu-1)^2}\right)^{3s}s\log\left(\frac{E}{s\varepsilon^2\delta}\cdot\frac{2e\nu^2}{(\nu-1)^2}\right)\right).
    \ee
    Finally, choosing $\nu>1$ large enough such that $\frac{2e\nu^2}{(\nu-1)^2}\le6$ ($\nu\ge20$ suffices), we bound the total running time of the classical simulation algorithm as
    \be
        \mathcal O\left(n^3+\frac{n^2s}{\varepsilon^2}\left(\frac{6E}{s\varepsilon^2}\right)^{3s}\log\left(\frac E{s\varepsilon^2\delta}\right)\right),
    \ee
    while the total error in Eq.~\eqref{boundlambdamu} is at most $\lambda+\mu=\varepsilon$, with probability $1-\delta$.
\end{proof}

The classical simulation algorithm from \cref{appthm_simul} takes as input the (approximate) decomposition of the state $\psi$ as in Eq.~\eqref{eq:decomppsi}. However, in practice, the state of interest is typically described as the output state of some bosonic circuit. In what follows, we show that a circuit description of the decomposition in Eq.~\eqref{eq:decomppsi} can be efficiently obtained from the circuit description of relevant examples of bosonic circuits (\cref{theo:efficient_nG}). In particular, we show that this leads to efficient classical simulation algorithms for two specific examples of families of bosonic circuits: in the first case, these are circuits composed of Gaussian gates and cubic phase gates (\cref{thm:simucubic}), while in the second case they are composed of Gaussian gates and creation and annihilation operators, known as photon additions/subtractions in the quantum optics literature (\cref{th:simuphotaddsub}).


\subsubsection*{Doped Gaussian circuits}
Before diving into the simulation results for specific bosonic circuits, let us start by defining the general class of \emph{$(\kappa,s)$-doped Gaussian gates}. See \cref{fig:kappa-s_doped} for a pictorial representation.  
\begin{defi}[(Doped Gaussian gate)]
    An operator ${C}$ is an $n$-mode $(\kappa,s)$-doped Gaussian gate if it is a composition of $n$-mode Gaussian unitaries ${G}_0,\dots,{G}_s$ and at most $s$ (possibly non-unitary and non-Gaussian) $\kappa$-local gates $ W_1,\dots, W_s$:
    \begin{equation}
        {C} = {G}_s  W_s {G}_{s-1} \dots {G}_1  W_1 {G}_0\,.
    \end{equation}
    Moreover, an operator is said to be $\kappa$-local if it is a function of at most $\kappa$ operators from the set of quadrature operators.
\end{defi}
\begin{figure}
    \centering
    \includegraphics[width=0.9\linewidth]{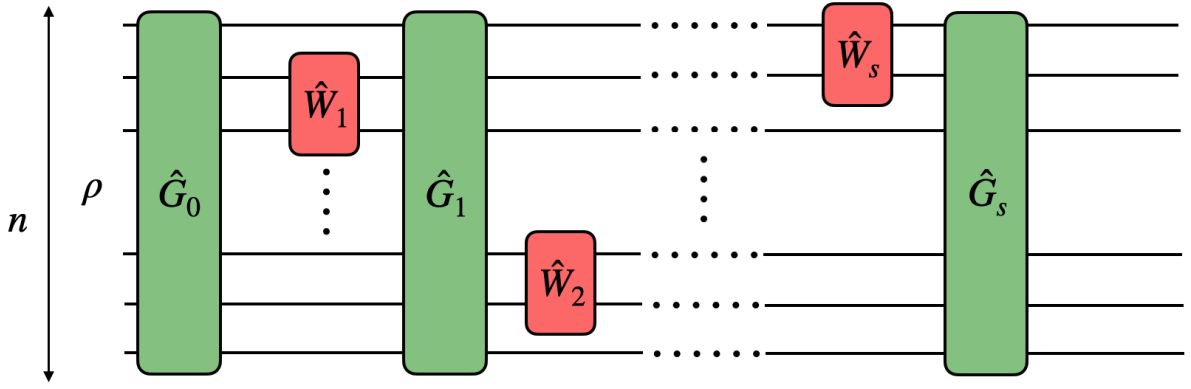}
    \caption{Pictorial representation of a $(\kappa,s)$-doped Gaussian state. By definition, an $n$-mode $(\kappa,s)$-doped Gaussian state is a state prepared by applying Gaussian unitaries ${G}_0,\dots, G_s$ and at most $s$ non-Gaussian $\kappa$-local gates $ W_1,\dots, W_s$ to the $n$-mode vacuum.}
    \label{fig:kappa-s_doped}
\end{figure}
Note that we are using a more general definition of doped Gaussian gates than the one in~\cite{mele2024learning}, as we are also allowing for doping gates which are not unitaries (e.g.~annihilation or creation operators). Let us proceed with the definition of $(\kappa,s)$-doped Gaussian state.
\begin{defi}[(Doped Gaussian state)]
    An $n$-mode $(\kappa,s)$-doped Gaussian state $\psi$ is a state prepared by applying an $n$-mode $(\kappa,s)$-doped Gaussian gate ${C}$ to the vacuum state:
\begin{equation}\label{eq_prop}
    \ket \psi \propto  C \ket{0}^{\otimes n}\,.
\end{equation}
\end{defi}
Note that we are using the proportional symbol in \eqref{eq_prop} because the doped Gaussian gate ${C}$ can be a non-unitary operator. Let us state the following remarkable result about the symplectic rank of doped Gaussian states, which is a slight generalisation of \cite[Theorem~S67]{mele2024learning}. 
\begin{theo}[(Symplectic rank of doped Gaussian states)]\label{thm_doped}
    Let $\kappa s\le n$. Then, any $n$-mode $(\kappa,s)$-doped Gaussian state $\psi$ has symplectic rank less or equal to $\kappa s$:
    \begin{equation}
        \mathfrak{s}(\psi)\le \kappa s\,.
    \end{equation}
\end{theo}
\begin{proof}
    The statement is a slight generalisation of \cite[Theorem~S67]{mele2024learning}, which was originally stated for doped Gaussian states generated by doped Gaussian \emph{unitary} gates. However, it is simple to observe that the same proof extends to doped Gaussian states produced via \emph{non-unitary} doped Gaussian gates as well. This is because the proof of \cite[Theorem~S67]{mele2024learning} relies solely on the fact that each doping gate admits a power series expansion in the quadrature operators. Such a proof is summarised in the proof of \cref{theo:efficient_nG} below.
\end{proof}
Thanks to \cref{thm_doped} and to the definition of symplectic rank, for any $n$-mode $(\kappa,s)$-doped Gaussian state $\psi$ with $\kappa s\le n$ there exist an $n$-mode Gaussian unitary ${G}$ and a $(\kappa s)$-mode state $\phi$ such that
\be\label{eq_comprr}
    \ket{\psi}={G}\left(\ket{\phi}\otimes\ket{0}^{\otimes(n-\kappa s)} \right)\,.
\ee
The following theorem proves that, given an efficient circuit description of any such $n$-mode $(\kappa,s)$-doped Gaussian state $\psi$, one can \emph{efficiently} find the Gaussian unitary ${G}$ and the $(\kappa s)$-mode state $\phi$ present in \eqref{eq_comprr}.
\begin{theo}[(Efficient non-Gaussian compression)]\label{theo:efficient_nG}
Let $\kappa s\le n$. Given an efficient circuit description of an $n$-mode $(\kappa,s)$-doped Gaussian state $\psi$, i.e., with Gaussian unitary gates described by their symplectic matrix and displacement vector, and non-Gaussian gates expressed as efficiently computable functions of quadrature operators, we can compute in time $\mathcal{O}(s^2n^3)$ a description of a Gaussian unitary $ G$ and an efficient circuit description of a $(\kappa s)$-mode state $\phi$ such that 
\be\label{eq-DEC}
\ket{\psi}={G}\left(\ket{\phi}\otimes\ket{0}^{\otimes(n-\kappa s)} \right)\,,
\ee
where $\phi$ is a $(\kappa,s)$-doped Gaussian state.
\end{theo}
\begin{proof}
The proof of the decomposition in \eqref{eq-DEC} is the exactly the same as the one in~\cite[Theorem~S67]{mele2024learning}, and so we refer the reader to the latter. Here, we provide a brief summary of the argument, together with an additional estimate of the computational time required to perform such a decomposition.

Let us consider the description of the $(\kappa,s)$-doped state $\psi$ as
    \begin{equation}\label{eq_def_psi}
        \ket{\psi} \propto {C} \ket0^{\otimes n}\,,
    \end{equation}
where
\begin{equation}\label{eq_cricuittt}
 {C} = {G}_s  W_s {G}_{s-1} \dots {G}_1  W_1 {G}_0\,,
    \end{equation}
with each $ W_i$ being a non-Gaussian $\kappa$-local gate and each $ G_i$ being a Gaussian unitary. Analogously to~\cite[Theorem S65]{mele2024learning}, one can observe that $C$ can be expressed as
\begin{equation}\label{eq_sdddd}
    C = {\tilde{G}}_s (\Pi_{i=1}^s  {\tilde W}_i) {G}_\mathrm{passive},
\end{equation}
where ${G}_\mathrm{passive}$ is a passive Gaussian unitary, each operator ${\tilde W}_i$ is defined as 
\be
{\tilde W}_i \coloneqq {G}_{\mathrm{passive}}^\dagger {\tilde{G}}_{i-1}^\dagger   W_i {\tilde{G}}_{i-1} {G}_{\mathrm{passive}}\,,
\ee
and each ${\tilde{G}}_i$ is defined by ${\tilde{G}}_i \coloneqq {G}_i\dots{G}_0$. Crucially, as shown in~\cite[Theorem S65]{mele2024learning},  ${G}_{\mathrm{passive}}$ can be constructed such that the operator ${\tilde W}_i$ acts on only the first $\kappa s$ modes, for all $i\in\{1,\dots,s\}$. The descriptions of each circuit element present in \eqref{eq_cricuittt} (i.e.~the descriptions of ${G}_\mathrm{passive}$, $\Pi_{i=1}^s  {\tilde W}_i$, and ${\tilde{G}}_s$) can be efficiently computed as follows.

Let us start by introducing some notation.
For $s' \in \{0,1,\dots,s\}$, let $S_{s'}$ be the symplectic matrix associated with the Gaussian unitary ${\tilde G}_{s'}$, and let $\mu(s',1) \leq \mu(s',2) \leq \dots \leq \mu(s',\kappa)$ be the quadratures involved in the power series that defines ${W}_{s'}$.

First, let us observe that given the description of $ G_0,\dots,  G_s$, the description of $S_{0},\ldots, S_s$ can be computed in time $\mathcal O(s^2n^3)$. Second, by denoting as $(e_i)_{i\in [2n]}$ the canonical basis vector of $\R^{2n}$, one can compute the vectors
\be
    (v_j)_{j\in[\kappa s]} \coloneqq (S_{s'-1}^T e_{\mu(s',r)})_{s' \in [s],r \in [\kappa]}
\ee
in time $\mathcal O(\kappa s n^2)$. As shown in~\cite[Eq.~(334)]{mele2024learning}, constructing a passive Gaussian unitary ${G}_{\mathrm{passive}}$ such that the operators ${\tilde W}_i$ in \eqref{eq_sdddd} act only on the first $\kappa s$ modes for all $i\in\{1,\dots,s\}$ is equivalent to finding a $2n\times 2n$ symplectic orthogonal matrix ${O}_\mathrm{passive}$ such that for all $j \in [\kappa s]$ and all $m \in \{2\kappa s + 1,\dots, 2n\}$,
\begin{equation}
    e_m^T {O}^T_{\mathrm{passive}} v_j = 0\,.
\end{equation}
Following \cite[Lemma~28]{mele2024efficient}, computing such ${O}_\mathrm{passive}$ can be done in time $\mathcal{O} (n^3)$, as we explain hereafter. The $2n\times 2n$ symplectic orthogonal matrix $ O_{\mathrm{passive}}$ is isomorphic to an $n\times n$ unitary $U_\mathrm{passive}$ through a well-defined vector space mapping. Specifically, for a $2n$-dimensional vector $v\coloneqq(v_1,\dots,v_{2n})$, there exists a bijective mapping to an $n$-dimensional complex vector $f(v)\coloneqq(v_1-iv_2,\dots,v_{2n-1}-iv_{2n})$. Through this bijection, the problem of finding $O_\mathrm{passive}$ maps to the problem of finding a unitary operation $U_\mathrm{passive}$ that maps the span of $\{f(v_1),\dots,f(v_{\kappa s})\}$ to the span of the first $\kappa s$ canonical basis vectors of this $n$-dimensional complex space. This is achieved using the Gram--Schmidt orthonormalization of $\{f(v_1),\dots,f(v_{\kappa s})\}$, which can be computed in time $\mathcal{O} (n \kappa^2 s^2)$ \cite{gloub1996matrix}. Completing this orthonormal family to a $2n\times2n$ orthogonal matrix is then done in time $\mathcal O(n^3)$.

Thanks to \eqref{eq_def_psi}, \eqref{eq_sdddd}, and to the fact that passive Gaussian unitaries preserve the vacuum, we obtain that
\be\label{stet_proof}
    \ket{\psi}&\propto{\tilde{G}}_s (\Pi_{i=1}^s  {\tilde W}_i)\ket{0}^{\otimes n}\,.
\ee
By noting that $(\Pi_{i=1}^s  {\tilde W}_i)$ acts only on the first $\kappa s$ modes, we can define a $(\kappa s)$-mode state $\ket{\phi}$ such that
\be
    \ket{\phi}\propto(\Pi_{i=1}^s{\tilde W}_i)\ket{0}^{\otimes (\kappa s)}\,.
\ee
Hence, by defining ${G}\coloneqq {\tilde{G}}_s$, \eqref{stet_proof} implies that
\be
    \ket{\psi} =  G (\ket  \phi \otimes \ket{0}^{\otimes (n-\kappa s)})\,.
\ee
In conclusion, thanks to the above estimates of the computational times, the descriptions of $ G$ and $\ket{\phi}$ can be computed in time $\mathcal{O}(s^2n^3 + \kappa s n^2 + n^3)=\mathcal{O}(s^2n^3)$ since $\kappa s\le n$. 
\end{proof}


\subsubsection*{Gaussian circuits doped with cubic phase gates}

In this section, we show that combining \cref{appthm_simul} and \cref{theo:efficient_nG} allows us to estimate Gaussian statistics of the output states of Gaussian circuits doped with cubic phase gates. This result is formalized as follows:

\begin{theo}[(Estimating output statistics of Gaussian circuits doped with cubic phase gates)]\label{thm:simucubic}
    Let $s\ge1$, $\varepsilon,\delta>0$ and let 
    \begin{equation}
        {U} = {G}_s e^{i\gamma_s\hat x_1^3} {G}_{s-1} \dots {G}_1 e^{i\gamma_1\hat x_1^3} {G}_0,
    \end{equation}
   where $\hat{x}_1$ denotes the position operator on the first mode, $\gamma_1,\dots,\gamma_s$ are real parameters of at most $\mathcal O(\poly\!(n))$ in magnitude, and $ G_0,\dots, G_s$ are Gaussian unitaries with entries of their symplectic matrix and displacement vector at most $\mathcal O(\poly\!(n))$ in magnitude. Given an $n$-mode Gaussian state $\ket G$, there is a classical randomized algorithm which outputs an estimate of the overlap $|\bra G{U}\ket{0}^{\otimes n}|^2$ up to additive precision $\varepsilon$ with probability $1-\delta$, in time 
\be
    \mathcal O\left(\frac{n^2s}{\varepsilon^2}\left(\frac{6E}{s\varepsilon^2}\right)^{3s}\log\left(\frac {E}{s\varepsilon^2\delta}\right) + \poly\!(n)^s\right),
\ee
where $E=\mathcal{O}(\poly\!(n)^s)$. In particular, the algorithm is efficient when $s=\mathcal O(1)$.
\end{theo}

Note that the cubic phase gates, denoted as $e^{i\gamma_s\hat x_1^3}$, act on the first mode without loss of generality, since swapping two modes is a Gaussian unitary operation.

\begin{proof}
    Let $\ket\psi={U}\ket{0}^{\otimes n}$ be the output state of the circuit.
    From the proof of \cref{appthm_simul}, it is enough to show that an approximate decomposition $( G,\tilde\phi)$ of $\psi$ can be computed efficiently (where $\tilde\phi$ is given in Fock basis) rather than an exact decomposition $( G,\phi)$, as long as $\|\phi-\tilde\phi\|_1\le\varepsilon$. We do so in four steps: 
\begin{itemize}
        \item We use \cref{theo:efficient_nG} to show that $\phi$ can be efficiently written as the $s$-mode output state of a Gaussian circuit doped with cubic phase gates.
        \item We use \cite[Lemma 1]{upreti2025bounding} to show that $\phi$ can further be written as the $s$-mode output state of a passive Gaussian circuit doped with displacements and cubic phase gates.
        \item We use \cite[Lemma 2]{upreti2025bounding} and \cite[Lemma 3]{upreti2025bounding} to show that $\phi$ is well approximated by a superposition of coherent states $\tilde\phi$ whose Fock states amplitudes can be computed in time $\mathcal O(\poly\!(n)^s)$.
        \item Finally, we obtain an energy bound on $\tilde\phi$ and apply \cref{appthm_simul}.
\end{itemize}

\medskip

Let
\begin{equation}
    \ket{\psi} \coloneqq  U \ket{0}^{\otimes n},
\end{equation}
with 
\begin{equation}
        {U} = {G}_s e^{i\gamma_s\hat x_1^3} {G}_{s-1} \dots {G}_1 e^{i\gamma_1\hat x_1^3} {G}_0\,.
    \end{equation}
From \cref{theo:efficient_nG}, the state ${\psi}$ can be expressed as
\begin{equation}
    \ket{\psi} = {G} (\ket \phi \otimes \ket 0^{\otimes (n-s)}),
\end{equation}
in time $\mathcal O(\poly\!(n))$, with
\begin{equation}
    \ket \phi\otimes \ket{0}^{\otimes (n-s)}= (\Pi_{i=1}^s {G}_{\mathrm{passive}}^\dagger {\tilde{G}}_{i-1}^\dagger  e^{i\gamma_i\hat{x}_1^3} {\tilde{G}}_{i-1} {G}_{\mathrm{passive}}) \ket 0^{\otimes n}\,,
\end{equation}
where ${G}_\mathrm{passive}$ is a passive Gaussian unitary over $n$ modes and ${\tilde{G}}_{i-1} = {G}_{i-1}\dots{G}_0$ for all $i\in\{1,\ldots, s\}$, while $\tilde{G}_i{G}_\mathrm{passive}$ is a Gaussian unitary operation over $s$ modes.

\medskip

From \cite[Lemma 1]{upreti2025bounding}, ${W}_i\coloneqq {G}_{\mathrm{passive}}^\dagger {\tilde{G}}_{i-1}^\dagger  e^{i\gamma_i\hat{x}_1^3} {\tilde{G}}_{i-1} {G}_{\mathrm{passive}}$ can be put in the form
\begin{equation}
    {W}_i  = {B}_i^\dagger e^{i\gamma_i' \hat x_1^3}{B}_i,
\end{equation}
for all $i$, where each ${B}_i$ is combination of a passive Gaussian unitary operation and a displacement operator over $s$ modes, whose expression can be computed efficiently, as well as the updated cubicity parameter $\gamma_i'$. This gives $\phi$ of the form 
\begin{equation}\label{eq:passivedoped}
    \ket{\phi} = (\Pi_{i=1}^s {B}_i^\dagger e^{i\gamma_i' \hat x_1^3}{B}_i) \ket{0}^{\otimes s}.
\end{equation}

\medskip

We now use the following simulation algorithm, which propagates the input vacuum through the various gates in Eq.~(\ref{eq:passivedoped}): starting with the input coherent state $\ket{0}^{\otimes s}$, each displaced passive Gaussian unitary converts a tensor product of coherent states to a (possibly different) tensor product of coherent states. On the other hand, whenever we encounter the cubic phase gate, we use \cite[Lemma 2]{upreti2025bounding} to approximate the action of the cubic phase gate on a coherent state as a superposition of coherent states:
\begin{lem}[(Approximation of $e^{i\gamma\hat{x}^3}\!\ket{\alpha}$)]\label{lem:csr_decomp_cubic}
     For $\alpha \in \C, \gamma \in \R$ with $\alpha,\gamma = \mathcal O(\poly\!(n))$, let $\ket{\phi_{\alpha,\gamma}} \coloneqq e^{i\gamma\hat{x}^3}\!\ket{\alpha}$, where $\ket{\alpha}$ is a single-mode coherent state of amplitude $\alpha$. For $\varepsilon =\mathcal O( 1/\poly\!(n))$, $\ket{\phi_{\alpha,\gamma}}$ can be approximated to precision $\varepsilon$ in $2$-norm by a finite superposition of coherent states of size $\mathcal O(\poly\!(n))$. Moreover, the superposition and coherent state amplitudes of a corresponding coherent state approximation $\ket{\tilde{\phi}_{\alpha,\gamma}}$ can be computed in $\mathcal O(\poly\!(n))$ time.
\end{lem}
\noindent Repeatedly following these steps over the $s$ layers in Eq.~(\ref{eq:passivedoped}), we obtain a superposition of coherent states ${\tilde \phi}$ which approximates $\phi$. The time complexity of computing ${\tilde \phi}$ and the precision with which it approximates $\phi$ in trace distance are obtained using \cite[Lemma 3]{upreti2025bounding}:
\begin{lem}[(Approximation via coherent state superposition)]\label{lem:csr}
     Given an $s$-mode state $\ket{\phi} = (\Pi_{i=1}^s {B}_i^\dagger e^{i\gamma_i' \hat{x_1}^3}{B}_i) \ket{0}^{\otimes s}$ where $\gamma_1',\dots,\gamma_s'=\mathcal O(\poly\!(n))\in\R$ and where $ B_1,\dots, B_s$ are passive Gaussian unitary gates with entries of their symplectic matrix and displacement vector at most $\mathcal O(\poly\!(n))$, we can obtain the description of a superposition of coherent states $\ket{\tilde{\phi}} = \sum_{i=1}^{N} c_i\ket{\alpha^{(i)}_1\alpha^{(i)}_2\dots\alpha^{(i)}_s}$ where $N = \mathcal O(\poly\!(n)^s)$ in time $\mathcal O(\poly\!(n)^s)$, such that ${\tilde{\phi}}$ is $\mathcal O(1/\poly\!(n))$-close in trace distance to ${\phi}$, and $|\alpha^{(i)}_{s'}|=\mathcal O(\poly\!(n))$ and $|c_i|=\mathcal O(\poly\!(n)^s)$, for all $s'\in\{1,\dots,s\}$ and all $i\in\{1,\dots,N\}$.
\end{lem}
\noindent Note that \cite[Lemma 2]{upreti2025bounding} and \cite[Lemma 3]{upreti2025bounding} are assuming exponential magnitudes for both the cubicity parameter of the cubic phase gates and the displacement amplitudes associated to the Gaussian unitary gates. Further, they require the output state to be approximated to exponential precision. In the present case, we consider the magnitude of the cubicity parameters and displacement amplitudes to be at most polynomially large, and only require the description of ${\phi}$ to be inverse-polynomially precise, and the same proof leads to Lemma~\ref{lem:csr} above. Since Fock basis amplitudes of coherent states can be computed efficiently, this result further implies that the Fock basis amplitudes of the approximation $\tilde{\phi}$ of ${\phi}$ can be computed in time $\mathcal{O}(\poly\!(n)^s)$.

\medskip

Let us now bound the mean photon number of $\tilde\phi$. Writing $\ket{\tilde \phi} = \sum_{i=1}^N c_i \ket{\alpha_{1}^{(i)}\dots\alpha_{m}^{(i)}}$ with $N = \mathcal O(\poly\!(n)^s)$, $|\alpha^{(i)}_{s'}|=\mathcal O(\poly\!(n))$, and $|c_i|=\mathcal O(\poly\!(n)^s)$, for all $s'\in\{1,\dots,s\}$ and all $i\in\{1,\dots,N\}$ as in Lemma~\ref{lem:csr} we obtain
\begin{align}
    \bra{\tilde \phi}\hat{N}\ket{\tilde \phi} &=\sum_{i,j=1}^N c_ic_j^* \bra{\alpha_{1}^{(j)}\dots\alpha_{s}^{(j)}}\hat N \ket{\alpha_{1}^{(i)}\dots\alpha_{s}^{(i)}}\nonumber\\
    &=\sum_{i,j=1}^N\sum_{k=1}^sc_i\alpha_k^{(i)}c_j^*\alpha_k^{(j)*}\braket{\alpha_{k}^{(j)}|\alpha_k^{(i)}}\nonumber\\
    &\le\sum_{k=1}^s\sum_{i,j=1}^N|c_i||\alpha_k^{(i)}||c_j||\alpha_k^{(j)}|\nonumber\\
    &=\mathcal O(\poly\!(n)^s).
\end{align}
Therefore, applying \cref{appthm_simul} with the approximate decomposition $(G,\tilde\phi)$ and $E=\mathcal O(\poly\!(n)^s)$ yields the final complexity of simulation stated in \cref{thm:simucubic}.
\end{proof}


\subsubsection*{Gaussian circuits doped with photon additions and subtractions}

When replacing cubic phase gates by applications of creation and annihilation operators (i.e.~$ a^{\dag j} a^k$), a classical simulation strategy similar to that of the previous section allows to estimate output Gaussian statistics efficiently for a constant number of non-Gaussian layers. In what follows, we refine this result using stellar rank techniques from \cite{chabaud2020classical} to obtain a classical algorithm for \textit{strong simulation} of Gaussian measurements that is efficient for a constant number of non-Gaussian layers. The intuition behind this improved result is that, for this class of circuits, we can take advantage of the fact that the state $\phi$ in the decomposition of \cref{eq:decomppsi} has a finite stellar rank to avoid the complexity induced by the energy truncation.

\begin{theo}[(Strong simulation of Gaussian circuits doped with photon additions and subtractions)]\label{th:simuphotaddsub}
    Given an input $n$-mode vacuum state $\ket{0}^{\otimes n}$ evolving through a (non-unitary) circuit described by
    \begin{equation}
        {C} = {G}_s ({a}_1^\dagger)^{j_s}({a}_1)^{k_s} {G}_{s-1} \dots {G}_1 ({a}_1^\dagger)^{j_1} ({a}_1)^{k_1} {G}_0,
    \end{equation}
   where $j_1,k_1,\dots,j_s,k_s \in \N$, any Gaussian measurement on $\frac1{\sqrt{\mathcal N}}{C}\ket{0}^{\otimes n}$ (where $\mathcal N$ is a normalization factor) can be strongly simulated in time \mbox{$\mathcal{O}\!\left((2sM+1)^{2s}s^3 M^3 2^{2sM}  + (2sM+1)^{s} + \poly\!(n)\right)$}, where $M\coloneqq \max(j_1,k_1,\dots,j_s,k_s)$.
\end{theo}

Note that the creation and annihilation operators act on the first mode without loss of generality, since swapping two modes is a Gaussian unitary operation.

\begin{proof}
Given
\begin{equation}
    {C} = {G}_s ({a}_1^\dagger)^{j_s}({a}_1)^{k_s} {G}_{s-1} \dots {G}_1 ({a}_1^\dagger)^{j_1} ({a}_1)^{k_1} {G}_0,
\end{equation}
and using \cref{theo:efficient_nG}, ${C}$ can be described in time $\mathcal{O}(\poly\!(n))$ as
\begin{equation}
    {C} = {G} (\Pi_{i=1}^{s}{\tilde W}_{i}) {G}_{\mathrm{passive}},
\end{equation}
where ${G}$ is a Gaussian unitary operation, ${G}_\mathrm{passive}$ is a passive linear transformation which can be absorbed into the input vacuum state, and ${\tilde W}_{i}$ is given by
\begin{equation}
    {\tilde W}_{i} = {G}_{\mathrm{passive}}^\dagger {\tilde{G}}_{i-1}^\dagger (a_1^\dagger)^{j_{i}} ({a_1})^{k_{i}}  {\tilde{G}}_{i-1} {G}_{\mathrm{passive}},
\end{equation}
where ${G}_\mathrm{passive}$ is a passive Gaussian unitary over $n$ modes and ${\tilde{G}}_{i-1} = {G}_{i-1}\dots{G}_0$ for all $i\in\{1,\ldots, s\}$, while $\tilde{G}_i{G}_\mathrm{passive}$ is a Gaussian unitary operation over $s$ modes. For brevity, we define 
\begin{equation}
    {\mathcal{G}}_{i-1}\coloneqq{\tilde{G}}_{i-1} {G}_{\mathrm{passive}}.
\end{equation}
Note that the description of ${\mathcal{G}}_{i-1}$ can be computed in time $\mathcal{O}(sn^3) = \mathcal{O}(\poly\!(n))$. Now,
\begin{equation}
    {\mathcal{G}}_{i-1}^\dagger (a_1^\dagger)^{j_{i}} ({a_1})^{k_{i}} {\mathcal{G}}_{i-1} = \left((\bm d_{i})_1 + \sum_{l=1}^s (S_{i})_{1,l} {a}_l^\dagger + (S_{i})_{1,s+l} {a}_l \right)^{j_i} \left((\bm d_{i})_1^* + \sum_{l=1}^s (S_{i})_{1,l} {a}_l + (S_{i})_{1,s+l} {a}_l^\dagger \right)^{k_i},
\end{equation}
where $\bm d_{i}$ and $S_i$ are the displacement vector and symplectic matrix associated with ${\mathcal{G}}_{i-1}$, respectively.
Therefore,
\begin{eqnarray}
    \Pi_{i=1}^{s}{\tilde W}_{i} &=& \left((\bm d_{s})_1 + \sum_{l=1}^s (S_{s})_{1,l} {a}_l^\dagger + (S_{s})_{1,s+l} {a}_l \right)^{j_s} \left((\bm d_{s})_1^* + \sum_{l=1}^s (S_{s})_{1,l} {a}_l + (S_{s})_{1,s+l} {a}_l^\dagger \right)^{k_s} \nonumber \\
    && \hspace{60mm} \vdots \nonumber \\
    &&\hspace{1mm} \left((\bm d_0)_1 + \sum_{l=1}^s (S_{1})_{1,l} {a}_l^\dagger + (S_{1})_{1,s+l} {a}_l \right)^{j_1} \left((\bm d_{0})_1^* + \sum_{l=1}^s (S_{1})_{1,l} {a}_l + (S_{1})_{1,s+l} {a}_l^\dagger \right)^{k_1}.
\end{eqnarray}
Therefore, the circuit can be seen as evolving an input $s$-mode state $(\Pi_{i=1}^{s}{\tilde W}_{i} )\ket{0}^{\otimes n}/\sqrt{\mathcal N}$ of finite support in the Fock basis together with the vacuum on the rest of the $n-s$ modes, through a Gaussian unitary operation acting on $n$ modes. Note that we have have normalized the input state by adding a normalization factor $\mathcal N$, which can be computed by summing the moduli squared of the Fock state amplitudes of the unnormalized state $(\Pi_{i=1}^{s}{\tilde W}_{i} )\ket{0}^{\otimes n}$. This can be done in time $\mathcal{O}(s_{\mathrm{support}})$, where $s_{\mathrm{support}}$ is the support size of the state, given by
\begin{equation}
    s_{\mathrm{support}} = \sum_{i=0}^d \binom{i+s-1}{s-1} = \binom{d+s}{s} \leq (d+1)^s,
\end{equation}
where $d$ is the degree of its stellar function, given by
\begin{equation}
    d = j_1 + k_1 + \dots + j_s + k_s.
\end{equation}
Setting $M \coloneqq \max(j_1,k_1,\dots,j_s,k_s)$, we have
\begin{equation}
    d \leq 2sM, \hspace{5mm} s_\mathrm{support} \leq (2sM+1)^{s}.
\end{equation}
Therefore, using \cite[Theorem 2]{chabaud2020classical}, the heterodyne detection on ${G}(\Pi_{i=1}^{s}{\tilde W}_{i}) \ket{0}^{\otimes n}$ can be strongly simulated in time $\mathcal{O}((2sM+1)^{2s}s^3 M^3 2^{2sM}  + \poly\!(n))$. Adding up the time to compute the description of $\Pi_{i=1}^{s}{\tilde W}_{i}$ from \cref{theo:efficient_nG}, the value of $\mathcal N$, and the description of ${\mathcal G}_0,\dots,{\mathcal G}_{s-1}$, the circuit can be simulated in $\mathcal{O}((2sM+1)^{2s}s^3 M^3 2^{2sM}  + (2sM+1)^{s} + \poly\!(n))$.
\end{proof}


\section{\texorpdfstring{$\varepsilon$}{ε}-approximate symplectic rank}\label{app:appsym}

The symplectic rank of a pure state with a well-defined covariance matrix is the number of symplectic eigenvalues that are strictly greater than exactly one. In particular, for any \(\varepsilon > 0\), a pure state whose symplectic eigenvalues are all equal to \(1+\varepsilon\) has maximum symplectic rank, regardless of how small \(\varepsilon\) is. This presents a conceptual drawback: in practical scenarios, a pure state with all symplectic eigenvalues as close to one as, say, \(1+10^{-100}\), would effectively behave as a Gaussian state (i.e.~with zero symplectic rank), even though it formally has maximum symplectic rank. To account for this issue, we introduce in this section a new non-Gaussianity measure: the \emph{\(\varepsilon\)-approximate symplectic rank}, defined as follows.

\begin{defi}[($\varepsilon$-approximate symplectic rank)]\label{def_eps_approx}
Let $\rho$ be a quantum state and $\varepsilon\in(0,1)$. The $\varepsilon$-approximate symplectic rank of $\rho$, denoted as $\mathfrak{s}_{\varepsilon}(\rho)$, is defined as the minimum symplectic rank among all the states that are $\varepsilon$-close to $\rho$ in infidelity. In formula,
\begin{equation}
    \begin{aligned}
        \mathfrak{s}_{\varepsilon}(\rho)\coloneqq \min_{\substack{\tau:\\1-F(\rho,\tau)\le \varepsilon}} \mathfrak{s}(\tau)\,.
    \end{aligned}
\end{equation}
\end{defi}

Like for the symplectic rank, the fact that we can define the approximate symplectic rank as a $\min$ rather than an $\inf$ stems from the fact that it is integer-valued and bounded.

By definition of symplectic rank, it follows that for any state $\rho$ and any $\varepsilon>0$ there exists an $\mathfrak{s}_{\varepsilon}(\rho)$-compressible state which is $\varepsilon$-close to $\rho$. Moreover, all the states that are $\varepsilon$-close to $\rho$ cannot be compressed in a number of modes smaller than $\mathfrak{s}_{\varepsilon}(\rho)$. Additionally, in \cref{def_eps_approx}, we used infidelity to quantify the distance between quantum states. Similar definitions can be formulated using other metrics that are non-increasing under quantum channels, such as the trace distance.


\subsection{Monotonicity of the approximate symplectic rank}

We have seen that the symplectic rank is non-increasing under post-selected Gaussian operations. In the forthcoming \cref{thm_main_sm_eps}, we show that the $\varepsilon$-approximate symplectic rank is non-increasing under the smaller class of \emph{Gaussian operations}. By definition, a Gaussian operation is a composition of the following five building blocks~\cite{giedke2002characterization}:
\begin{itemize}
    \item tensoring with a Gaussian state;
    \item applying a Gaussian unitary operation;
    \item performing (non-Gaussian) classical mixing;
    \item performing an heterodyne measurement;
    \item taking a partial trace.
\end{itemize}

\begin{theo}[(Monotonicity of the $\varepsilon$-approximate symplectic rank)]\label{thm_main_sm_eps}
    For any $\varepsilon\in[0,1]$, the $\varepsilon$-approximate symplectic rank is non-increasing under Gaussian operations. That is, for any $\rho$ and any Gaussian operation $G$, it holds that 
    \be
        \mathfrak{s}_\varepsilon(G(\rho))\le\mathfrak{s}_\varepsilon(\rho)\,.
    \ee
\end{theo}

\begin{proof}
In \cref{thm_main_sm}, we have already proven that the symplectic rank is non-increasing under post-selected Gaussian operations, and in particular under Gaussian operations. Let us generalise this result for $\varepsilon\in(0,1]$.

We need to prove that for any $\varepsilon\in(0,1]$, any $\rho$, and any Gaussian operation $G$, it holds that 
    \be
        \mathfrak{s}_\varepsilon(G(\rho))\le\mathfrak{s}_\varepsilon(\rho)\,.
    \ee
By definition of $\mathfrak{s}_\varepsilon(\rho)$, there exists $\tau$ such that $1-F(\rho,\tau)\le \varepsilon$ and $\mathfrak{s}_\varepsilon(\rho)=\mathfrak{s}(\tau)$. Hence, it holds that
\be
    \mathfrak{s}_\varepsilon(G(\rho))\leqt{(i)}  \mathfrak{s}(G(\tau))\leqt{(ii)} \mathfrak{s}(\tau)\eqt{(iii)}\mathfrak{s}_\varepsilon(\rho)\,.
\ee
Here, in (i), we exploited that the infidelity is non-increasing under quantum channels~\cite{MARK,NC} to observe that 
\be 
    1-F(G(\rho),G(\tau))\le 1-F(\rho,\tau) \le \varepsilon\,,
\ee
and thus $G(\tau)$ is a feasible point in the minimum problem that defines $\mathfrak{s}_\varepsilon(G(\rho))$. In (ii), we exploited that the symplectic rank is non-increasing under Gaussian operations, as a consequence of \cref{thm_main_sm}. Finally, in (iii), we used the definition of $\tau$.
\end{proof}


\subsection{Relation with symplectic eigenvalues}

We have seen that the symplectic rank of a pure state is given by the number of symplectic eigenvalues of its covariance matrix that are strictly larger than one. How this result can be generalized to the approximate setting? The forthcoming lemma provides such a generalization.

\begin{lem}(A large number of symplectic eigenvalues far from one imply a large $\varepsilon$-approximate symplectic rank)
    Let $\varepsilon\in[0,1]$ and let $\psi$ be a pure state over $n$ modes. 
    Let $t$ be the number of symplectic eigenvalues of the covariance matrix of $\psi$ that are larger than $1+\frac{2\varepsilon}{n}$. Then, the $\varepsilon$-approximate symplectic rank of $\rho$ is smaller or equal to $t$:
    \be
        \mathfrak{s}_\varepsilon(\psi)\le t\,.
    \ee
\end{lem}

\begin{proof}
Let $t$ be the number of symplectic eigenvalues of the covariance matrix of $\psi$ that are larger than $1+\varepsilon'$. The Williamson decomposition of the covariance matrix $V\!(\psi)$ is of the form
\begin{equation}
    V\!(\psi)=S(D_t\oplus d_{n-t})S^\intercal\,,
\end{equation}
where $d_{n-t}$ is a $2(n-t)$-dimensional diagonal matrix satisfying $\|d_{n-t}-\mathbb{1}\|_\infty\le \varepsilon'$. Hence, it follows that there exists a Gaussian unitary operation $G$ such that the covariance matrix of $G^\dagger \psi G$ satisfies 
\be
    V(G^\dagger \psi G )= D_t\oplus d_{n-t}\,,
\ee
and the first moment is zero
\be
    \bm{m}(G^\dagger \psi G)=0\,.
\ee
In particular, the reduced state $\rho$ on the last $n-t$ modes, defined as
\be
    \rho\coloneqq \Tr_{[t]}\left[G^\dagger \psi G\right]\,,
\ee
satisfies $V(\rho)=d_{n-t}$ and $\bm{m}(\rho)=0$. The idea now is that since $d_{n-t}$ is close to the identity and since the only state with a covariance matrix equal to the identity is the vacuum state, the state $\rho$ has to be close to the vacuum. 
   By defining the $t$-mode state
    \begin{equation}
        \ket{\phi}\coloneqq \frac{(\mathbb{1}_t\otimes \bra{0}^{\otimes (n-t)})\,G^\dagger\ket\psi}{\sqrt{\bra\psi G(\mathbb{1}_t\otimes \ketbra{0}^{\otimes (n-t)})\,G^\dagger\ket\psi}}\,,
    \end{equation}
    let us consider the $t$-compressible state $G\left(\ket{\phi}\otimes \ket{0}^{\otimes(n-t)}\right)$. The infidelity between such a state and $\psi$ satisfies
    \be
        1-F\!\left(\psi,G(\phi\otimes \ketbra{0}^{\otimes(n-t)})G^\dagger \right)
        &=1- \Tr\!\left[G^\dagger\psi G  \,\left(\phi\otimes \ketbra{0}^{\otimes(n-t)}\right)\right]\\
        &=1-\bra\psi G(\mathbb{1}_t\otimes \ketbra{0}^{\otimes (n-t)})\,G^\dagger\ket\psi\\
        &=\Tr\!\left[\rho \left(\mathbb{1}-\ketbra{0}^{\otimes (n-t)}\right)\right]\\
        &\leqt{(i)} \Tr[\rho \hat{N}_{n-t}  ]\\
        &\eqt{(ii)}\frac{\Tr[V(\rho)-\mathbb{1}]}{4}+\frac{\|\bm{m}(\rho)\|}{2}\\
        &=\frac{\Tr[d_{n-t}-\mathbb{1}]}{4}\\
        &\le \frac{(n-t)\varepsilon'}{2}\\
        &\le \frac{n \varepsilon'}{2}\,.
    \ee
Here, in (i) we introduced the photon number operator on the last $n-t$ modes and we used the operator inequality $\ket{0}\!\!\bra{0}^{\otimes n}\ge \mathbb{1}-\hat{N}_{n-t}$, and in (ii) we exploited the general formula for the mean photon number of a state in Eq.~\eqref{formula_mean_energy}. Hence, by choosing $\varepsilon'\coloneqq \frac{2\varepsilon}{n}$, we obtain that the $t$-compressible state $G\left(\ket{\phi}\otimes \ket{0}^{\otimes(n-t)}\right)$ is $\varepsilon$-close to $\psi$, and in particular we conclude that $\mathfrak{s}_\varepsilon(\psi)\le t$.
\end{proof}

The above lemma establishes that if the number of symplectic eigenvalues of $V(\psi)$ that are sufficiently far away from one is $t$, then the approximate symplectic rank is upper bounded by $t$. Can a similar lower bound exist? Or, equivalently, is it possible to express the approximate symplectic rank in terms of the number of symplectic eigenvalues that are sufficiently far away from one? Let us observe that this is not possible in general, as there exist pure quantum states which are very close but with very different symplectic eigenvalues. For example, the two states $\ket{0}$ and $\sqrt{1-\varepsilon^2}\ket{0}+\varepsilon\ket{\lceil\varepsilon^{-1}\rceil} $ are $O(\varepsilon)$-close, but the symplectic eigenvalue associated with the former is equal to $1$ while the symplectic eigenvalue associated with the latter is $2+O(\varepsilon)$. Consequently, even if the symplectic eigenvalue of $\sqrt{1-\varepsilon^2}\ket{0}+\varepsilon\ket{\lceil\varepsilon^{-1}\rceil}$ is very far away from one, its $O(\varepsilon)$-approximate symplectic rank is zero. 

This reveals a peculiar discontinuity of symplectic eigenvalues: even if the state $\sqrt{1-\varepsilon^2}\ket{0}+\varepsilon\ket{\lceil\varepsilon^{-1}\rceil} $ converges (in trace distance) to $\ket{0}$ for $\varepsilon\rightarrow 0^+$, the symplectic eigenvalue of the former converges to $2$ and the symplectic eigenvalue of the latter is exactly \(1\). This discontinuity arises because the second moment of the energy of the state \(\sqrt{1-\varepsilon^2}\ket{0}+\varepsilon\ket{\lceil\varepsilon^{-1}\rceil}\) diverges as \(\varepsilon \to 0^+\). Indeed, as a consequence of \cref{pert_bound_symp} in \cref{sec_pert_bounds} below, symplectic eigenvalues are continuous with respect to the state (in trace distance) for states with a \emph{finite} second moment of the energy.


\subsection{Symplectic fidelities}

In this section, we introduce a notion of symplectic fidelities, similar to stellar fidelities for the stellar rank~\cite{hahn2024assessing}.
\begin{defi}[(\texorpdfstring{$k$}{k}-symplectic fidelity)]
Given a quantum state $\rho$ over $n$ modes and $0\le k\le n$, its $k$-symplectic fidelity is defined as
\begin{equation}
    \begin{aligned}
        f_k^{\mathfrak s}(\rho)\coloneqq\sup_{\substack{\tau:\\\mathfrak{s}(\tau)\le k}}F(\rho,\tau),
    \end{aligned}
\end{equation}
where the supremum is taken over all the states $\tau$ with symplectic rank of at most $k$, and $F$ denotes the fidelity.
\end{defi}
It is important to remark that, by definition, the $k$-symplectic fidelity equals $1$ whenever $k \ge \mathfrak{s}(\rho)$, since we are considering the ball of states with symplectic rank less than $k$, which includes $\rho$ itself.

In general, it is possible to connect symplectic fidelities and approximate symplectic rank as follows:

\begin{lem}[(Equivalence between approximate symplectic rank and symplectic fidelities)]\label{lem:asr-sF:supp}
Let $\rho$ be a state over $n$ modes.
For all $k\le n$ and all $0\le\varepsilon<1$,
\begin{equation}
    f^{\mathfrak s}_k(\rho)\ge1-\varepsilon\Leftrightarrow\forall\varepsilon'>\varepsilon,\,\mathfrak{s}_{\varepsilon'}(\rho)\le k,
\end{equation}
which implies, for all $1\le k<\mathfrak{s}(\rho)$,
\begin{equation}
    \mathfrak{s}_\varepsilon(\rho) =
    \begin{cases} 
    \mathfrak{s}(\rho) & \!\!\!\text{for } 
    \varepsilon\in(0,1-f^{\mathfrak s}_{\mathfrak{s}(\rho)-1}(\rho)),\\
    k & \!\!\!\text{for } 
    \varepsilon\in(1-f^{\mathfrak s}_k(\rho),1-f^{\mathfrak s}_{k-1}(\rho)),\\
    0 & \!\!\!\text{for } \varepsilon\in[1-f^{\mathfrak s}_0(\rho),1].
    \end{cases}
\end{equation}
\end{lem}

\begin{proof}
Let us start by proving the following: for all $k\le n$ and all $0\le\varepsilon\le1$,
\be\label{eq:implication1}
    \mathfrak{s}_\varepsilon(\rho)\le k\Rightarrow f^{\mathfrak s}_k(\rho)\ge1-\varepsilon.
\ee
Indeed, if the approximate symplectic rank satisfies $\mathfrak{s}_\varepsilon(\rho)\le k$, then there exists a state $\tau$ such that $1 - F(\rho, \tau) \le \varepsilon$ and $\mathfrak{s}(\tau)\le k$. By the definition of $\sup$, we have
\be 
f_k^{\mathfrak s}(\rho)\ge F(\rho,\tau)\ge 1-\varepsilon.
\ee 

Now by the definition of symplectic fidelities, assuming $f^{\mathfrak s}_k(\rho)\ge1-\varepsilon$ implies that for all $\varepsilon'>\varepsilon$, there exists a state $\tau$ such that $F(\rho, \tau) > 1 - \varepsilon'$ and $\mathfrak{s}(\tau) \le k$, which also implies that the $\varepsilon'$-approximate symplectic rank of $\rho$ is at most $k$ for all $\varepsilon'>\varepsilon$. We thus obtain, for all $k\le n$ and all $0\le\varepsilon\le1$,
\begin{align}
    f^{\mathfrak s}_k(\rho)\ge1-\varepsilon&\Rightarrow\forall\varepsilon'>\varepsilon,\,\mathfrak{s}_{\varepsilon'}(\rho)\le k\nonumber\\
    &\Rightarrow\forall\varepsilon'>\varepsilon,\,f^{\mathfrak s}_k(\rho)\ge1-\varepsilon'\nonumber\\
    &\Rightarrow f^{\mathfrak s}_k(\rho)\ge1-\varepsilon,
\end{align}
where we have used \cref{eq:implication1} in the second line. This shows that all implications are equivalences and completes the proof.
\end{proof}

It is interesting to point out that if the state $\rho$ is pure, the problem becomes simpler, as the symplectic fidelities can be computed via an optimization over Gaussian unitary operations. 

\begin{lem}[(Computing symplectic fidelities)]\label{appth:compsympfid}
    Let $\psi$ be a pure state over $n$ modes. For all $k\le n$, its symplectic fidelities can be expressed as
    \be
    f^{\mathfrak s}_k(\psi)=\sup_{ G}\Tr[({\mathbb1}_k\otimes\ket0\!\bra0^{\otimes(n-k)}) G^\dag\ket\psi\!\bra\psi G],
    \ee
    where the optimization is over all the possible Gaussian unitaries $G$. 
\end{lem}
\begin{proof}
For pure states, since the fidelity becomes a linear function and since the set of states of symplectic rank bounded by $k$ is a convex set whose extreme points are pure states of symplectic rank bounded by $k$, it is sufficient to run the optimization over pure states.
Moreover, by definition of symplectic rank, a pure state ${\tau}$ satisfies $\mathfrak s({\tau})\le k$ if and only if it can be written as
\be
    \tau=G\left(\phi\otimes\ketbra{0}^{\otimes(n-k)}\right)G^\dagger
\ee
for a suitable $n$-mode Gaussian unitary $G$ and a suitable $k$-mode state $\phi$. Hence, by introducing the $k$-mode state
\be
        \ket{\phi_G}\coloneqq\frac{(\mathbb{1}_k\otimes \bra{0}^{\otimes (n-k)})\,G^\dagger\ket\psi}{\sqrt{\Tr[({\mathbb1}_k\otimes\ket0\!\bra0^{\otimes(n-k)}) G^\dag\ket\psi\!\bra\psi G] }}\,,
\ee 
the symplectic fidelity can be rewritten as
\be
    f^{\mathfrak s}_{k}(\psi) & = \sup_{G,\phi}F\!\left(\psi, G\left(\phi\otimes\ketbra{0}^{\otimes(n-k)}\right)G^\dagger\right)\\
    &=\sup_{G,\phi}F\!\left(G^\dagger\psi G, \phi\otimes\ketbra{0}^{\otimes(n-k)}\right)\\
    &=\sup_{G,\phi} \Tr[({\mathbb1}_k\otimes\ket0\!\bra0^{\otimes(n-k)}) G^\dag\ket\psi\!\bra\psi G]  F\!\left(\phi_G, \phi \right)\\
    &=\sup_{G} \Tr[({\mathbb1}_k\otimes\ket0\!\bra0^{\otimes(n-k)}) G^\dag\ket\psi\!\bra\psi G]\,,
\ee 
which concludes the proof.
\end{proof}

In particular, $f^{\mathfrak s}_0$ is the fidelity with the closest Gaussian pure state. Together with  \cref{lem:asr-sF:supp}, this result implies that the approximate symplectic rank can often be computed efficiently for a small number of modes, as it can be inferred from the list of symplectic fidelities, which are obtainable through the above optimization for systems with a small number of modes. For many modes, however, we expect the optimization to be more challenging and leave a detailed numerical study to future work.


\subsection{Approximate Gaussian conversion}

We have seen that the approximate symplectic rank is non-increasing under Gaussian operations. This implies that the approximate symplectic rank may also be used to give bounds on Gaussian protocols for non-Gaussian state conversion \cite{Albarelli2018}: if $\mathcal G(\rho^{\otimes k})=\sigma^{\otimes m}$ for some Gaussian protocol $\mathcal G$, then
\be\label{eq:exact_approxGconvapp1}
\mathfrak{s}_\varepsilon(\sigma^{\otimes m})\le \mathfrak{s}_\varepsilon(\rho^{\otimes k}).
\ee
In this section, we show that we can extend these bounds to the case of approximate Gaussian conversion, similar to the approximate stellar rank \cite{hahn2024assessing}:

\begin{theo}[(Gaussian conversion bounds based on the approximate symplectic rank)]\label{th:conversion}
    If $\mathcal G(\rho^{\otimes k})$ is $\delta$-close in trace distance to $\sigma^{\otimes m}$, then
\be\label{eq:approxGconvapp2}
\mathfrak{s}_{\varepsilon+\sqrt{2\delta}}(\sigma^{\otimes m})\le \mathfrak{s}_\varepsilon(\rho^{\otimes k}),
\ee
for all $\varepsilon>0$. For pure states $\rho=\Psi$ and $\sigma=\Phi$ this can be tightened as
\be\label{eq:approxGconvpureapp}
\mathfrak{s}_{\varepsilon+\delta}(\Phi^{\otimes m})\le \mathfrak{s}_\varepsilon(\Psi^{\otimes k}).
\ee
These bounds can be phrased in terms of symplectic fidelities as
\be\label{eq:boundfsrsupp}
f_n^{\mathfrak s}(\rho^{\otimes k})\le f_n^{\mathfrak s}(\sigma^{\otimes m})+\sqrt{2\delta},
\ee
for all $n\in\mathbb N$, and, for pure states $\rho=\Psi$ and $\sigma=\Phi$,
\be
f_n^{\mathfrak s}(\Psi^{\otimes k})\le f_n^{\mathfrak s}(\Phi^{\otimes m})+\delta.
\ee
Moreover, when allowing for post-selection, if $\mathcal G(\rho^{\otimes k})$ is $\delta$-close in trace distance to $\sigma^{\otimes m}$ with any nonzero probability, then for all $\varepsilon>0$,
\be\label{eq:approxGconvapp2_ps}
\mathfrak{s}_{\sqrt{2\delta}}(\sigma^{\otimes m})\le \mathfrak{s}(\rho^{\otimes k}),
\ee
and for pure states:
\be\label{eq:approxGconvpureapp_ps}
\mathfrak{s}_{\delta}(\Phi^{\otimes m})\le \mathfrak{s}(\Psi^{\otimes k})=k\mathfrak{s}(\Psi).
\ee
\end{theo}

\begin{proof} 
Let $\omega$ be a density operator which can be obtained from $\rho^{\otimes k}$ using Gaussian operations, and such that \mbox{$\frac12\|\omega-\sigma^{\otimes m}\|_1\le\delta$}.
We make use the following result \cite[Lemma 5]{rastegin2003lower}: 

\begin{lem}\label{lem:diffFid}
For density operators $\rho$, $\sigma$ and $\tau$, we have
\begin{equation}
    \left|F(\rho,\tau)-F(\sigma,\tau)\right|\le\sqrt{\|\rho-\sigma\|_1}.
\end{equation}
Moreover, if $\rho=\Psi$ and $\sigma=\Phi$ are pure states,
\begin{equation}\label{eq:diffF}
    \left|F(\Psi,\tau)-F(\Phi,\tau)\right|\le \frac12\|\Psi-\Phi\|_1.
\end{equation}
\end{lem}

 \cref{lem:diffFid} implies that the set $\{\tau|F(\sigma^{\otimes m},\tau)\ge1-(\varepsilon+\sqrt{2\delta})\}$ contains the set $\{\tau|F(\omega,\tau)\ge1-\varepsilon\}$, so
\begin{equation}\label{eq:conv-bound-asr-inter}
    \begin{aligned}
        \mathfrak s_{\varepsilon+\sqrt{2\delta}}(\sigma^{\otimes m})&=\min_{\substack{\tau:\\F(\sigma^{\otimes m},\tau)\ge1-(\varepsilon+\sqrt{2\delta})}}\mathfrak s(\tau)\\
        &\le\min_{\tau,F(\omega,\tau)\ge1-\varepsilon}\mathfrak s(\tau)\\
        &=\mathfrak s_\varepsilon(\omega).
    \end{aligned}
\end{equation}
By the exact Gaussian conversion bound in \cref{eq:exact_approxGconvapp1}, we have $\mathfrak s_\varepsilon(\rho^{\otimes k})\ge\mathfrak s_\varepsilon(\omega)$, which together with \cref{eq:conv-bound-asr-inter} implies the conversion bound in \cref{eq:approxGconvapp2}: if $m$ copies of $\sigma$ can be obtained from $k$ copies of $\rho$ using Gaussian operations, then for all $0\le\varepsilon\le1$:
\begin{equation}
    \mathfrak s_\varepsilon(\rho^{\otimes k})\ge\mathfrak s_{\varepsilon+\sqrt{2\delta}}(\sigma^{\otimes m}).
\end{equation}
We now rephrase this bound in terms of symplectic fidelities using  \cref{lem:asr-sF:supp}. Since the approximate symplectic rank is integer-valued, this is equivalent to
\begin{equation}
    \mathfrak s_\varepsilon(\rho^{\otimes k})\le n\Rightarrow \mathfrak s_{\varepsilon+\sqrt{2\delta}}(\sigma^{\otimes m})\le n,
\end{equation}
for all $n\in\mathbb N$, or equivalently, if $\mathcal G(\rho^{\otimes k})$ is $\delta$-close in trace distance to $\sigma^{\otimes m}$, then
\begin{equation}\label{eq:impl-asr}
    \mathfrak s_\varepsilon(\rho^{\otimes k})\le n\Rightarrow \mathfrak s_{\varepsilon+\sqrt{2\delta}}(\sigma^{\otimes m})\le n,
\end{equation}
for all $n\in\mathbb N$ and all $0\le\varepsilon\le1$.

Now, by  \cref{lem:asr-sF:supp}, for all $n\in\mathbb N$ and all $0\le\varepsilon<1$,
\begin{equation}\label{eq:impl-asr2}
    1-f^\mathfrak s_n(\rho^{\otimes k})<\varepsilon\Rightarrow\mathfrak s_\varepsilon(\rho^{\otimes k})\le n,
\end{equation}
and by \cref{eq:implication1}, for all $n\in\mathbb N$ and all $0\le\varepsilon<1-\sqrt{2\delta}$
\begin{equation}\label{eq:impl-asr3}
    \mathfrak s_{\varepsilon+\sqrt{2\delta}}(\sigma^{\otimes m})\le n\Rightarrow1-f^\mathfrak s_n(\sigma^{\otimes m})\ge\varepsilon+\sqrt{2\delta}.
\end{equation}
Hence, if $\mathcal G(\rho^{\otimes k})$ is $\delta$-close in trace distance to $\sigma^{\otimes m}$, then for all $n\in\mathbb N$ and all $0\le\varepsilon<1-\sqrt{2\delta}$ the following chain of implications holds:
\begin{align}
    1-f^\mathfrak s_n(\rho^{\otimes k})<\varepsilon&\Rightarrow\mathfrak s_\varepsilon(\rho^{\otimes k})\le n\nonumber\\
    &\Rightarrow\mathfrak s_{\varepsilon+\sqrt{2\delta}}(\sigma^{\otimes m})\le n\nonumber\\
    &\Rightarrow1-f^\mathfrak s_n(\sigma^{\otimes m})\ge\varepsilon+\sqrt{2\delta},
\end{align}
where we used \cref{eq:impl-asr2} in the first line, \cref{eq:impl-asr} in the second line, and \cref{eq:impl-asr3} in the last line. This implies $1-f^\mathfrak s_n(\rho^{\otimes k})\ge1-f^\mathfrak s_n(\sigma^{\otimes m})-\sqrt{2\delta}$ for all $n\in\mathbb N$, and thus
\begin{equation}
f^\mathfrak s_n(\rho^{\otimes k})\le f^\mathfrak s_n(\sigma^{\otimes m})+\sqrt{2\delta},
\end{equation}
for all $n\in\mathbb N$, completing the proof of \cref{eq:boundfsrsupp}.

Finally, the proof of \cref{eq:approxGconvapp2_ps} is analogous, using the fact that the exact symplectic rank is non-increasing under post-selected Gaussian operations (see \cref{thm_main_sm}). Moreover, the proofs are identical when $\rho=\Psi$ and $\sigma=\Phi$ are pure states, by replacing $\sqrt{2\delta}$ with $\delta$.

\end{proof}


\section{Perturbation bounds for continuous-variable quantum states}\label{sec_pert_bounds}


This section is devoted to establishing lower bounds on the trace distance between two (possibly non-Gaussian) states in terms of the norm distance of their first moments and covariance matrices, which we regard as new technical tools of independent interest for the field of continuous-variable quantum information. These new bounds are summarized in the forthcoming \cref{proof_lower_bound}, which is proven below in \cref{sec_proofs}. 
\begin{theo}[(Lower bounds on the trace distance between non-Gaussian states)]\label{proof_lower_bound}
Let $\rho$ and $\sigma$ be (possibly non-Gaussian) quantum states. Let $\bm{m}(\rho),V(\rho)$ be the first moment and the covariance matrix of $\rho$, and let $\bm{m}(\sigma),V(\sigma)$ be the first moment and the covariance matrix of $\sigma$. Then, the trace distance between $\rho$ and $\sigma$ can be lower bounded as follows:
\begin{align}
    \|\rho-\sigma\|_1&\ge \frac{\|\bm{m}(\sigma)-\bm{m}(\rho)\|^2}{16 \max\!\left({\Tr\!\left[\hat{E}  \rho \right]},{\Tr\!\left[\hat{E}\sigma\right]}\right)}\,,\label{ineq_first}\\
        \|\rho-\sigma\|_1&\ge \frac{\|V(\rho)-V(\sigma)\|_{\infty}^2}{1549\,\max(\Tr[\hat{E}^2\rho],\Tr[\hat{E}^2\sigma])}\label{ineq_cov}\,,
\end{align}
Here, $\hat{E}$ denotes the energy operator; given a matrix $V$, $|V|\coloneqq \sqrt{V^\dagger V}$ denotes its modulus, $\|V\|_\infty$ denotes the operator norm (i.e.~the maximum singular values), and $\|V\|_1\coloneqq\Tr|V|$ denotes the trace norm. Furthermore, given a vector $\textbf{m}$, $\|\bm{m}\|\coloneqq \sqrt{\bm{m}^\intercal\bm{m}}$ denotes its Euclidean norm. 
\end{theo}
The bounds stated in \cref{proof_lower_bound} contribute to the rapidly growing literature regarding \emph{bosonic trace distance bounds}~\cite{mele2024learning,holevo2024estimatestracenormdistancequantum,holevo2024estimatesburesdistancebosonic,fanizza2024efficienthamiltonianstructuretrace,mele2025achievableratesnonasymptoticbosonic,bittel2024optimaltracedistanceboundsfreefermionic}. The latter constitute a technical toolbox of independent interest~\cite{holevo2024estimatestracenormdistancequantum, holevo2024estimatesburesdistancebosonic}. Indeed, they can be useful e.g.~in analysing tomography of Gaussian states~\cite{mele2024learning}, testing properties of Gaussian systems~\cite{bittel2024optimaltracedistanceboundsfreefermionic}, learning Gaussian processes~\cite{Haah_2023}, and characterizing quantum Gaussian observables~\cite{holevo2024estimatestracenormdistancequantum}.

\subsection{Bosonic trace distance bounds}
Let us give a brief literature review regarding bosonic trace distance bounds~\cite{mele2024learning,holevo2024estimatestracenormdistancequantum,holevo2024estimatesburesdistancebosonic,fanizza2024efficienthamiltonianstructuretrace,mele2025achievableratesnonasymptoticbosonic}. A central question in this line of research is how to bound the trace distance between two continuous-variable quantum states $\rho$ and $\sigma$ in terms of the distance between their first moments, $\|\textbf{m}(\rho)-\textbf{m}(\sigma)\|$, and their covariance matrices, $\|V(\rho)-V(\sigma)\|$. This question is of fundamental importance because, although the trace distance plays a key role in all the areas of quantum information science~\cite{HOLEVO}, it is generally hard to compute for continuous-variable systems. In contrast, first moments and covariance matrices are typically much easier to estimate, both analytically and experimentally. The answer to this question depends critically on the assumptions made about the states $\rho$ and $\sigma$; various bounds have been derived under different assumptions regarding the Gaussianity or the purity of the states:
\begin{itemize}
    \item \textbf{Upper bound on the trace distance between Gaussian states}: It is well known that Gaussian states are uniquely determined by their first moments and covariance matrices. This naturally raises the fundamental question: \emph{If two Gaussian states have first moments and covariance matrices that are $\varepsilon$-close, how close are the states in trace distance?} This problem was first introduced in \cite{mele2024learning}, where partial results were obtained. Later, Holevo~\cite{holevo2024estimatestracenormdistancequantum} and Fanizza et al.~\cite{fanizza2024efficienthamiltonianstructuretrace} made further progress on this problem. However, neither the findings of~\cite{mele2024learning} nor those of Holevo~\cite{holevo2024estimatestracenormdistancequantum} and Fanizza et al.~\cite{fanizza2024efficienthamiltonianstructuretrace} have fully resolved the question of establishing a tight upper bound on the trace distance. Remarkably, this question has been recently fully solved by~\cite{bittel2024optimalestimatestracedistance}, where the following tight bound has been established:

    \emph{Let $\rho$ be the Gaussian state with first moment $\bm{m}$ and covariance matrix $V$. Similarly, let $\sigma$ be the Gaussian state with first moment $\bm{t}$ and covariance matrix $W$. Then, the following the trace norm between $\rho$ and $\sigma$ satisfies the following tight upper bounds~\cite{bittel2024optimalestimatestracedistance}:
    \be\label{ineqqqq}
        \left\|\rho-\sigma\right\|_1 \le \frac{1+\sqrt{3}}{4}\Tr\!\left[|V-W|\,\Omega^\intercal \left(\frac{V+W}{2}\right)\!\Omega\right] + \sqrt{2\min(\|V\|_\infty,\|W\|_\infty)}\, \|\bm{m}-\bm{t}\|\, .
    \ee    
    Moreover, note that H\"older's inequality implies that the term $\Tr\!\left[|V-W|\Omega^\intercal\left(\frac{V+W}{2}\right)\Omega\right]$ in Eq.~\eqref{ineqqqq} can be further upper bounded either by \( \max(\|V\|_\infty,\|W\|_\infty) \|V-W\|_1 \), or by $\max(\Tr V, \Tr W) \|V-W\|_\infty $.}

In the fermionic setting~\cite{aaronson2023efficienttomographynoninteractingfermion,bittel2024optimaltracedistanceboundsfreefermionic}, a similar inequality has been established first in the particular case of pure states~\cite{aaronson2023efficienttomographynoninteractingfermion} and later in the general case of mixed states~\cite{bittel2024optimaltracedistanceboundsfreefermionic}.

    \item \textbf{Lower bound on the trace distance between Gaussian states}: In~\cite[Eq.~S265]{mele2024learning} the following lower bound on the trace distance between Gaussian states was established:

    \emph{Let $\rho$ be the Gaussian state with first moment $\bm{m}$ and covariance matrix $V$. Similarly, let $\sigma$ be the Gaussian state with first moment $\bm{t}$ and covariance matrix $W$. Then, the trace distance between $\rho$ and $\sigma$ can be lower bounded as~\cite[Eq.~S265]{mele2024learning}:
\be\label{thm_trace_distance_lower_bound}
   \|\rho-\sigma\|_1&\ge \frac{1}{100}\min\!\left\{1,\frac{\| \bm{m}-\bm{t} \|}{\sqrt{4\min(\|V\|_\infty,\|W\|_\infty)+1}}\right\} \,,\\
   \|\rho-\sigma\|_1&\ge \frac{1}{100}\min\!\left\{1,\frac{\| V-W \|_2}{4\min(\|V\|_\infty,\|W\|_\infty)+1}\right\} \,.
\ee
}

    \item \textbf{Upper bound on the trace distance between a non-Gaussian state and a pure Gaussian state}: In~\cite{mele2024learning} the following upper bound on the trace distance between a non-Gaussian state and a pure Gaussian state was established:

    \emph{Let $\psi$ be a pure Gaussian state and let $\rho$ be a (possibly non-Gaussian) state. Then, the trace distance between $\psi$ and $\rho$ can be upper bounded as~\cite{mele2024learning}
    \be\label{upper_bound_d_tr_pure}
        &\frac12\|\rho-\psi\|_1 \le \sqrt{\max\!\left({\Tr\!\left[\hat{E}  \rho \right]},{\Tr\!\left[\hat{E}\sigma\right]}\right)}\sqrt{ 2\|\textbf{m}(\rho)-\textbf{m}(\psi)\|^2+\|V\!(\rho)-V\!(\psi)\|_\infty}\,,
    \ee
    where $\hat{E}$ denotes the energy operator.}
    \item \textbf{Lower bound on the trace distance between non-Gaussian states}: In this work --- specifically in \cref{proof_lower_bound} --- we establish lower bounds on the trace distance between two possibly non-Gaussian states $\rho$ and $\sigma$:
   \be
    \|\rho-\sigma\|_1&\ge \frac{\|\bm{m}(\sigma)-\bm{m}(\rho)\|^2}{16 \max\!\left({\Tr\!\left[\hat{E}  \rho \right]},{\Tr\!\left[\hat{E}\sigma\right]}\right)}\,,\\
        \|\rho-\sigma\|_1&\ge \frac{\|V(\rho)-V(\sigma)\|_{\infty}^2}{1549\,\max(\Tr[\hat{E}^2\rho],\Tr[\hat{E}^2\sigma])}\label{ineq_cov1}\,.
   \ee
    It is worth noting that this trace distance lower bound scales quadratically with the norm distance between the covariance matrices, in contrast to the lower bound for Gaussian states in \cref{thm_trace_distance_lower_bound}, which exhibits linear scaling. Nevertheless, this quadratic behaviour is shown to be tight, as demonstrated in \cref{thm_pert_V} below.
\end{itemize}
Another related result is presented in~\cite{mele2025achievableratesnonasymptoticbosonic}, where an \textbf{algorithm to compute the trace distance between two Gaussian states up to a fixed precision} has been developed. The construction of this algorithm relies on two key steps:
\begin{itemize}
    \item First, Ref.~\cite{mele2025achievableratesnonasymptoticbosonic} derives a simple bound on the probability \( P_{> M} \) that a given Gaussian state has more than \( M \) photons, establishing that such a probability decays exponentially in \( M \). The result is as follows~\cite{mele2025achievableratesnonasymptoticbosonic}:

\emph{Let \(\rho\) be an \(n\)-mode Gaussian state, and let \(P_{> M}\) denote the probability that \(\rho\) has more than \(M\) photons. According to Born's rule, this probability is defined by}
\be
    P_{> M} \coloneqq \Tr[(\mathbb{1} - \Pi_M)\rho]\,,
\ee
\emph{where \(\Pi_M\) is the projector onto the subspace spanned by all \(n\)-mode Fock states with total photon number less than or equal to \(M\):}
\be\label{proj_fockmain}
    \Pi_M \coloneqq \sum_{\textbf{k} \in \mathbb{N}^n : \sum_{i=1}^n k_i \le M} \ketbra{\textbf{k}}\,,
\ee
\emph{with \(\ket{\textbf{k}} = \ket{k_1} \otimes \ket{k_2} \otimes \ldots \otimes \ket{k_n}\) denoting an \(n\)-mode Fock state. Then, the following bound holds:}
\be
    P_{> M} \le 2^n e^{1/2} \, e^{-M / (4\Tr[\rho \hat{N}] + 2)}\,,
\ee
\emph{where \(\hat{N}\) is the photon number operator.}
\item Second, as an application of this bound, Ref.~\cite{mele2025achievableratesnonasymptoticbosonic} establishes an upper bound on the trace distance between the Gaussian state \(\rho\) and its finite-dimensional truncation \(\rho_M\), defined as
\be
    \rho_M \coloneqq \frac{\Pi_M \rho \Pi_M}{\Tr[\Pi_M \rho \Pi_M]}\,,
\ee
where \(\Pi_M\) is the projector given in Eq.~\eqref{proj_fockmain}. The trace distance between \(\rho\) and \(\rho_M\) is shown to vanish exponentially in \(M\) as follows:
\be
    \frac{1}{2}\left\|\rho - \rho_M\right\|_1 \le 2^{n/2} e^{1/4} \, e^{-M / \left(8\Tr[\rho \hat{N}] + 4\right)}\,.
\ee

\end{itemize}
These results play a key role in constructing the aforementioned algorithm that computes the trace distance between two Gaussian states up to a fixed precision~\cite{mele2025achievableratesnonasymptoticbosonic}.

\subsection{Trace distance perturbation bounds on first moment, covariance matrices, and symplectic eigenvalues}

The trace distance bounds of \cref{proof_lower_bound} can be equivalently seen as continuity bounds on the first moments and the covariance matrices. In other words, we are answering the following fundamental question: If two states are \(\varepsilon\)-close in trace distance, how close are their first moments and covariance matrices? More specifically, given two states $\rho$ and $\sigma$, how to upper bound the norm distance between their first moments $\|\bm{m}(\rho)-\bm{m}(\sigma)\|$ and their covariance matrices $\|V(\rho)-V(\sigma)\|$ in terms of the trace distance $\frac12\|\rho-\sigma\|_1$? The forthcoming \cref{thm_pert_m} and the forthcoming \cref{thm_pert_V}  provide such bounds for the first moments and covariance matrices, respectively. The proof of these theorems are direct consequences of the above \cref{proof_lower_bound}.

\begin{theo}[(Trace distance perturbation bounds on first moments)]\label{thm_pert_m}
    Let $\rho$ and $\sigma$ be two quantum states with mean energy upper bounded by a finite constant $E$, that is $\Tr[\hat{E}\rho]\le E$ and $\Tr[\hat{E}\sigma]\le E$.  Then, the Euclidean distance between their first moments can be upper bounded in terms of the trace distance as
    \be
        \|\bm{m}(\rho)-\bm{m}(\sigma)\|\le 4\sqrt{E\|\rho-\sigma\|_1}\,,
    \ee
    where norm $\|\cdot\|$ is defined as $\|\bm{v}\|\coloneqq \sqrt{\bm{v}^\intercal\bm{v}}$.
\end{theo}

\begin{theo}[(Trace distance perturbation bounds on covariance matrices)]\label{thm_pert_V}
    Let $\rho$ and $\sigma$ be two quantum states with second moment of the energy upper bounded by a finite constant $E$, that is $\Tr[\hat{E}^2\rho]\le E^2$ and $\Tr[\hat{E}^2\sigma]\le E^2$.  Then, the operator norm distance between their covariance matrices can be upper bounded in terms of the trace distance as
    \be\label{tight_bound}
        \|V(\rho)-V(\sigma)\|_\infty\le  40\,E \sqrt{\|\rho-\sigma\|_1} \,.
    \ee
\end{theo}

The above theorem implies that if two energy-constrained states $\rho$ and $\sigma$ are $\varepsilon$-close in trace distance, then their covariance matrices are at most $O(\sqrt{\varepsilon})$ far apart. Moreover, the bound in Eq.~\eqref{tight_bound} is tight in the scaling with the second moment of the energy $E$ and the trace norm $\|\rho-\sigma\|_1$. To show this, let us consider the following example. 

\begin{ex}[(Tightness in the scaling of the parameters in \cref{thm_pert_V})]
    For any $n\in\mathbb{N}$, let us consider the state
    \begin{equation}
         \rho_n\coloneqq\left(1-\frac{1}{n}\right)|0\rangle\!\langle0|+\frac{1}{n}|n^2 \rangle\!\langle n^2|\,,
    \end{equation}
    which satisfies
    \begin{align}
        \|V(\rho_n)-V(|0\rangle\!\langle0|)\|_\infty &=2n\,,\\
        E_n\coloneqq\sqrt{\Tr[\hat{E}^2\rho_n]}&=\frac1n\left(n^2+\frac{1}{2}\right)^2\,,\\
        \|\rho_n-|0\rangle\!\langle0|\|_1&=\frac{2}{n}\,.
    \end{align}
    The above identities imply that
    \begin{equation}
        \frac{\|V(\rho_n)-V(|0\rangle\!\langle0|)\|_\infty}{E_n\sqrt{\|\rho_n-|0\rangle\!\langle0|\|_1}}\propto 1+O(n^{-1})\,,
    \end{equation}
    establishing the tightness of the bound of \cref{thm_pert_V} in terms of the scaling of all the parameters involved.
\end{ex}

Before proving the above \cref{thm_pert_m} and \cref{thm_pert_V}, let us see the following interesting application, which turned out to be crucial in the analysis of the symplectic rank in \cref{sec_non_dec} (specifically, in the proof of \cref{slight_pert}).
\begin{theo}[(Trace distance perturbation bound on symplectic eigenvalues)]\label{pert_bound_symp}
    Given a quantum state $\rho$, let us denote as $\nu_i^{\downarrow}(\rho)$ the $i$th largest symplectic eigenvalue of its covariance matrix $V(\rho)$. Let $\rho$ and $\sigma$ be two quantum states with second moment of the energy upper bounded by a finite constant $E$, that is $\Tr[\hat{E}^2\rho]\le E^2$ and $\Tr[\hat{E}^2\sigma]\le E^2$. Then, it holds that
    \begin{equation}
        \max_i\left|\nu_i^{\downarrow}(\rho)-\nu_i^{\downarrow}(\sigma)\right|\le 640 E^2\sqrt{\|\rho-\sigma\|_1}\,.
    \end{equation}
\end{theo}

Before proving such a theorem, we need to establish some preliminary results. Let us start by stating the following known perturbation bound on the symplectic diagonalization.

\begin{lem}[(Perturbation on symplectic diagonalization)~\cite{idel2017}]\label{lem:wolf}
    Let $V_1,V_2\in\mathbb{R}^{2n\times 2n}$ be two strictly positive matrices with symplectic diagonalizations $V_1=S_{1}D_1S_{1}^{\intercal}$ and $V_2=S_{2}D_2S_{2}^{\intercal}$, where the elements on the diagonal of $D_1$ and $D_2$ are arranged in descending order. Then, the operator norm distance between $D_1$ and $D_2$ can be upper bounded in terms of the operator norm distance between $V_1$ and $V_2$ as follows:
\begin{equation}
    \|D_1-D_2\|_\infty\le \sqrt{K\!(V_1)K\!(V_2)}\|V_1-V_2\|_\infty\,,
\end{equation}
where $K\!(V)$ denotes the condition number of the matrix $V$, defined as $K\!(V)\coloneqq  \|V\|_\infty \|V^{-1}\|_\infty$.
\end{lem}

Moreover, the condition number of a covariance matrix can be upper bounded as follows.
\begin{lem}[(Bounding the condition number of the covariance matrix in terms of the energy)]\label{lemma_condition}
    Let $\rho$ be a quantum state with covariance matrix $V(\rho)$. The condition number 
    \be
     K(V(\rho))\coloneqq \|V(\rho)\|_\infty \|V^{-1}(\rho)\|_\infty
    \ee
    can be upper bounded in terms of the second moment of the energy as 
    \be
        K(V(\rho))\le 16 \Tr[\hat{E}^2\rho]\,.
    \ee
\end{lem}

\begin{proof}
    It holds that
    \be
        K(V(\rho))&\leqt{(i)} \|V(\rho)\|_\infty^2\\
        &\leqt{(ii)}16 \left(\Tr[\hat{E}\rho]\right)^2\\
        &\leqt{(iii)} 16\Tr[\hat{E}^2\rho]\,,
    \ee
    where, in (i), we exploited ~\cite[Lemma~S77]{mele2024learning}; in (ii) we leveraged Eq.~\eqref{bounded_cov_matrix}; and, in (iii), we used the concavity of the square root function. 
\end{proof}

We are now ready to prove \cref{pert_bound_symp}.

\begin{proof}[Proof of \cref{pert_bound_symp}]
It holds that
\be
    \max_i\left|\nu_i^{\downarrow}(\rho)-\nu_i^{\downarrow}(\sigma)\right|&\leqt{(i)}\sqrt{K(V(\rho))K(V(\sigma))}\|V(\rho)-V(\sigma)\|_\infty\\
    &\leqt{(ii)}16 E^2\|V(\rho)-V(\sigma)\|_\infty\\
    &\leqt{(iii)} 640 E^2 \sqrt{\|\rho-\sigma\|_1}
\ee
    where, in (i), we applied  \cref{lem:wolf}; in (ii), we exploited  \cref{lemma_condition}, and, in (iii), we leveraged \cref{thm_pert_V}.
\end{proof}

\subsection{Proofs of the trace distance lower bounds}\label{sec_proofs}
Let us now proceed with the proof of the above stated \cref{thm_pert_m} and \cref{thm_pert_V}. Let us start by establishing the following lower bound on the trace distance between two quantum states in terms of the difference of their expectation values of a quadrature operator $\hat{x}$ (i.e.~a position operator or a momentum operator) and in terms of the difference of their expectation values of $\hat{x}^2$.

\begin{lem}[(Lower bound on the trace distance in terms of the mean value of a quadrature operator)]\label{lemma_x_x2}
    Let $\rho$ and $\sigma$ be quantum states. Let $\hat{x}$ be a quadrature operator associated with a given mode. Then, the trace norm between $\rho$ and $\sigma$ can be lower bounded in terms of the expectation value of $\hat{x}$ as
    \begin{align}\label{eq_lower_boundx}
    \|\rho-\sigma\|_1\ge \frac{|\Tr\!\left[\hat{x} (\rho-\sigma)\right]|^2}{8 \max\!\left(\Tr\!\left[\hat{x}^2\rho\right],\Tr\!\left[\hat{x}^2\sigma\right]\right)}\,,
    \end{align}
    and in terms of the second moment as
    \begin{align}\label{eq_lower_boundx2}
            \|\rho-\sigma\|_1&\ge\frac{|\Tr[\hat{x}^2(\rho-\sigma)]|^2}{4\max(
  \Tr[\hat{x}^4\rho],\Tr[\hat{x}^4\sigma])}\,.
    \end{align}
\end{lem}

\begin{proof}
Let us start with the proof of Eq.~\eqref{eq_lower_boundx}. The trace norm can be expressed as follows~\cite{Sumeet_book}:
\begin{equation}
    \|\rho-\sigma\|_1=\max_{\|\Theta\|_\infty\le 1}\left|\Tr[\Theta(\rho-\sigma)]\right|\,,
\end{equation}
where the maximum is over all operators with operator norm upper bounded by one. Let us consider the quadrature operator $\hat{x}$ in spectral decomposition:
\begin{equation}
    \hat{x}=\int_{x\in\mathbb{R}}x \ket{x}\!\!\bra{x}dx\,,
\end{equation}
where $\ket{x}$ denotes the (generalized) eigenvectors of $\hat{x}$~\cite{BUCCO}. By defining for all $\bar{x}>0$ the operator
\begin{equation}
    \Theta_{\bar{x}}\coloneqq \frac{1}{\bar{x}}\int_{|x|\le\bar{x}}\mathrm{d }x\,x \ket{x}\!\!\bra{x}\,,
\end{equation}
we have that $\|\Theta_{\bar{x}}\|_\infty=1$ and thus 
\be\label{eq_used}
    \|\rho-\sigma\|_1&\ge \left|\Tr[ \Theta_{\bar{x}}(\rho-\sigma) ]\right|\\
    &\ge \Tr[ \Theta_{\bar{x}}(\rho-\sigma) ]  \\
    &= \frac{1}{\bar{x}}\Tr\!\left[\int_{|x|\le\bar{x}}\mathrm{d }x\,x\ket{x}\!\!\bra{x}(\rho-\sigma)\right]\\
    &= \frac{1}{\bar{x}}\Tr\!\left[\hat{x} (\rho-\sigma)\right]-\frac{1}{\bar{x}}\Tr\!\left[\int_{|x|\ge\bar{x}}\mathrm{d }x\,x\ket{x}\!\!\bra{x}(\rho-\sigma)\right]\\
    &\ge \frac{1}{\bar{x}}\Tr\!\left[\hat{x} (\rho-\sigma)\right]-\frac{2}{\bar{x}}\max\left(\left|\Tr\!\left[\int_{|x|\ge\bar{x}}\mathrm{d }x\,x\ket{x}\!\!\bra{x}\rho\right]\right|, \left|\Tr\!\left[\int_{|x|\ge\bar{x}}\mathrm{d }x\,x\ket{x}\!\!\bra{x}\sigma\right]\right|\right)\,.
\ee
Moreover, by exploiting that
\be
    -\int_{|x|\ge\bar{x}}\mathrm{d }x\,|x|\ket{x}\!\!\bra{x}\le \int_{|x|\ge\bar{x}}\mathrm{d }x\,x\ket{x}\!\!\bra{x}\le \int_{|x|\ge\bar{x}}\mathrm{d }x\,|x|\ket{x}\!\!\bra{x}\,,
\ee
it follows that
\be
    \left|\Tr\!\left[\int_{|x|\ge\bar{x}}\mathrm{d }x\,x\ket{x}\!\!\bra{x}\rho\right]\right|&\le \Tr\!\left[\int_{|x|\ge\bar{x}}\mathrm{d }x\,|x|\ket{x}\!\!\bra{x}\rho\right]\\
    & \le \frac{1}{\bar{x}}\Tr\!\left[\int_{|x|\ge\bar{x}}\mathrm{d }x\,x^2\ket{x}\!\!\bra{x}\rho\right]\\  \\
    & \le \frac{1}{\bar{x}}\Tr\!\left[\int_{x\in\mathbb{R}}\mathrm{d }x\,x^2\ket{x}\!\!\bra{x}\rho\right]\\  \\
    &=\frac{1}{\bar{x}} \Tr\!\left[\hat{x}^2\rho\right]\,.
\ee
Consequently, by substituting in Eq.~\eqref{eq_used}, we have that
\be
    \|\rho-\sigma\|_1\ge \frac{1}{\bar{x}}\Tr\!\left[\hat{x} (\rho-\sigma)\right]-\frac{2}{\bar{x}^2}\max\!\left(\Tr\!\left[\hat{x}^2\rho\right],\Tr\!\left[\hat{x}^2\sigma\right]\right)\,.
\ee
By exchanging the role of $\rho$ and $\sigma$, we also obtain that
\be
    \|\rho-\sigma\|_1\ge -\frac{1}{\bar{x}}\Tr\!\left[\hat{x} (\rho-\sigma)\right]-\frac{2}{\bar{x}^2}\max\!\left(\Tr\!\left[\hat{x}^2\rho\right],\Tr\!\left[\hat{x}^2\sigma\right]\right)\,,
\ee
and thus also that
\begin{align*}
    \|\rho-\sigma\|_1\ge \frac{1}{\bar{x}}|\Tr\!\left[\hat{x} (\rho-\sigma)\right]|-\frac{2}{\bar{x}^2}\max\!\left(\Tr\!\left[\hat{x}^2\rho\right],\Tr\!\left[\hat{x}^2\sigma\right]\right)\,.
\end{align*}
By optimizing over $\bar{x}$, we get:
\begin{align}
    \|\rho-\sigma\|_1\ge \frac{|\Tr\!\left[\hat{x} (\rho-\sigma)\right]|^2}{8 \max\!\left(\Tr\!\left[\hat{x}^2\rho\right],\Tr\!\left[\hat{x}^2\sigma\right]\right)}\,.
\end{align}
This concludes the proof of Eq.~\eqref{eq_lower_boundx}.

Now, let us proceed with the proof of Eq.~\eqref{eq_lower_boundx2}. Note that $\hat{x}^2$ can be written in spectral decomposition as
\begin{equation}
    \hat{x}^2=\int_{-\infty}^\infty x^2 \ket{x}\!\!\bra{x}dx\,.
\end{equation}
By defining for all $\bar{x}>0$ the operator
\begin{equation}
    \Theta^2_{\bar{x}}\coloneqq \frac{1}{\bar{x}^2}\int_{|x|\le\bar{x}}x^2 \ket{x}\!\!\bra{x}dx\,,
\end{equation}
we have that $\|\Theta^2_{\bar{x}}\|_\infty=1$ and thus
\begin{align}
    \|\rho-\sigma\|_1&\ge \Tr[\Theta^2_{\bar{x}}(\rho-\sigma)]\nonumber\\
    &=\frac{1}{\bar{x}^2}\left(\Tr[\hat{x}^2(\rho-\sigma)] - 
  \Tr\!\left[\int_{|x|\ge\bar{x}}x^2 \ket{x}\!\!\bra{x}dx\,(\rho-\sigma)\right]
   \right)\nonumber\\
    &\ge\frac{1}{\bar{x}^2}\left(\Tr[\hat{x}^2(\rho-\sigma)] - 
  \Tr\!\left[\int_{|x|\ge\bar{x}}x^2 \ket{x}\!\!\bra{x}dx\,\rho\right]
   \right)\nonumber\\
    &\ge\frac{1}{\bar{x}^2}\left(\Tr[\hat{x}^2(\rho-\sigma)] - \frac{1}{\bar{x}^2}
  \Tr\!\left[\int_{|x|\ge\bar{x}}x^4 \ket{x}\!\!\bra{x}dx\,\rho\right]
   \right)\nonumber\\
    &\ge\frac{1}{\bar{x}^2}\Tr[\hat{x}^2(\rho-\sigma)] - \frac{1}{\bar{x}^4}
  \Tr[\hat{x}^4\rho]\nonumber\\
    &\ge\frac{1}{\bar{x}^2}\Tr[\hat{x}^2(\rho-\sigma)] - \frac{1}{\bar{x}^4}\max(
  \Tr[\hat{x}^4\rho],\Tr[\hat{x}^4\sigma])\,.
\end{align}
Exchanging $\rho$ and $\sigma$, we obtain the lower bound
\begin{align}
    \|\rho-\sigma\|_1\ge\frac{1}{\bar{x}^2}|\Tr[\hat{x}^2(\rho-\sigma)]| - \frac{1}{\bar{x}^4}\max(
  \Tr[\hat{x}^4\rho],\Tr[\hat{x}^4\sigma])\,.
\end{align}
By optimizing over $\bar{x}$, we conclude that
\be
    \|\rho-\sigma\|_1\ge\frac{|\Tr[\hat{x}^2(\rho-\sigma)]|^2}{4\max(
  \Tr[\hat{x}^4\rho],\Tr[\hat{x}^4\sigma])}\,.
\ee
\end{proof}

The bound of the above lemma exhibits a dependence upon the expectation value of the observable $\hat{x}^4$. In the next lemma, we show how to upper bound the latter in terms of the energy.

\begin{lem}[(Bounding $\hat{x}^4$ in terms of $\hat{E}^2$)]\label{lemmax4}
    Let $\hat{x}$ and $\hat{E}$ be the position operator and the energy operator, respectively, of a single-mode system. Then, the following operator inequality holds:
    \begin{equation}
        \hat{x}^4\preceq \left(3+6\sqrt{2}+2\sqrt{6}\right)\hat{E}^2\,.
    \end{equation}
\end{lem}

\begin{proof}
    The goal is to find $\alpha$ such that
$\hat{x}^4\preceq \alpha \hat{E}^2$. To this end, the crucial idea is to exploit the \emph{diagonally dominant criterion}: if a hermitian matrix $A=(a_{ij})_{i,j}$ is such that 
\be
a_{nn} \ge \sum_{m:\,m\neq n} |a_{mn}|\qquad\forall \,n\,,
\ee
then necessarily $A$ is positive semidefinite~\cite[Chapter 6]{HJ1}. Hence, the operator $\alpha \hat{E}^2-\hat{x}^4$ is positive semidefinite if for all $n\in\mathbb{N}$ it holds that
\begin{equation}
     \sum_{j\ne n} |\bra{j}(\alpha\hat{E}^2-\hat{x}^4)\ket{n}|\le\bra{n}(\alpha\hat{E}^2-\hat{x}^4)\ket{n} \,.
\end{equation}
Since $\hat{E}^2=\sum_{n=0}^\infty\left(n+\frac12\right)^2\ketbra{n}$, it suffices to choose $\alpha$ such that
\begin{equation}
    \alpha \ge \sup_{n\in\mathbb{N}} f(n)\,,
\end{equation}
where $f:\mathbb{N}\longmapsto \mathbb{R}$ is defined as
\begin{equation}
    f(n)\coloneqq \frac{1}{\left( n+\frac12 \right)^2}\sum_{j=0}^\infty|\bra{j}\hat{x}^4\ket{n}|\,.
\end{equation}
By exploiting that~\cite{BUCCO} 
\begin{align}
\hat{x}&=\frac{a+a^\dagger}{\sqrt{2}}\,,\\
a\ket{n}&=\sqrt{n}\ket{n-1}\,,\\
a^\dagger\ket{n}&=\sqrt{n+1}\ket{n+1}\,,
\end{align}
one obtains that 
\begin{align}
    f(n)=\frac{1}{4\left( n+\frac12 \right)^2}&\Bigg(\sqrt{n(n-1)(n-2)(n-3)}+(4n-2)\sqrt{n(n-1)} +3(n+1)^2\nonumber\\
    &\quad + (4n+6)\sqrt{(n+1)(n+2)}+3n^2+\sqrt{(n+1)(n+2)(n+3)(n+4)}\Bigg) 
\end{align}
for all $n\ge 3$. Moreover, it holds that
\begin{align}
    f(2)&=\frac{\left(6\sqrt{2}+39+28\sqrt{3}+6\sqrt{10}\right)}{4\left(2+\frac{1}{2}\right)^{2}}\,,\\
    f(1)&=\frac{\left(21+14\sqrt{6}+2\sqrt{30}\right)}{4\left(1+\frac{1}{2}\right)^{2}}\,,\\
    f(0)&=3+6\sqrt{2}+2\sqrt{6}\,.
\end{align}
Consequently, one can show that
\begin{equation}
    \sup_{n\in\mathbb{N}} f(n)=f(0)=3+6\sqrt{2}+2\sqrt{6}\,.
\end{equation}
Hence, it suffices to choose $\alpha=3+6\sqrt{2}+2\sqrt{6}$. This concludes the proof.
\end{proof}
We are now ready to prove Eq.~\eqref{ineq_first} in \cref{proof_lower_bound}.
\begin{proof}[Proof of Eq.~\eqref{ineq_first}]
Thanks to  \cref{lemma_x_x2}, we have that
\begin{align}
    \|\rho-\sigma\|_1\ge \frac{|\Tr\!\left[\hat{x}_1 (\rho-\sigma)\right]|^2}{8 \max\!\left(\Tr\!\left[\hat{x}_1^2\rho\right],\Tr\!\left[\hat{x}_1^2\sigma\right]\right)}\,,
\end{align}
where $\hat{x}_1$ denotes the position operator of the first mode.  Since $\Tr\!\left[\hat{x}_1^2\rho\right]\le 2\Tr\!\left[\hat{E}\rho\right]$, it follows that
\begin{align}\label{eq_low_enen}
    \|\rho-\sigma\|_1\ge \frac{|\Tr\!\left[\hat{x}_1 (\rho-\sigma)\right]|^2}{16 {\max\!\left(\Tr\!\left[\hat{E}\rho\right],\Tr\!\left[\hat{E}\sigma\right]\right)}}\,.
\end{align}
By defining 
\begin{equation}
    \bm{w}\coloneqq \frac{\bm{m}(\sigma)-\bm{m}(\rho)}{\|\bm{m}(\sigma)-\bm{m}(\rho)\|}\,,
\end{equation}
it holds that
\begin{equation}
    \|\bm{m}(\sigma)-\bm{m}(\rho)\|=\bm{w}^\intercal\left[\bm{m}(\sigma)-\bm{m}(\rho)\right]=\Tr[(\bm{w}^\intercal\hat{\bm{R}})(\sigma-\rho)]\,.
\end{equation}
Furthermore, \cite[Lemma 13]{Lami_2020} establishes that there exists a passive Gaussian unitary operation $G$ such that $\bm{w}^\intercal \hat{\bm{R}}= G^\dagger\hat{x}_1 G$. Hence, we have that
\begin{equation}\label{eq_mg}
    \|\bm{m}(\sigma)-\bm{m}(\rho)\|=\Tr[\hat{x}_1 G(\sigma-\rho)G^\dagger]\,.
\end{equation}
Consequently, we get
\begin{align}
    \|\rho-\sigma\|_1&=\|G(\rho-\sigma)G^\dagger\|_1\nonumber\\
    &\ge \frac{\left|\Tr\!\left[\hat{x}_1 G(\rho-\sigma)G^\dagger\right]\right|^2}{16 {\max\!\left(\Tr\!\left[\hat{E} G \rho G^\dagger\right],\Tr\!\left[\hat{E}G\sigma G^\dagger\right]\right)}}\nonumber\\
    &=\frac{\|\bm{m}(\sigma)-\bm{m}(\rho)\|^2}{16 {\max\!\left(\Tr\!\left[\hat{E}  \rho \right],\Tr\!\left[\hat{E}\sigma\right]\right)}}\,,
\end{align}
where the first line exploits the invariance of the trace distance under unitary operations, the second line follows from Eq.~\eqref{eq_low_enen}, and the last line exploits Eq.~\eqref{eq_mg} together with the fact that passive Gaussian unitary operations are energy-preserving. 
\end{proof}

Before proving the trace distance perturbation bounds on the covariance matrices, we need to establish the following preliminary result:

\begin{lem}[(Trace distance perturbation bounds on second moments)]\label{lemma_sec_q}
    Let $\rho$ and $\sigma$ be quantum states. Let $T(\rho)$ be the second moment matrix defined as
    \begin{equation}
         T(\rho)\coloneqq \Tr[\left\{\hat{\bm{R}},\hat{\bm{R}}^\intercal\right\}\rho]\,.
    \end{equation}
    Then, the operator norm distance between the second moments $T(\rho)$ and $T(\sigma)$ can be upper bounded in terms of the trace distance as
    \begin{equation}
        \|T(\rho)-T(\sigma)\|_\infty\le 17\, \sqrt{\max\!\left({\Tr[\hat{E}^2  \rho ]},{\Tr[\hat{E}^2\sigma]}\right)\|\rho-\sigma\|_1}\,.
    \end{equation}
\end{lem}

\begin{proof}
Lemma~\eqref{lemma_x_x2} establishes that 
\begin{align}\label{eq_rho_sigma_trace}
    \|\rho-\sigma\|_1\ge\frac{|\Tr[\hat{x}_1^2(\rho-\sigma)]|^2}{4\max(
  \Tr[\hat{x}_1^4\rho],\Tr[\hat{x}_1^4\sigma])}\,.
\end{align}
Moreover, thanks to  \cref{lemmax4}, we have that
\begin{equation}
    \hat{x}_1^4\preceq c\hat{E}_1^2\,,
\end{equation}
where $\hat{E}_1$ denotes the energy operator of the first mode and $c$ is defined by
\begin{equation}
    c\coloneqq  3+6\sqrt{2}+2\sqrt{6}\,.
\end{equation}
Moreover, note that $\hat{E}_1^2\preceq \hat{E}^2$, where $\hat{E}$ denotes the total energy operator. Consequently, by substituting in Eq.~\eqref{eq_rho_sigma_trace}, we obtain
\begin{align}\label{eq_rho_sigma_trace2}
    \|\rho-\sigma\|_1\ge\frac{|\Tr[\hat{x}_1^2(\rho-\sigma)]|^2}{4c\max(
  \Tr[\hat{E}^2\rho],\Tr[\hat{E}^2\sigma])}\,.
\end{align}
Let $\bm{w}$ with $\bm{w}^\intercal \bm{w}=1$ such that
\begin{equation}
   \left\| T(\rho)-T(\sigma)\right\|_\infty= |\bm{w}^\intercal \left[T(\rho)-T(\sigma)\right]\bm{w}|=2\left|\Tr[(\rho-\sigma)(\bm{w}^\intercal \hat{\bm{R}})^2]\right|\,.
\end{equation}
By exploiting \cite[Lemma 13]{Lami_2020}, there exists a passive Gaussian unitary operation $G$ such that $\bm{w}^\intercal \hat{\bm{R}}= G^\dagger\hat{x}_1 G$. Consequently, we have that
\begin{equation}\label{eq_key}
    \left\| T(\rho)-T(\sigma)\right\|_\infty=2\left|\Tr\!\left[ G(\rho-\sigma)G^\dagger \,\hat{x}_1^2  \right]\right|\,.
\end{equation}
Hence,
\begin{align}\label{eq_to_use}
    \|\rho-\sigma\|_1&=\|G(\rho-\sigma)G^\dagger \|_1\nonumber\\
    &\ge\frac{|\Tr[\hat{x}_1^2G(\rho-\sigma)G^\dagger]|^2}{4c\max(
  \Tr[\hat{E}^2G\rho G^\dagger],\Tr[\hat{E}^2G\sigma G^\dagger])}\nonumber\\
    & =\frac{\|T(\rho)-T(\sigma)\|_\infty^2}{16c\max(
  \Tr[\hat{E}^2\rho ],\Tr[\hat{E}^2\sigma])}\,,
\end{align}
where the first inequality is given by Eq.~\eqref{eq_rho_sigma_trace2} and in the last inequality we employed Eq.~\eqref{eq_key} together with the fact that $G$ is passive. Hence, by rearranging and by exploiting that $\sqrt{16c}\le 17$, we conclude the proof.
\end{proof}
We are now ready to prove Eq.~\eqref{ineq_cov} in the above-stated \cref{proof_lower_bound}, which concludes the proof of the latter theorem.
\begin{proof}[Proof of Eq.~\eqref{ineq_cov}]
    Note that Eq.~\eqref{ineq_first} implies that
\begin{equation}\label{eq_c}
     \|\bm{m}(\sigma)-\bm{m}(\rho)\|\le 4 \sqrt{\max\!\left(\Tr[\hat{E}\rho],\Tr[\hat{E}\sigma]\right)}\sqrt{\|\rho-\sigma\|_1}\,.
\end{equation}
By the triangle inequality, for any $A$ and $B$ it holds that:
\begin{equation}
    \|A+B\|_\infty^2\le 2\left( \|A\|_\infty^2+\|B\|_\infty^2  \right)\,.
\end{equation}
Consequently, for any $A$ and $B$ it holds that
\begin{equation}
    \|A\|_\infty^2\ge \frac{\|A+B\|_\infty^2}{2}-\|B\|_\infty^2\,,
\end{equation}
and hence
\begin{equation}
    \|A-B\|_\infty^2\ge \frac{\|A\|_\infty^2}{2}-\|B\|_\infty^2\,.
\end{equation}
Consequently,
\begin{align}\label{Eq_cov}
    \|T(\rho)-T(\sigma)\|_\infty^2&=\|V(\rho)+2\bm{m}(\rho)\bm{m}(\rho)^\intercal-V(\sigma)- 2\bm{m}(\sigma)\bm{m}(\sigma)^\intercal\|_\infty^2\\
    &\ge \frac{1}{2}\|V(\rho)-V(\sigma)\|_{\infty}^2-4\|\bm{m}(\rho)\bm{m}(\rho)^\intercal-\bm{m}(\sigma)\bm{m}(\sigma)^\intercal\|_{\infty}^2\,.\label{Eq_cov_00}
\end{align}
By exploiting the triangle inequality and the fact that $\|\bm{v}\bm{w}^\intercal\|_\infty=\|\bm{v}\|\|\bm{w}\|$, we obtain
\begin{equation}\label{eq_a}
    \|\bm{m}(\rho)\bm{m}(\rho)^\intercal-\bm{m}(\sigma)\bm{m}(\sigma)^\intercal\|_{\infty}\le \|\bm{m}(\sigma)-\bm{m}(\rho)\|\left(\|\bm{m}(\rho)\|+\|\bm{m}(\sigma)\|\right).
\end{equation}
By employing Cauchy--Schwarz inequality and the concavity of the square root function, the norm of the first moment of any state $\rho$ can be upper bounded as follows:
\begin{equation}\label{eq_b}
    \|\bm{m}(\rho)\|=\sqrt{\sum_{i=1}^{2n}|\Tr[R_i\rho]|^2}\le \sqrt{\sum_{i=1}^{2n}\Tr[R_i^2\rho]}=\sqrt{2\Tr[\hat{E}\rho]} \,.
\end{equation}
By exploiting Eq.~\eqref{eq_a}, Eq.~\eqref{eq_b}, and Eq.~\eqref{eq_c}, one obtains that
\begin{equation}
    \|\bm{m}(\rho)\bm{m}(\rho)^\intercal-\bm{m}(\sigma)\bm{m}(\sigma)^\intercal\|_{\infty}\le  8\sqrt{2}\max\left(\Tr[\hat{E}\rho],\Tr[\hat{E}\sigma]\right)\sqrt{\|\rho-\sigma\|_1}.
\end{equation}
In particular, it holds that
\begin{equation}
    \|\bm{m}(\rho)\bm{m}(\rho)^\intercal-\bm{m}(\sigma)\bm{m}(\sigma)^\intercal\|_{\infty}^2\le  2^7\left[\max\left(\Tr[\hat{E}\rho],\Tr[\hat{E}\sigma]\right)\right]^2\|\rho-\sigma\|_1\,.
\end{equation}
Hence, by substituting in Eq.~\eqref{Eq_cov}, we have that
\begin{equation}\label{Eq_cov2}
    \|T(\rho)-T(\sigma)\|_\infty^2\ge \frac{1}{2}\|V(\rho)-V(\sigma)\|_{\infty}^2-2^9\left[\max\left(\Tr[\hat{E}\rho],\Tr[\hat{E}\sigma]\right)\right]^2\|\rho-\sigma\|_1\,.
\end{equation}
By substituting in Eq.~\eqref{eq_to_use} and rearranging the terms, we conclude that
\begin{align}\label{eq_to_use_2}
    \|\rho-\sigma\|_1&\ge \frac{\|V(\rho)-V(\sigma)\|_{\infty}^2}{2^{10}\max\left(\Tr[\hat{E}\rho],\Tr[\hat{E}\sigma]\right)^2+ 32 \,c\max\left(\Tr[\hat{E}^2\rho],\Tr[\hat{E}^2\sigma]\right)}\nonumber\\
    &\ge \frac{\|V(\rho)-V(\sigma)\|_{\infty}^2}{(2^{10}+32c)\max\left(\Tr[\hat{E}^2\rho],\Tr[\hat{E}^2\sigma]\right)}\,,
\end{align}
where in the last inequality we used that $ \Tr[\hat{E}\rho]^2\le \Tr[\hat{E}^2\rho]$. By rearranging and using ${2^{10}+32c}\le{1549}$, we conclude the proof.
\end{proof}



\section{Perturbation bounds for classical probability distributions}\label{Sec_classical_bound}

In this section, we apply the method introduced in \cref{sec_pert_bounds} to the setting of classical probability theory; specifically, we establish lower bounds on the total variation distance between two classical probability distributions in terms of the difference of their mean values and their covariance matrices. The inequalities derived in this section may be of independent interest for the classical probability literature. The non-interested reader can skip the present section, as it does not have any application to the analysis of the symplectic rank.

Let $p:\mathbb{R}^n\mapsto\mathbb{R}$ and  $q:\mathbb{R}^n\mapsto\mathbb{R}$ be probability distributions over $\mathbb{R}^n$. The $L_1$-distance between $p$ and $q$ is defined as
\begin{equation}
    \|p-q\|_1\coloneqq \int_{\mathbb{R}^n}\mathrm{d}^n\bm{x}\,|p(\bm{x})-q(\bm{x})|\,.
\end{equation}
The \emph{total variation distance} between $p$ and $q$ is given by $\frac12\|p-q\|_1$ and is a very important notion of distance between classical probability distributions~\cite{NC}. The question is: if the moments of $p$ and $q$ differs by $\varepsilon$, what can we say about the total variation distance between $p$ and $q$? Here, we provide a quantitative answer to this question, by finding lower bounds on the the total variation distance between $p$ and $q$ in terms of the norm distance between their moments. The approach is completely analogous to the one we used in the previous section in order to find lower bounds on the trace distance between two quantum states in terms of the norm distance of the moments.

Let us start with the following.

\begin{lem}[(Lower bounds on the total variation distance in terms of the difference of the mean values)]\label{classical_lower_mean}
    Let $p:\mathbb{R}^n\mapsto\mathbb{R}$ and  $q:\mathbb{R}^n\mapsto\mathbb{R}$ be probability distributions over $\mathbb{R}^n$. Then, it holds that
    \begin{equation}
        \|p-q\|_1\ge \frac{\left\|\mathbb{E}_p\!\left[\bm{x}\right]-\mathbb{E}_q\!\left[\bm{x}\right]\right\|^2}{8\max\left(\mathbb{E}_p\!\left[\|\bm{x}\|^2\right],\mathbb{E}_q\!\left[\|\bm{x}\|^2\right]\right)}\,,
    \end{equation}
    where $\|\bm{x}\|\coloneqq \sqrt{\sum_{i=1}^nx_i^2}$ and where $\mathbb{E}_p$ denotes the expectation value with respect to the probability distribution $p$.
\end{lem}

\begin{proof}
    The proof is similar to the one of \cref{thm_pert_m}. Note that for any $\bar{x}\ge0$ we have that
    \begin{align}
        \|p-q\|_1&=\int_{\mathbb{R}^n}\mathrm{d}^n\bm{x}\,|p(\bm{x})-q(\bm{x})|\nonumber\\
        &\ge \frac{1}{\bar{x}}\int_{|x_1|\le \bar{x}}\mathrm{d}^n\bm{x}\,x_1[p(\bm{x})-q(\bm{x})]\nonumber\\
        &=\frac{1}{\bar{x}}\int_{\mathbb{R}^n}\mathrm{d}^n\bm{x}\,x_1[p(\bm{x})-q(\bm{x})]-\frac{1}{\bar{x}}\int_{|x_1|\ge \bar{x}}\mathrm{d}^n\bm{x}\,x_1[p(\bm{x})-q(\bm{x})]\nonumber\\
        &\ge \frac{1}{\bar{x}}\left(\mathbb{E}_p[x_1]-\mathbb{E}_q[x_1]\right)-\frac{1}{\bar{x}^2}\int_{|x_1|\ge \bar{x}}\mathrm{d}^n\bm{x}\,x_1^2[p(\bm{x})+q(\bm{x})]\nonumber\\
        &\ge \frac{1}{\bar{x}}\left(\mathbb{E}_p[x_1]-\mathbb{E}_q[x_1]\right)-\frac{2}{\bar{x}^2}\max\left(\mathbb{E}_p[x_1^2],\mathbb{E}_q[x_1^2]\right)\,.
    \end{align}
    By exchanging the role of $p$ and $q$, we obtain that
    \begin{equation}
        \|p-q\|_1\ge \frac{1}{\bar{x}}|\mathbb{E}_p[x_1]-\mathbb{E}_q[x_1]|-\frac{2}{\bar{x}^2}\max\left(\mathbb{E}_p[x_1^2],\mathbb{E}_q[x_1^2]\right)\,.
    \end{equation}
    By optimizing over $\bar{x}$, we have that
    \begin{equation}\label{eq_class}
        \|p-q\|_1\ge \frac{\left|\mathbb{E}_p\!\left[x_1\right]-\mathbb{E}_q\!\left[x_1\right]\right|^2}{8\max\left(\mathbb{E}_p\!\left[x_1^2\right],\mathbb{E}_q\!\left[x_1^2\right]\right)}\,.
    \end{equation}
    In particular, the following less tight bound holds:
    \begin{equation}
        \|p-q\|_1\ge \frac{\left|\mathbb{E}_p\!\left[x_1\right]-\mathbb{E}_q\!\left[x_1\right]\right|^2}{8\max\left(\mathbb{E}_p\!\left[\|\bm{x}\|^2\right],\mathbb{E}_q\!\left[\|\bm{x}\|^2\right]\right)}\,.
    \end{equation}
    Let us define
    \begin{equation}
        \bm{w}\coloneqq \frac{\mathbb{E}_p\!\left[\bm{x}\right]-\mathbb{E}_q\!\left[\bm{x}\right]}{\left\|\mathbb{E}_p\!\left[\bm{x}\right]-\mathbb{E}_q\!\left[\bm{x}\right]\right\|}\,,
    \end{equation}
    so that
    \begin{equation}
        \left\|\mathbb{E}_p\!\left[\bm{x}\right]-\mathbb{E}_q\!\left[\bm{x}\right]\right\|=\mathbb{E}_p\!\left[\bm{w}^\intercal\bm{x}\right]-\mathbb{E}_q[\bm{w}^\intercal\bm{x}]\,.
    \end{equation}
    Let $O$ be an orthogonal matrix with first column equal to $\bm{w}$. Let $\tilde{p}$
and $\tilde{q}$ defined as
\begin{align}
    \tilde{p}(\bm{x})&\coloneqq p(O\bm{x})\,,\\
    \tilde{q}(\bm{x})&\coloneqq q(O\bm{x})\,.
\end{align}
Note that
\begin{align}
    \mathbb{E}_{\tilde{p}}\!\left[x_1\right]&=\mathbb{E}_p\!\left[\bm{w}^\intercal\bm{x}\right]\,,\\
    \mathbb{E}_{\tilde{q}}\!\left[x_1\right]&=\mathbb{E}_q\!\left[\bm{w}^\intercal\bm{x}\right]\,,
\end{align}
and thus we have that
\begin{align}
    \mathbb{E}_{\tilde{p}}\!\left[x_1\right]-\mathbb{E}_{\tilde{q}}\!\left[x_1\right]&=\left\|\mathbb{E}_p\!\left[\bm{x}\right]-\mathbb{E}_q\!\left[\bm{x}\right]\right\|\,.
\end{align}
Moreover, since $O$ is orthogonal, we have also that
\begin{align}
\mathbb{E}_p\!\left[\|\bm{x}\|^2\right]&=\mathbb{E}_{\tilde{p}}\!\left[\|\bm{x}\|^2\right]\,,\\
\mathbb{E}_q\!\left[\|\bm{x}\|^2\right]&=\mathbb{E}_{\tilde{q}}\!\left[\|\bm{x}\|^2\right]\,.
\end{align}
Consequently, we conclude that
\begin{align}
        \|p-q\|_1&=\|\tilde{p}-\tilde{q}\|_1\nonumber\\
        &\ge  \frac{\left|\mathbb{E}_{\tilde{p}}\!\left[x_1\right]-\mathbb{E}_{\tilde{q}}\!\left[x_1\right]\right|^2}{8\max\left(\mathbb{E}_{\tilde{p}}\!\left[\|\bm{x}\|^2\right],\mathbb{E}_{\tilde{q}}\!\left[\|\bm{x}\|^2\right]\right)}\nonumber\\
        &  = \frac{\left\|\mathbb{E}_p\!\left[\bm{x}\right]-\mathbb{E}_q\!\left[\bm{x}\right]\right\|^2}{8\max\left(\mathbb{E}_p\!\left[\|\bm{x}\|^2\right],\mathbb{E}_q\!\left[\|\bm{x}\|^2\right]\right)}\,.
    \end{align}
\end{proof}

We also have the following.

\begin{lem}(Lower bounds on the total variation distance with the difference of the second moment matrices)\label{classical_lower_mean2}
    Let $p:\mathbb{R}^n\mapsto\mathbb{R}$ and  $q:\mathbb{R}^n\mapsto\mathbb{R}$ be probability distributions over $\mathbb{R}^n$. Then, it holds that
    \begin{equation}
        \|p-q\|_1\ge \frac{\left\|\mathbb{E}_p\!\left[\bm{x}\bm{x}^\intercal\right]-\mathbb{E}_q\!\left[\bm{x}\bm{x}^\intercal\right]\right\|_\infty^2}{4\max\left(\mathbb{E}_p\!\left[\|\bm{x}\|^4\right],\mathbb{E}_q\!\left[\|\bm{x}\|^4\right]\right)}\,,
    \end{equation}
    where $\|\bm{x}\|\coloneqq \sqrt{\sum_{i=1}^nx_i^2}$.
\end{lem}
\begin{proof}
    The proof is completely analogous to the one of  \cref{classical_lower_mean} and of  \cref{lemma_sec_q}. Namely, analogously to Eq.~\eqref{eq_class}, we find
    \begin{equation}
        \|p-q\|_1\ge \frac{\left|\mathbb{E}_p\!\left[x_1^2\right]-\mathbb{E}_q\!\left[x_1^2\right]\right|^2}{4\max\left(\mathbb{E}_p\!\left[x_1^4\right],\mathbb{E}_q\!\left[x_1^4\right]\right)}\,,
    \end{equation}
    and thus also that
    \begin{equation}
            \|p-q\|_1\ge \frac{\left|\mathbb{E}_p\!\left[x_1^2\right]-\mathbb{E}_q\!\left[x_1^2\right]\right|^2}{4\max\left(\mathbb{E}_p\!\left[\|\bm{x}\|^4\right],\mathbb{E}_q\!\left[\|\bm{x}\|^4\right]\right)}\,.
    \end{equation}
    Moreover, let $\bm{w}$ with $\bm{w}^\intercal \bm{w}=1$ such that
\begin{equation}
   \left\|\mathbb{E}_p\!\left[\bm{x}\bm{x}^\intercal\right]-\mathbb{E}_q\!\left[\bm{x}\bm{x}^\intercal\right]\right\|_\infty= |\bm{w}^\intercal\left( \mathbb{E}_p\!\left[\bm{x}\bm{x}^\intercal\right]-\mathbb{E}_q\!\left[\bm{x}\bm{x}^\intercal\right]\right)\bm{w}|\,,
\end{equation}
so that
\begin{equation}
    \left\|\mathbb{E}_p\!\left[\bm{x}\bm{x}^\intercal\right]-\mathbb{E}_q\!\left[\bm{x}\bm{x}^\intercal\right]\right\|_\infty=|\mathbb{E}_p[(\bm{w}^\intercal\bm{x})^2]-\mathbb{E}_q[(\bm{w}^\intercal\bm{x})^2]|.
\end{equation}
Let $O$ be an orthogonal matrix with first column equal to $\bm{w}$. Let $\tilde{p}$
and $\tilde{q}$ defined as
\begin{align}
    \tilde{p}(\bm{x})&\coloneqq p(O\bm{x})\,,\\
    \tilde{q}(\bm{x})&\coloneqq q(O\bm{x})\,.
\end{align}
Note that
\begin{align}
    \mathbb{E}_{\tilde{p}}\!\left[x_1^2\right]&=\mathbb{E}_p\!\left[(\bm{w}^\intercal\bm{x})^2\right]\,,\\
    \mathbb{E}_{\tilde{q}}\!\left[x_1^2\right]&=\mathbb{E}_q\!\left[(\bm{w}^\intercal\bm{x})^2\right]\,,
\end{align}
and thus we have that
\begin{align}
    \left|\mathbb{E}_{\tilde{p}}[x_1^2]-\mathbb{E}_{\tilde{q}}[x_1^2]\right|&=\left\|\mathbb{E}_p\!\left[\bm{x}\bm{x}^\intercal\right]-\mathbb{E}_q\!\left[\bm{x}\bm{x}^\intercal\right]\right\|_\infty\,.
\end{align}
Moreover, since $O$ is orthogonal, we have also that
\begin{align}
\mathbb{E}_p\!\left[\|\bm{x}\|^4\right]&=\mathbb{E}_{\tilde{p}}\!\left[\|\bm{x}\|^4\right]\,,\\
\mathbb{E}_q\!\left[\|\bm{x}\|^4\right]&=\mathbb{E}_{\tilde{q}}\!\left[\|\bm{x}\|^4\right]\,.
\end{align}
Consequently, we conclude that
\begin{align}
    \|p-q\|_1&=\|\tilde{p}-\tilde{q}\|_1\nonumber\\
    &\ge \frac{\left|\mathbb{E}_{\tilde{p}}\!\left[x_1^2\right]-\mathbb{E}_{\tilde{q}}\!\left[x_1^2\right]\right|^2}{4\max\left(\mathbb{E}_{\tilde{p}}\!\left[\|\bm{x}\|^4\right],\mathbb{E}_{\tilde{q}}\!\left[\|\bm{x}\|^4\right]\right)}\nonumber\\
    &=\frac{\left\|\mathbb{E}_p\!\left[\bm{x}\bm{x}^\intercal\right]-\mathbb{E}_q\!\left[\bm{x}\bm{x}^\intercal\right]\right\|_\infty^2}{4\max\left(\mathbb{E}_{p}\!\left[\|\bm{x}\|^4\right],\mathbb{E}_{q}\!\left[\|\bm{x}\|^4\right]\right)}\,.
\end{align}
\end{proof}

Finally, we establish the following lower bound on the total variation distance in terms of the covariance matrices.

\begin{lem}[(Lower bounds on the total variation distance in terms of the difference of the covariance matrices)]
        Let $p:\mathbb{R}^n\mapsto\mathbb{R}$ be a probability distribution over $\mathbb{R}^n$ and let $V(p)$ be its covariance matrix defined by
        \begin{equation}
            V(p)\coloneqq \mathbb{E}_p\!\left[(\bm{x}-\mathbb{E}_p[\bm{x}])(\bm{x}-\mathbb{E}_p[\bm{x}])^\intercal \right]\,.
        \end{equation}
        Let $p$ and $q$ be probability distributions over $\mathbb{R}^n$. Then, it holds that
    \begin{equation}
        \|p-q\|_1\ge \frac{\left\|V(p)-V(q)\right\|_\infty^2}{72\max\left(\mathbb{E}_p\!\left[\|\bm{x}\|^4\right],\mathbb{E}_q\!\left[\|\bm{x}\|^4\right]\right)}\,,
    \end{equation}
    where $\|\bm{x}\|\coloneqq \sqrt{\sum_{i=1}^nx_i^2}$.
\end{lem}

\begin{proof}
    The proof is analogous to the one of \cref{thm_pert_V}. Namely, by using that
    \begin{equation}
        V(\rho)=\mathbb{E}_p[\bm{x}\bm{x}^\intercal]-\mathbb{E}_p[\bm{x}]\mathbb{E}_p[\bm{x}]^\intercal
    \end{equation}
    we have that
    \begin{align}
        \|p-q\|_1&\ge \frac{\left\|\mathbb{E}_p\!\left[\bm{x}\bm{x}^\intercal\right]-\mathbb{E}_q\!\left[\bm{x}\bm{x}^\intercal\right]\right\|_\infty^2}{4\max\left(\mathbb{E}_p\!\left[\|\bm{x}\|^4\right],\mathbb{E}_q\!\left[\|\bm{x}\|^4\right]\right)}\nonumber\\
        &=\frac{\left\|V(p)-V(q)+\mathbb{E}_p[\bm{x}]\mathbb{E}_p[\bm{x}]^\intercal-\mathbb{E}_q[\bm{x}]\mathbb{E}_q[\bm{x}]^\intercal\right\|_\infty^2}{4\max\left(\mathbb{E}_p\!\left[\|\bm{x}\|^4\right],\mathbb{E}_q\!\left[\|\bm{x}\|^4\right]\right)}\nonumber\\
        &\ge  \frac{\frac{1}{2}\left\|V(p)-V(q)\right\|_\infty^2-\left\|\mathbb{E}_p[\bm{x}]\mathbb{E}_p[\bm{x}]^\intercal-\mathbb{E}_q[\bm{x}]\mathbb{E}_q[\bm{x}]^\intercal\right\|_\infty^2}{4\max\left(\mathbb{E}_p\!\left[\|\bm{x}\|^4\right],\mathbb{E}_q\!\left[\|\bm{x}\|^4\right]\right)}\,.
    \end{align}
    where in the last line we exploited that $\|A-B\|_\infty^2\ge \frac{\|A\|_\infty^2}{2}-\|B\|_\infty^2$. Moreover, we have that
    \begin{align}
        \left\|\mathbb{E}_p[\bm{x}]\mathbb{E}_p[\bm{x}]^\intercal-\mathbb{E}_q[\bm{x}]\mathbb{E}_q[\bm{x}]^\intercal\right\|_\infty&\le \left(\|\mathbb{E}_p[\bm{x}]\|+\|\mathbb{E}_q[\bm{x}]\|\right)\|\mathbb{E}_p[\bm{x}]-\mathbb{E}_q[\bm{x}]\|\nonumber\\
        &\le 2\sqrt{\max\left(\mathbb{E}_p[\|\bm{x}\|^2],\mathbb{E}_q[\|\bm{x}\|^2]\right)}\|\mathbb{E}_p[\bm{x}]-\mathbb{E}_q[\bm{x}]\|\nonumber\\
        &\le 2\sqrt{8}\max\left(\mathbb{E}_p[\|\bm{x}\|^2],\mathbb{E}_q[\|\bm{x}\|^2]\right)\sqrt{\|p-q\|_1}\nonumber\\
        &\le2\sqrt{8}\sqrt{\max\left(\mathbb{E}_p[\|\bm{x}\|^4],\mathbb{E}_q[\|\bm{x}\|^4]\right)}\sqrt{\|p-q\|_1{}} 
    \end{align}
    where in the third line we exploited  \cref{classical_lower_mean} and the concavity of the square root. Consequently, we obtain that
    \begin{align}
        \|p-q\|_1&\ge \frac{\left\|V(p)-V(q)\right\|_\infty^2}{8\max\left(\mathbb{E}_p\!\left[\|\bm{x}\|^4\right],\mathbb{E}_q\!\left[\|\bm{x}\|^4\right]\right)}-8\|p-q\|_1\,.
    \end{align}
    Rearranging the terms, we conclude the proof.
\end{proof}


\medskip \noindent 
\textbf{\em Acknowledgements.\,}---
F.A.M., S.F.E.O., and U.C.\ are grateful to \'A.~Capel, N.~Datta, L.~Lami and A.~Winter for organising the BIRS-IMAG workshop \textit{`Towards infinite dimension and beyond in quantum information'} (May 2025, Granada, Spain), where preliminary ideas were discussed. We warmly thank G.~Adesso for helpful comments. F.A.M.~acknowledges precious discussions with L.~Lami, especially regarding the block decomposition of Gaussian unitary operations, the lower bounds on the trace distance between non-Gaussian states, and the irreversibility of the resource theory of non-Gaussianity. F.A.M.\ and S.F.E.O.\ acknowledge useful discussions with L.~Bittel and Ryuji Takagi. S.F.E.O. thanks L.~Leone for useful discussion on the strong monotonicity and its connection with approximate measures. 
U.C.\ acknowledges inspiring discussions with M.~Walschaers, O.~Hahn, A.~Ferraro, and G.~Ferrini. F.A.M.\ thanks the University of Amsterdam and QuSoft for their hospitality.
U.C.\ and V.U.\ acknowledge funding from the European Union’s Horizon Europe Framework Programme (EIC Pathfinder Challenge project Veriqub) under Grant Agreement No.~101114899. S.F.E.O.\ acknowledges support from the PNRR MUR project PE0000023-NQSTI.

\bibliography{biblio}
\bibliographystyle{unsrtnat}

\end{document}